\documentclass[11pt]{article}

\usepackage[letterpaper, top=.97in, bottom=.98in, left=0.98in, right=.98in]{geometry}

\usepackage{titling}
\thanksmarkseries{arabic}
\newcommand*\samethanks[1][\value{footnote}]{\footnotemark[#1]}

\RequirePackage[T1]{fontenc}
\RequirePackage[utf8]{inputenc}
\usepackage[greek,english]{babel}
\usepackage{alphabeta}
\usepackage{amssymb}
\usepackage{microtype}
\usepackage{amsmath}
\usepackage{amsfonts}
\usepackage{mathtools}
\usepackage[mathscr]{euscript}
\usepackage{amsthm}
\usepackage{nicefrac}
\usepackage{dsfont}

\usepackage[pdftex,
backref=page,
plainpages = false,
pdfpagelabels,
hyperfootnotes=true,
pdfpagemode=FullScreen,
bookmarks=true,
bookmarksopen = true,
bookmarksnumbered = true,
breaklinks = true,
hyperfigures,
pagebackref,
urlcolor = magenta,
urlcolor = MidnightBlack,
anchorcolor = green,
hyperindex = true,
colorlinks = true,
linkcolor = black!60!blue,
citecolor = black!30!green
]{hyperref}

\usepackage[svgnames]{xcolor}
\usepackage{bm,ifthen,wrapfig,url,cite,rotate,aliascnt,algorithmic,algorithm,enumerate,enumitem,todonotes,fancybox,tabularx,colortbl,xspace,graphicx,caption}
\usepackage{stmaryrd,lmodern}
\usepackage[nameinlink]{cleveref}

\newcommand{\tw}{{\sf tw}}

\newcommand{\p}{{\sf p}}

\newcommand{\N}{\mathbb{N}}

\newcommand{\MSO}{{\sf MSO}\xspace}

\newcommand{\hh}{\end{document}}

\newcommand{\bd}{{\sf bd}}

\newcommand{\inter}{{\sf int}}
\newcommand{\remove}[1]{}

\newcommand{\bigmid}{\;\big|\;}
\newcommand{\cupall}{\pmb{\pmb{\bigcup}}}

\newcommand{\FPT}{{\sf FPT}\xspace}

\DeclareFontEncoding{LS1}{}{}
\DeclareFontSubstitution{LS1}{stix}{m}{n}
\DeclareSymbolFont{symbolsstix}{LS1}{stixscr}{m}{n}
\SetSymbolFont{symbolsstix}{bold}{LS1}{stixscr}{b}{n}
\DeclareMathSymbol{\mathvisiblespace}{0}{symbolsstix}{"B6}
\newcommand{\mathspace}{\mathvisiblespace}

\RequirePackage{stmaryrd}
\usepackage{textcomp}
\DeclareUnicodeCharacter{2286}{\subseteq}
\DeclareUnicodeCharacter{2192}{\ifmmode\to\else\textrightarrow\fi}
\DeclareUnicodeCharacter{2203}{\ensuremath\exists}
\DeclareUnicodeCharacter{183}{\cdot}
\DeclareUnicodeCharacter{2200}{\forall}
\DeclareUnicodeCharacter{2264}{\leq}
\DeclareUnicodeCharacter{2265}{\geq}
\DeclareUnicodeCharacter{8614}{\mathbin{\mapsto}}
\DeclareUnicodeCharacter{8656}{\Leftarrow}
\DeclareUnicodeCharacter{8657}{\Uparrow}
\DeclareUnicodeCharacter{8658}{\Rightarrow}
\DeclareUnicodeCharacter{8659}{\Downarrow}
\DeclareUnicodeCharacter{8669}{\rightsquigarrow}
\newcommand{\eqdef}{\stackrel{{\scriptsize\rm def}}{=}}
\DeclareUnicodeCharacter{8797}{\eqdef}
\DeclareUnicodeCharacter{8870}{\vdash}
\DeclareUnicodeCharacter{8873}{\Vdash}
\DeclareUnicodeCharacter{22A7}{\models}
\DeclareUnicodeCharacter{9121}{\lceil}
\DeclareUnicodeCharacter{9123}{\lfloor}
\DeclareUnicodeCharacter{9124}{\rceil}
\DeclareUnicodeCharacter{2208}{\in}
\DeclareUnicodeCharacter{9126}{\rfloor}
\DeclareUnicodeCharacter{9655}{\triangleright}
\DeclareUnicodeCharacter{9665}{\triangleleft}
\DeclareUnicodeCharacter{9671}{\diamond}
\DeclareUnicodeCharacter{9675}{\circ}
\DeclareUnicodeCharacter{10178}{\bot}
\DeclareUnicodeCharacter{10229}{\longleftarrow}
\DeclareUnicodeCharacter{10230}{\longrightarrow}
\DeclareUnicodeCharacter{10231}{\longleftrightarrow}
\DeclareUnicodeCharacter{10232}{\Longleftarrow}
\DeclareUnicodeCharacter{10233}{\Longrightarrow}
\DeclareUnicodeCharacter{10234}{\Longleftrightarrow}
\DeclareUnicodeCharacter{10236}{\longmapsto}
\DeclareUnicodeCharacter{10238}{\Longmapsto} 
\DeclareUnicodeCharacter{10503}{\Mapsto}    
\DeclareUnicodeCharacter{10971}{\mathrel{\not\hspace{-0.2em}\cap}}
\DeclareUnicodeCharacter{65294}{\ldotp}
\DeclareUnicodeCharacter{65372}{\mid}

\usepackage{color}
\definecolor{MidnightBlack}{rgb}{0.1,0.1,.32}
\definecolor{MidnightBlue}{rgb}{0.1,0.1,0.44}
\definecolor{Black}{rgb}{0,0, 0}
\definecolor{Blue}{rgb}{0, 0 ,1}
\definecolor{Red}{rgb}{1, 0 ,0}
\definecolor{White}{rgb}{1, 1, 1}
\definecolor{Grey}{rgb}{.6, .6, .6}
\definecolor{Mygreen}{rgb}{.0, .7, .0}
\definecolor{Yellow}{rgb}{.55,.55,0}
\definecolor{darkmagenta}{rgb}{0.30, 0.0, 0.30}
\definecolor{darkorange}{rgb}{1.0, 0.55, 0.0}
\definecolor{ao}{rgb}{1.0, 0.13, 0.32}
\definecolor{brandeisblue}{rgb}{0.0, 0.44, 1.0}
\definecolor{resu}{HTML}{e5e7fe}

\usepackage{tikz}

\usetikzlibrary{arrows.meta}

\makeatletter

\pgfdeclarearrow{
  name = ipe _linear,
  defaults = {
    length = +1bp,
    width  = +.666bp,
    line width = +0pt 1,
  },
  setup code = {
    \pgfarrowssetbackend{0pt}
    \pgfarrowssetvisualbackend{
      \pgfarrowlength\advance\pgf@x by-.5\pgfarrowlinewidth}
    \pgfarrowssetlineend{\pgfarrowlength}
    \ifpgfarrowreversed
      \pgfarrowssetlineend{\pgfarrowlength\advance\pgf@x by-.5\pgfarrowlinewidth}
    \fi
    \pgfarrowssettipend{\pgfarrowlength}
    \pgfarrowshullpoint{\pgfarrowlength}{0pt}
    \pgfarrowsupperhullpoint{0pt}{.5\pgfarrowwidth}
    \pgfarrowssavethe\pgfarrowlinewidth
    \pgfarrowssavethe\pgfarrowlength
    \pgfarrowssavethe\pgfarrowwidth
  },
  drawing code = {
    \pgfsetdash{}{+0pt}
    \ifdim\pgfarrowlinewidth=\pgflinewidth\else\pgfsetlinewidth{+\pgfarrowlinewidth}\fi
    \pgfpathmoveto{\pgfqpoint{0pt}{.5\pgfarrowwidth}}
    \pgfpathlineto{\pgfqpoint{\pgfarrowlength}{0pt}}
    \pgfpathlineto{\pgfqpoint{0pt}{-.5\pgfarrowwidth}}
    \pgfusepathqstroke
  },
  parameters = {
    \the\pgfarrowlinewidth,%
    \the\pgfarrowlength,%
    \the\pgfarrowwidth,%
  },
}

\pgfdeclarearrow{
  name = ipe _pointed,
  defaults = {
    length = +1bp,
    width  = +.666bp,
    inset  = +.2bp,
    line width = +0pt 1,
  },
  setup code = {
    \pgfarrowssetbackend{0pt}
    \pgfarrowssetvisualbackend{\pgfarrowinset}
    \pgfarrowssetlineend{\pgfarrowinset}
    \ifpgfarrowreversed
      \pgfarrowssetlineend{\pgfarrowlength}
    \fi
    \pgfarrowssettipend{\pgfarrowlength}
    \pgfarrowshullpoint{\pgfarrowlength}{0pt}
    \pgfarrowsupperhullpoint{0pt}{.5\pgfarrowwidth}
    \pgfarrowshullpoint{\pgfarrowinset}{0pt}
    \pgfarrowssavethe\pgfarrowinset
    \pgfarrowssavethe\pgfarrowlinewidth
    \pgfarrowssavethe\pgfarrowlength
    \pgfarrowssavethe\pgfarrowwidth
  },
  drawing code = {
    \pgfsetdash{}{+0pt}
    \ifdim\pgfarrowlinewidth=\pgflinewidth\else\pgfsetlinewidth{+\pgfarrowlinewidth}\fi
    \pgfpathmoveto{\pgfqpoint{\pgfarrowlength}{0pt}}
    \pgfpathlineto{\pgfqpoint{0pt}{.5\pgfarrowwidth}}
    \pgfpathlineto{\pgfqpoint{\pgfarrowinset}{0pt}}
    \pgfpathlineto{\pgfqpoint{0pt}{-.5\pgfarrowwidth}}
    \pgfpathclose
    \ifpgfarrowopen
      \pgfusepathqstroke
    \else
      \ifdim\pgfarrowlinewidth>0pt\pgfusepathqfillstroke\else\pgfusepathqfill\fi
    \fi
  },
  parameters = {
    \the\pgfarrowlinewidth,%
    \the\pgfarrowlength,%
    \the\pgfarrowwidth,%
    \the\pgfarrowinset,%
    \ifpgfarrowopen o\fi%
  },
}

\newdimen\ipeminipagewidth

\tikzstyle{ipe import} = [
  x=1bp, y=1bp,
  ipe node stretch/.store in=\ipenodestretch,
  ipe stretch normal/.style={ipe node stretch=1},
  ipe stretch normal,
  ipe node/.style={
    anchor=base west, inner sep=0, outer sep=0, scale=\ipenodestretch
  },
  ipe mark scale/.store in=\ipemarkscale,
  ipe mark tiny/.style={ipe mark scale=1.1},
  ipe mark small/.style={ipe mark scale=2},
  ipe mark normal/.style={ipe mark scale=3},
  ipe mark large/.style={ipe mark scale=5},
  ipe mark normal, 
  ipe circle/.pic={
    \draw[line width=0.2*\ipemarkscale]
      (0,0) circle[radius=0.5*\ipemarkscale];
    \coordinate () at (0,0);
  },
  ipe disk/.pic={
    \fill (0,0) circle[radius=0.6*\ipemarkscale];
    \coordinate () at (0,0);
  },
  ipe fdisk/.pic={
    \filldraw[line width=0.2*\ipemarkscale]
      (0,0) circle[radius=0.5*\ipemarkscale];
    \coordinate () at (0,0);
  },
  ipe box/.pic={
    \draw[line width=0.2*\ipemarkscale, line join=miter]
      (-.5*\ipemarkscale,-.5*\ipemarkscale) rectangle
      ( .5*\ipemarkscale, .5*\ipemarkscale);
    \coordinate () at (0,0);
  },
  ipe square/.pic={
    \fill
      (-.6*\ipemarkscale,-.6*\ipemarkscale) rectangle
      ( .6*\ipemarkscale, .6*\ipemarkscale);
    \coordinate () at (0,0);
  },
  ipe fsquare/.pic={
    \filldraw[line width=0.2*\ipemarkscale, line join=miter]
      (-.5*\ipemarkscale,-.5*\ipemarkscale) rectangle
      ( .5*\ipemarkscale, .5*\ipemarkscale);
    \coordinate () at (0,0);
  },
  ipe cross/.pic={
    \draw[line width=0.2*\ipemarkscale, line cap=butt]
      (-.5*\ipemarkscale,-.5*\ipemarkscale) --
      ( .5*\ipemarkscale, .5*\ipemarkscale)
      (-.5*\ipemarkscale, .5*\ipemarkscale) --
      ( .5*\ipemarkscale,-.5*\ipemarkscale);
    \coordinate () at (0,0);
  },
  /pgf/arrow keys/.cd,
  ipe arrow normal/.style={scale=1},
  ipe arrow tiny/.style={scale=.4},
  ipe arrow small/.style={scale=.7},
  ipe arrow large/.style={scale=1.4},
  ipe arrow normal,
  /tikz/.cd,
  ipe arrows/.style={
    ipe normal/.tip={
      ipe _pointed[length=1bp, width=.666bp, inset=0bp,
                   quick, ipe arrow normal]},
    ipe pointed/.tip={
      ipe _pointed[length=1bp, width=.666bp, inset=0.2bp,
                   quick, ipe arrow normal]},
    ipe linear/.tip={
      ipe _linear[length = 1bp, width=.666bp,
                  ipe arrow normal, quick]},
    ipe fnormal/.tip={ipe normal[fill=white]},
    ipe fpointed/.tip={ipe pointed[fill=white]},
    ipe double/.tip={ipe normal[] ipe normal},
    ipe fdouble/.tip={ipe fnormal[] ipe fnormal},
    ipe arc/.tip={ipe normal},
    ipe farc/.tip={ipe fnormal},
    ipe ptarc/.tip={ipe pointed},
    ipe fptarc/.tip={ipe fpointed},
  },
  ipe arrows, 
]

\tikzset{
  rgb color/.code args={#1=#2}{%
    \definecolor{tempcolor-#1}{rgb}{#2}%
    \tikzset{#1=tempcolor-#1}%
  },
}

\makeatother

\usetikzlibrary{patterns}
\usetikzlibrary{calc,3d}
\usetikzlibrary{intersections,decorations.pathreplacing,fit}
\usetikzlibrary{shapes.geometric,hobby,calc}
\usetikzlibrary{decorations.pathmorphing}
\usetikzlibrary{decorations.text}
\usetikzlibrary{shapes.misc}
\usetikzlibrary{decorations,shapes}

\newcounter{func}
\newcommand{\newfun}[1]{f_{\refstepcounter{func}\label{#1}\thefunc}}
\newcommand{\funref}[1]{\hyperref[#1]{f_{\ref*{#1}}}} 
\newcounter{con}

\newcommand{\conref}[1]{\hyperref[#1]{c_{\ref*{#1}}}} 

\newcommand{\mynewtheorem}[2]{
	\newaliascnt{#1}{dummy}
	\newtheorem{#1}[#1]{#2}
	\aliascntresetthe{#1}
}

\theoremstyle{plain}
\mynewtheorem{result}{Theorem}
\mynewtheorem{theorem}{Theorem}
\mynewtheorem{proposition}{Proposition}
\mynewtheorem{corollary}{Corollary}
\mynewtheorem{lemma}{Lemma}
\theoremstyle{definition}
\mynewtheorem{definition}{Definition}
\mynewtheorem{importantdef}{Definition}
\theoremstyle{remark}
\mynewtheorem{sublemma}{Sublemma}
\mynewtheorem{claim}{Claim}
\mynewtheorem{remark}{Remark}
\mynewtheorem{fact}{Fact}
\mynewtheorem{conjecture}{Conjecture}
\mynewtheorem{question}{Question}
\mynewtheorem{observation}{Observation}
\mynewtheorem{problem}{Problem}
\mynewtheorem{note}{Note}

\addto\extrasenglish{

}

\Crefname{subsection}{Subsection}{Subsections}
\Crefname{result}{Theorem}{Theorems}
\Crefname{importantdef}{Definition}{Definitions}

\usepackage[many]{tcolorbox}

 \tcolorboxenvironment{result}{
    colframe=blue, sharp corners=west,
    colback=resu!30!white,
	rightrule=0pt,
	toprule=0pt,
	bottomrule=0pt,
	top=2pt,
	right=2pt,
	bottom=2pt,
	left =2pt
 }

 \tcolorboxenvironment{importantdef}{
    colframe=orange, sharp corners=west,
    colback=orange!05!white,
	rightrule=0pt,
	toprule=0pt,
	bottomrule=0pt,
	top=2pt,
	right=2pt,
	bottom=2pt,
	left =2pt
 }

\tikzstyle{ipe stylesheet} = [
  ipe import,
  even odd rule,
  line join=round,
  line cap=butt,
  ipe pen normal/.style={line width=0.4},
  ipe pen heavier/.style={line width=0.8},
  ipe pen fat/.style={line width=1.2},
  ipe pen ultrafat/.style={line width=2},
  ipe pen normal,
  ipe mark normal/.style={ipe mark scale=3},
  ipe mark large/.style={ipe mark scale=5},
  ipe mark small/.style={ipe mark scale=2},
  ipe mark tiny/.style={ipe mark scale=1.1},
  ipe mark normal,
  /pgf/arrow keys/.cd,
  ipe arrow normal/.style={scale=7},
  ipe arrow large/.style={scale=10},
  ipe arrow small/.style={scale=5},
  ipe arrow tiny/.style={scale=3},
  ipe arrow normal,
  /tikz/.cd,
  ipe arrows, 
  <->/.tip = ipe normal,
  ipe dash normal/.style={dash pattern=},
  ipe dash dotted/.style={dash pattern=on 1bp off 3bp},
  ipe dash dashed/.style={dash pattern=on 4bp off 4bp},
  ipe dash dash dotted/.style={dash pattern=on 4bp off 2bp on 1bp off 2bp},
  ipe dash dash dot dotted/.style={dash pattern=on 4bp off 2bp on 1bp off 2bp on 1bp off 2bp},
  ipe dash normal,
  ipe node/.append style={font=\normalsize},
  ipe stretch normal/.style={ipe node stretch=1},
  ipe stretch normal,
  ipe opacity 10/.style={opacity=0.1},
  ipe opacity 30/.style={opacity=0.3},
  ipe opacity 50/.style={opacity=0.5},
  ipe opacity 75/.style={opacity=0.75},
  ipe opacity opaque/.style={opacity=1},
  ipe opacity opaque,
]
\definecolor{red}{rgb}{1,0,0}
\definecolor{blue}{rgb}{0,0,1}
\definecolor{green}{rgb}{0,1,0}
\definecolor{yellow}{rgb}{1,1,0}
\definecolor{orange}{rgb}{1,0.647,0}
\definecolor{gold}{rgb}{1,0.843,0}
\definecolor{purple}{rgb}{0.627,0.125,0.941}
\definecolor{gray}{rgb}{0.745,0.745,0.745}
\definecolor{brown}{rgb}{0.647,0.165,0.165}
\definecolor{navy}{rgb}{0,0,0.502}
\definecolor{pink}{rgb}{1,0.753,0.796}
\definecolor{seagreen}{rgb}{0.18,0.545,0.341}
\definecolor{turquoise}{rgb}{0.251,0.878,0.816}
\definecolor{violet}{rgb}{0.933,0.51,0.933}
\definecolor{darkblue}{rgb}{0,0,0.545}
\definecolor{darkcyan}{rgb}{0,0.545,0.545}
\definecolor{darkgray}{rgb}{0.663,0.663,0.663}
\definecolor{darkgreen}{rgb}{0,0.392,0}
\definecolor{darkmagenta}{rgb}{0.545,0,0.545}
\definecolor{darkorange}{rgb}{1,0.549,0}
\definecolor{darkred}{rgb}{0.545,0,0}
\definecolor{lightblue}{rgb}{0.678,0.847,0.902}
\definecolor{lightcyan}{rgb}{0.878,1,1}
\definecolor{lightgray}{rgb}{0.827,0.827,0.827}
\definecolor{lightgreen}{rgb}{0.565,0.933,0.565}
\definecolor{lightyellow}{rgb}{1,1,0.878}
\definecolor{black}{rgb}{0,0,0}
\definecolor{white}{rgb}{1,1,1}

\usepackage{titling}

\usepackage{textpos}

\title{Parameterizing the quantification of CMSO: model checking on minor-closed graph classes\thanks{All authors were supported by  the French-German Collaboration ANR/DFG Project UTMA (ANR-20-CE92-0027). The first author was also supported by  ANR project ELIT (ANR-20-CE48-0008). 
The second author was also supported by the project BOBR that is funded from the European Research Council (ERC) under the European Union's Horizon 2020 research and innovation programme with grant agreement No. 948057. The third author was also supported by the Franco-Norwegian project PHC AURORA 2024 (Projet n° 51260WL).
Emails: \texttt{ignasi.sau@lirmm.fr}, \texttt{giannos.stamoulis@mimuw.edu.pl}, \texttt{sedthilk@thilikos.info}.}}

\author{\bigskip
Ignasi Sau%
\thanks{LIRMM, Univ Montpellier, CNRS, Montpellier, France.} 
\and
Giannos Stamoulis\thanks{University of Warsaw, Poland. Part of the research work for this paper was conducted when G.S. was affiliated with LIRMM, Univ Montpellier, CNRS, Montpellier, France.}
\and
Dimitrios M. Thilikos\samethanks[2]}
\date{}

\begin{document}

\maketitle

\thispagestyle{empty}

\begin{abstract}
\noindent Given a graph $G$ and a vertex set $X$,
the \emph{annotated treewidth} $\tw(G,X)$ of $X$ in $G$ is  the maximum treewidth of an $X$-rooted minor of $G$, i.e., a minor $H$
where the model of each vertex of $H$ contains some vertex of $X$.
That way, $\tw(G,X)$ can be seen as a measure of the contribution of $X$ to the tree-decomposability  of $G$.
We introduce the logic \textsf{CMSO/tw} as the fragment of  monadic second-order logic on graphs obtained by restricting set
quantification to sets of bounded annotated treewidth.
We prove the following Algorithmic Meta-Theorem (AMT): \emph{for every non-trivial
minor-closed graph class,  model checking for} \textsf{CMSO/tw} \emph{formulas can be done in quadratic
time}.
Our proof works for the more general \textsf{CMSO/tw}+\textsf{dp} logic,
that is  \textsf{CMSO/tw} enhanced by disjoint-path predicates.
Our AMT can be seen as an extension of Courcelle's theorem to minor-closed
graph classes  where the bounded-treewidth condition in the input graph is
replaced by the bounded-treewidth quantification in the formulas.
Our results  yield, as special cases, all known AMTs whose combinatorial restriction is non-trivial minor-closedness.
\bigskip

\end{abstract}


\begin{textblock}{20}(12.2, 3.2)
   \includegraphics[width=30px]{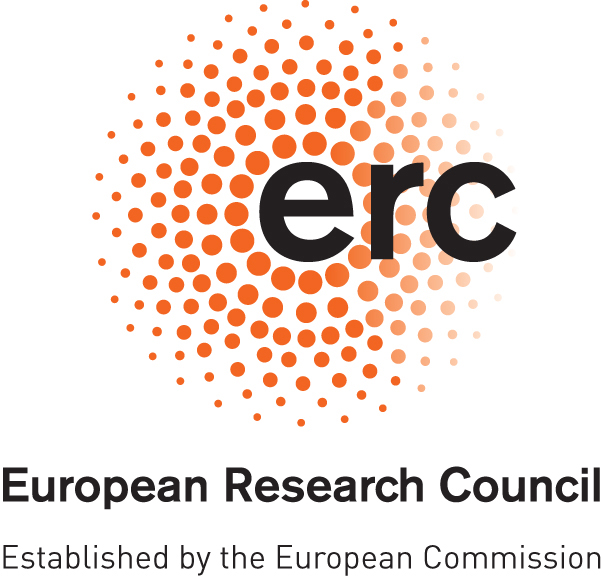}%
\end{textblock}
\begin{textblock}{20}(12.2, 3.8)
   \includegraphics[width=30px]{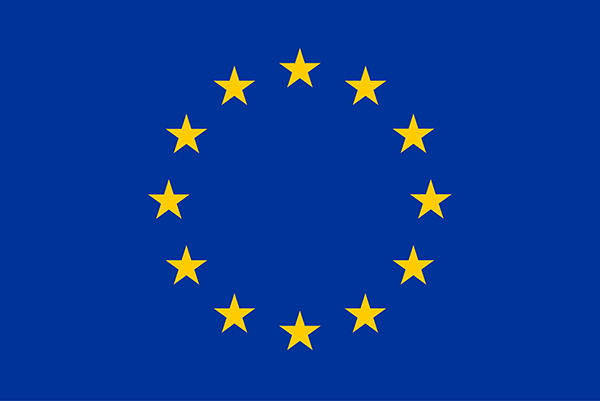}%
\end{textblock}

\newpage

\tableofcontents

\pagenumbering{arabic}
\thispagestyle{empty}
\newpage

\setcounter{page}{1}

\section{Introduction}
\label{sec-intro}

A very active line of research lying in the territory between Logic and Algorithms deals with the derivation of \emph{Algorithmic Meta-Theorems}, or AMTs for short.
According to Grohe and Kreutzer~\cite{GroheK09},
AMTs state that certain families of algorithmic problems, typically defined by some \emph{logical condition}, are ``tractable'' when the input graphs satisfy some specific property, typically defined by some \emph{combinatorial condition}.
When the logical condition is expressibility in a logic ${\sf L}$ and the structural condition is membership  in some graph class ${\cal C}$, the notion of ``tractability'' is most commonly considered as the following: if $\varphi$ is a formula in ${\sf L}$ and $G$ is a graph in ${\cal C}$, then deciding
whether $G\models \varphi$ (this is known as the \textsc{Model Checking} problem for ${\sf L}$) can be done in time
$f(|\varphi|, {\cal C})\cdot |G|^{{\cal O}(1)},$ for some function $f$.

Clearly, there are two natural orthogonal directions to boost the power of an AMT: consider logics ${\sf L}$ that are as general as possible or/and consider graph classes ${\cal C}$ that are as general as possible.
There is an abundance of such results concerning different compromises between the logical and combinatorial conditions \cite{Kreutzer11,GroheK09,Grohe08}.
Among them, a considerable amount of research deals with the definition of logics where AMTs can be derived under the  combinatorial condition
that $\mathcal{C}$ is some non-trivial {\sl minor-closed} graph class\cite{FominGST22,FlumG01fixe,Bojanczyk21separ,SchirrmacherSV23first,SchirrmacherSSTV24mode}.%
\footnote{A graph class $\mathcal{C}$ is \emph{minor-closed} if it contains all minors of its graphs (for the definition of the minor relation, see \Cref{def_twMSO}).
 $\mathcal{C}$  is ``non-trivial'' if it does not contain all graphs.
According to Robertson and Seymour's theorem \cite{RobertsonS04}, if $\mathcal{C}$  is non-trivial and minor-closed then it excludes some big enough clique as a minor.}
The pursuit of such  AMTs aimed at the definition of  several logics able to express, in the most general manner, the rich algorithmic
potential of the Graph Minors series of Robertson and Seymour.

In this paper we introduce a  fragment of counting monadic second-order logic that is able to subsume all\footnote{This assertion excludes variants based on cardinality constraints (see e.g.~\cite{KnopKMT19simp}).}
known AMTs on minor-closed graph classes to a single meta-algorithmic statement.
The key idea
in our logic is to introduce parameterization in  the quantification
of \textsf{CMSO} formulas, as discussed in~\Cref{sec_contribution}.

\paragraph{Courcelle's theorem: limitations and extensions.}
The archetypical example of an AMT is Courcelle's theorem~\cite{Courcelle90them,Courcelle97,Courcelle92} (see also \cite{BoriePT92auto,ArnborgLS91easy}), asserting that graph properties definable in \textsf{CMSO} (counting monadic second-order logic)
are decidable in linear time on graphs of bounded treewidth
(for the formal definition of \textsf{CMSO}, see \Cref{basic_logic}).
Due to the \textsl{grid-exclusion theorem} \cite{RobertsonS90b},
this AMT can be rewritten, using minor-exclusion, as follows.
\begin{proposition}
\label{question}
There exist a function $f: \mathbb{N}^2\to\mathbb{N}$ and an algorithm that,  given a planar graph $H$, a formula $φ\in\mathsf{CMSO}$,  and an $H$-minor-free graph $G$,
decides whether $G\models φ$ in time $f(|\varphi|,|H|)\cdot |G|$.
\end{proposition}

Interestingly, the planarity condition for the excluded minor $H$ in \autoref{question}
is \textsl{unavoidable} if one wishes to stick to the  full expressibility potential of  \textsf{CMSO}  (see \cite{Kreutzer12}).
Therefore, the emerging challenge is the following:
\begin{eqnarray}
\begin{minipage}{14cm}
\textsl{Identify fragments of} \textsf{CMSO} \textsl{that permit the proof of  AMTs as the one in \autoref{question} without the planarity condition for the excluded minor $H$.}
\end{minipage}\label{question2}
\end{eqnarray}
The first result in this direction dates back to \cite{FlumG01fixe}, where
Flum and  Grohe showed that a ``planarity-free'' version of \autoref{question} holds if we consider
first-order logic  (\textsf{FO})
instead of  the more expressive \textsf{CMSO}.
Here we should stress that
the meta-algorithmics of \textsf{FO} offered a rich line of research that nowadays
has achieved to prove AMTs under combinatorial
conditions that are much more general than minor-closedness (see \cite{Seese96line,DawarGK07loca,DvorakKT13firs,NesetrilM12spar,NesetrilO08,NesetrilO08I,NesetrilO08II,GroheKS17dec,BonnetKTW22twinI,DreierMS23}).
While the meta-algorithmics of \textsf{FO}  are pretty well studied  (see \cite{KreutzerD09,DvorakKT13firs,BonnetGMSTT22twin}),
the challenge in \eqref{question2} has established an independent line of research
that resulted in the introduction of  {intermediate} logics
that enjoy some of the expressibility characteristics of  \textsf{CMSO}
and permit tractable model checking on the more general realm of minor-closed graph classes.

The first result in this direction consisted in
enhancing \textsf{FO} with extra predicates expressing  the existence of a path avoiding some fixed number of vertices.
Such predicates are not \textsf{FO}-expressible, as they require quantification over sets of vertices.
This extension  was introduced independently
by  Schirrmacher, Siebertz, and Vigny
in \cite{SchirrmacherSV23first} (under the name {\sf FO+conn}) and
by  Bojańczyk in \cite{Bojanczyk21separ} (under the name \emph{separator logic}).
The second and more expressive extension of \textsf{FO}, also introduced by Schirrmacher, Siebertz, and Vigny in \cite{SchirrmacherSV23first}, is {\sf FO+dp}, which enhances \textsf{FO} with predicates expressing the existence of disjoint paths between certain pairs of vertices.
For {\sf FO+conn}, a tractable model checking algorithm in graphs excluding a (topological) minor
has been given by Pilipczuk, Schirrmacher, Siebertz,  Toruńczyk, and Vigny~\cite{PilipczukSSTV22algor}.
For the more expressive {\sf FO+dp}, a tractable model checking algorithm in graphs excluding a minor has been  given by
Golovach, Stamoulis, and Thilikos in~\cite{GolovachST22model} (this result has recently been extended for graphs excluding a topological minor in~\cite{SchirrmacherSSTV24mode}).
Note that, in both aforementioned logics,
expressibility of \textsf{FO} is boosted by allowing some particular and restricted \textsl{set} quantification ``hidden'' in the introduced predicates.
However, these logics do not allow further queries on the ``hidden'' quantified sets, and therefore they can be viewed more
as enhancements of  \textsf{FO} rather than restrictions of \textsf{CMSO}.
As an attempt to
go further than this, another enhancement of \textsf{FO} was recently introduced by Fomin,  Golovach, Sau,  Stamoulis, and   Thilikos~\cite{FominGSST23}, the so-called \emph{compound logic} $\tilde{Θ}$, which combines \textsf{FO} and \textsf{CMSO} and is essentially a language that was tailored so as to express general families of graph modification problems.
In~\cite{FominGSST23}, the authors provide an  AMT
on minor-excluding graph classes  for the enhancement of the compound logic with disjoint-paths predicates, denoted by $\tilde{Θ}^{\text{\textsf{dp}}}$.

\subsection{Our contribution}
\label{sec_contribution}

In this paper we deviate from the above approaches that deal with {enhancements of {\sf FO}}.
In the logic that we introduce, rather than allowing only particular combinations of quantified sets, we now allow arbitrary quantification on sets by restricting the {\sl structure} of each quantified set.
Namely,
we demand them to have bounded \emph{annotated treewidth},
a parameter introduced in~\cite{ThilikosW23excl},
which is defined as follows: given a graph $G$ and a set $X\subseteq V(G)$, the \emph{annotated treewidth} of $X$ in $G$, denoted by $\tw(G,X)$, is the maximum treewidth of an \emph{$X$-rooted minor} of $G$, i.e., a minor $H$ where the model of each vertex of $H$ contains some vertex of $X$.
That way, $\tw(G,X)$ can be seen as a measure of the contribution of $X$ to the tree-decomposability of $G$, in the sense that it corresponds to the largest treewidth of all minors of $G$ in which the set $X$ is ``spread''  across the vertex models.
For instance, if $Γ_{n}$ is a the $(n\times n)$-grid and $X$ are the vertices of its ``perimeter'', then $\tw(Γ_{n},X)=2,$
while if $X$ contains all the vertices of some ``diagonal'' of  $Γ_{n}$
or of the ``middle vertical'' path of $Γ_{n}$, then $\tw(Γ_{n},X)=Ω(n^{1/2})$.

More formally, we denote by $\mathsf{CMSO/tw}$ the restriction of \textsf{CMSO} where instead of using the quantifier $\exists X$ (resp.
$\forall X$) for a set variable $X$,
we introduce the quantifier $\exists_{k} X$ (resp.
$\forall_{k} X$) for some $k\in\mathbb{N}$, where $\exists_{k} X$ and $\forall_{k} X$ mean that the quantification is applied on sets $X\in 2^{V(G)}\cup 2^{E(G)}$ where $\tw(G,X) \leq k$.\footnote{For the case of quantification on a set $X\subseteq E(G)$ of edges,
$\tw(G,X) \leq k$ means that the set of the endpoints of the edges in $X$ has annotated treewidth at most $k$.}
In particular, we define
\begin{eqnarray}
\exists_{k} X\ φ \coloneqq \exists X\ (\tw(G,X)≤k \wedge φ)  & \mbox{and} & \forall_{k} X\ φ \coloneqq \neg(\exists_{k} X\ \neg φ).\label{quant}
\end{eqnarray}
(Note that $\forall_{k} X\ φ = \forall X\ (\tw(G,X)≤k \rightarrow φ)$.)

The main conceptual novelty of our approach is to  {\sl parameterize the quantification} allowed in the considered $\mathsf{CMSO}$ formulas
and consider the parameter of annotated treewidth
as the measure of this parametrization.
To the best of our knowledge, this is the first time that such a ``parameterized fragment'' of \textsf{CMSO} is considered.
Let us mention here existing work on \emph{extensions} of \textsf{MSO} by generalized quantifiers for \textsf{MSO}~\cite{Andersson02onse,Kontinen10defi,BurtschickV98lind} originating in the corresponding work for \textsf{FO}~\cite{Mostowski57onag,Lindstrom66gene} (see~\cite{Vaananen99gene} and~\cite[Chapter 12]{EbbinghausF99fini} for surveys -- for second order Category and Measure
quantifiers, see~\cite{MioSH18mona,BojanczykKS21msop,MichalewskiM16meas,Bojanczyk16thin}).

Concerning the choice of the parameter $\tw$ in the above definition, we should stress that among all possible minor-monotone
parameters,\footnote{A graph parameter is \emph{minor-monotone} if its value on a graph is never smaller than its value on its minors.} treewidth is the best possible choice. Indeed, as we argue in \autoref{sec_conclusions}, if $\p$ is a minor-monotone parameter that is not functionally
larger\footnote{We say that a parameter $\p$ is \emph{functionally larger} than a parameter $\p'$ if there is a function $f:\N\to\N$ such that for every graph
$G$, $\p'(G)≤f(\p(G))$.
} than treewidth, then the logic $\mathsf{CMSO/\p}$
has the same expressive power as the general $\mathsf{CMSO}$.

It is important to agree that,   when evaluating the length $|φ|$ of a formula
$φ\in\mathsf{CMSO/tw},$ the length of the numerical representation of $k$ is also taken into account.
{Note also that, for every fixed $k$, the property that $\tw(G,X) \leq k$ is \textsf{MSO}-expressible, and therefore $\mathsf{CMSO/tw}$ can indeed be seen as a fragment of  $\textsf{CMSO}$.}
Our results indicate that the quantification defined in \eqref{quant} is the key ingredient  for a
general answer to question \eqref{question2}.

We denote by $\mathsf{CMSO/tw}\!+\!\mathsf{dp}$ the logic obtained from $\mathsf{CMSO/tw}$ by allowing also disjoint-paths predicates, as it was already done in \cite{SchirrmacherSV23first,GolovachST22model} for the case of \textsf{FO}.
Note that  the application of this enhancement
on $\mathsf{CMSO}$ does not add any expressibility power, as disjoint-paths predicates are already expressible in $\mathsf{CMSO}$.
However,
$\mathsf{CMSO/tw}\!+\!\mathsf{dp}$ is more general than $\mathsf{CMSO/tw}$,
as the disjoint-paths predicates may involve quantification on vertex sets of \textsl{unbounded} annotated treewidth.
Our result, in its most general form, is the following.

\begin{result}\label{main_theorem_intro}
There exist a function $f: \mathbb{N}^2\to\mathbb{N}$ and an algorithm that,  given a  graph $H$, a formula $φ\in\mathsf{CMSO/tw}\!+\!\mathsf{dp}$,  and an $H$-minor-free graph $G$,
decides whether $G\models φ$ in time $f(|\varphi|,|H|)\cdot |G|^2$.
\end{result}

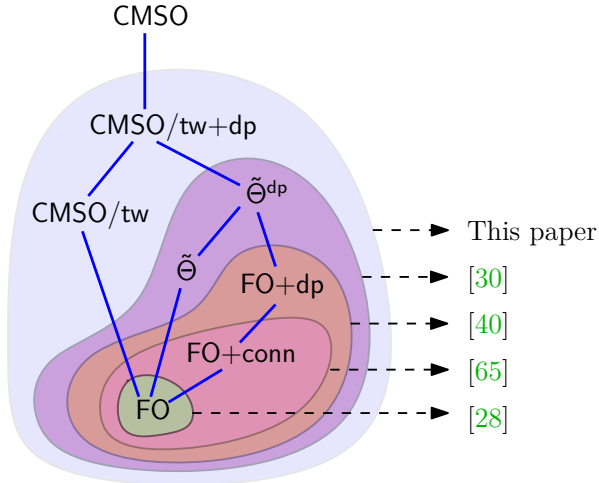
\begin{figure}[ht]
\centering
\scalebox{.88}{
\tikzstyle{ipe stylesheet} = [
  ipe import,
  even odd rule,
  line join=round,
  line cap=butt,
  ipe pen normal/.style={line width=0.4},
  ipe pen heavier/.style={line width=0.8},
  ipe pen fat/.style={line width=1.2},
  ipe pen ultrafat/.style={line width=2},
  ipe pen normal,
  ipe mark normal/.style={ipe mark scale=3},
  ipe mark large/.style={ipe mark scale=5},
  ipe mark small/.style={ipe mark scale=2},
  ipe mark tiny/.style={ipe mark scale=1.1},
  ipe mark normal,
  /pgf/arrow keys/.cd,
  ipe arrow normal/.style={scale=7},
  ipe arrow large/.style={scale=10},
  ipe arrow small/.style={scale=5},
  ipe arrow tiny/.style={scale=3},
  ipe arrow normal,
  /tikz/.cd,
  ipe arrows, 
  <->/.tip = ipe normal,
  ipe dash normal/.style={dash pattern=},
  ipe dash dotted/.style={dash pattern=on 1bp off 3bp},
  ipe dash dashed/.style={dash pattern=on 4bp off 4bp},
  ipe dash dash dotted/.style={dash pattern=on 4bp off 2bp on 1bp off 2bp},
  ipe dash dash dot dotted/.style={dash pattern=on 4bp off 2bp on 1bp off 2bp on 1bp off 2bp},
  ipe dash normal,
  ipe node/.append style={font=\normalsize},
  ipe stretch normal/.style={ipe node stretch=1},
  ipe stretch normal,
  ipe opacity 10/.style={opacity=0.1},
  ipe opacity 30/.style={opacity=0.3},
  ipe opacity 50/.style={opacity=0.5},
  ipe opacity 75/.style={opacity=0.75},
  ipe opacity opaque/.style={opacity=1},
  ipe opacity opaque,
]
\definecolor{red}{rgb}{1,0,0}
\definecolor{blue}{rgb}{0,0,1}
\definecolor{green}{rgb}{0,1,0}
\definecolor{yellow}{rgb}{1,1,0}
\definecolor{orange}{rgb}{1,0.647,0}
\definecolor{gold}{rgb}{1,0.843,0}
\definecolor{purple}{rgb}{0.627,0.125,0.941}
\definecolor{gray}{rgb}{0.745,0.745,0.745}
\definecolor{brown}{rgb}{0.647,0.165,0.165}
\definecolor{navy}{rgb}{0,0,0.502}
\definecolor{pink}{rgb}{1,0.753,0.796}
\definecolor{seagreen}{rgb}{0.18,0.545,0.341}
\definecolor{turquoise}{rgb}{0.251,0.878,0.816}
\definecolor{violet}{rgb}{0.933,0.51,0.933}
\definecolor{darkblue}{rgb}{0,0,0.545}
\definecolor{darkcyan}{rgb}{0,0.545,0.545}
\definecolor{darkgray}{rgb}{0.663,0.663,0.663}
\definecolor{darkgreen}{rgb}{0,0.392,0}
\definecolor{darkmagenta}{rgb}{0.545,0,0.545}
\definecolor{darkorange}{rgb}{1,0.549,0}
\definecolor{darkred}{rgb}{0.545,0,0}
\definecolor{lightblue}{rgb}{0.678,0.847,0.902}
\definecolor{lightcyan}{rgb}{0.878,1,1}
\definecolor{lightgray}{rgb}{0.827,0.827,0.827}
\definecolor{lightgreen}{rgb}{0.565,0.933,0.565}
\definecolor{lightyellow}{rgb}{1,1,0.878}
\definecolor{black}{rgb}{0,0,0}
\definecolor{white}{rgb}{1,1,1}
\begin{tikzpicture}[ipe stylesheet]
  \filldraw[shift={(141.791, 420.716)}, xscale=0.6364, yscale=0.5814, ipe pen heavier, fill=blue, ipe opacity 10]
    (0, 0)
     .. controls (9.1, 37.0305) and (27.3, 67.4145) .. (63.05, 80.7075)
     .. controls (98.8, 94.0005) and (152.1, 90.2025) .. (187.6667, 58.3638)
     .. controls (223.2333, 26.5252) and (241.0667, -33.3542) .. (246.7867, -84.6272)
     .. controls (252.5067, -135.9002) and (246.1133, -178.5668) .. (206.6467, -199.9002)
     .. controls (167.18, -221.2335) and (94.64, -221.2335) .. (53.17, -213.6375)
     .. controls (11.7, -206.0415) and (1.3, -190.8495) .. (-3.9, -174.708)
     .. controls (-9.1, -158.5665) and (-9.1, -141.4755) .. (-9.1, -111.0915)
     .. controls (-9.1, -80.7075) and (-9.1, -37.0305) .. cycle;
  \filldraw[shift={(186.909, 338.605)}, xscale=0.6364, yscale=0.5814, ipe pen heavier, line cap=round, fill=white]
    (0, 0)
     -- (0, 0);
  \filldraw[draw=black, ipe pen heavier, fill=darkmagenta, ipe opacity 30]
    (192.668, 374.8977)
     .. controls (202.1817, 389.2503) and (203.8363, 413.5397) .. (217.0727, 425.6843)
     .. controls (230.309, 437.829) and (255.127, 437.829) .. (272.3937, 415.09)
     .. controls (289.6603, 392.351) and (299.3757, 346.873) .. (280.7182, 323.2348)
     .. controls (262.0607, 299.5967) and (215.0303, 297.7983) .. (187.5152, 300.8992)
     .. controls (160, 304) and (152, 312) .. (148.7348, 321.2855)
     .. controls (145.4697, 330.571) and (146.9393, 341.142) .. (156.3605, 348.6355)
     .. controls (165.7817, 356.129) and (183.1543, 360.545) .. cycle;
  \filldraw[draw=black, ipe pen heavier, fill=darkorange, ipe opacity 30]
    (214.9348, 368.7367)
     .. controls (225.2757, 376.465) and (226.5513, 388.93) .. (237.1165, 394.0585)
     .. controls (247.6817, 399.187) and (267.5363, 396.979) .. (277.05, 381.5222)
     .. controls (286.5637, 366.0653) and (285.7363, 337.3597) .. (268.7772, 321.3508)
     .. controls (251.818, 305.342) and (218.727, 302.03) .. (197.2422, 303.921)
     .. controls (175.7573, 305.812) and (165.8787, 312.906) .. (162.5697, 322.2905)
     .. controls (159.2607, 331.675) and (162.5213, 343.35) .. (173.8547, 350.6833)
     .. controls (185.188, 358.0167) and (204.594, 361.0083) .. cycle;
  \filldraw[draw=black, ipe pen heavier, fill=violet, ipe opacity 30]
    (242.082, 364.116)
     .. controls (258.8183, 366.5823) and (269.9547, 366.5467) .. (273.341, 358.1103)
     .. controls (276.7273, 349.674) and (272.3637, 332.837) .. (256.1818, 322.4185)
     .. controls (240, 312) and (212, 308) .. (196, 309.3333)
     .. controls (180, 310.6667) and (176, 317.3333) .. (175.4455, 325.2855)
     .. controls (174.891, 333.2377) and (177.782, 342.4753) .. (190.3957, 349.5783)
     .. controls (203.0093, 356.6813) and (225.3457, 361.6497) .. cycle;
  \filldraw[draw=black, ipe pen heavier, fill=lightgreen, ipe opacity 50]
    (208.8001, 336.2559)
     .. controls (214.591, 331.8396) and (216.2456, 325.2152) .. (214.591, 321.351)
     .. controls (212.9365, 317.4868) and (207.9728, 316.3827) .. (201.8058, 315.4019)
     .. controls (195.6388, 314.421) and (188.2684, 313.5633) .. (184.9593, 318.5316)
     .. controls (181.6502, 323.4999) and (182.4024, 334.2941) .. (187.7422, 338.5872)
     .. controls (193.0819, 342.8803) and (203.0092, 340.6721) .. cycle;
  \node[ipe node, font=\large]
     at (190.601, 321.903) {${\sf FO}$};
  \node[ipe node, font=\large]
     at (145.864, 406.752) {${\sf CMSO/tw}$};
  \node[ipe node, font=\large]
     at (180.672, 490.825) {${\sf CMSO}$};
  \node[ipe node, font=\large]
     at (235.337, 376.237) {$\text{\textsf{FO}+\textsf{dp}}$};
  \node[ipe node, font=\large]
     at (207.973, 381.522) {${\sf \tilde{\Theta}}$};
  \node[ipe node, font=\large]
     at (237.755, 414.645) {${\sf \tilde{\Theta}}^{\sf dp}$};
  \node[ipe node, font=\large]
     at (170.746, 444.454) {$\text{\textsf{CMSO/tw}+\textsf{dp}}$};
  \node[ipe node, font=\large]
     at (212.201, 346.398) {$\text{\textsf{FO}+\textsf{conn}}$};
  \begin{scope}[shift={(49.4545, 208.373)}, xscale=0.6364, yscale=0.5814]
    \filldraw[draw=blue, ipe pen fat, fill=white]
      (236, 400)
       -- (292, 368);
    \filldraw[draw=blue, ipe pen fat, fill=white]
      (316, 308)
       -- (304, 348);
    \filldraw[draw=blue, ipe pen fat, fill=white]
      (252, 292)
       -- (232, 212);
    \filldraw[draw=blue, ipe pen fat, fill=white]
      (186, 334)
       -- (224, 212);
    \filldraw[draw=blue, ipe pen fat, fill=white]
      (190, 360)
       -- (220, 400);
    \filldraw[draw=blue, ipe pen fat, fill=white]
      (228, 424)
       -- (228, 480);
    \filldraw[draw=blue, ipe pen fat, fill=white]
      (292, 352)
       -- (264, 316);
    \filldraw[draw=blue, ipe pen fat, fill=white]
      (292, 252)
       -- (316, 280);
    \filldraw[draw=blue, ipe pen fat, fill=white]
      (244, 208)
       -- (280, 232);
  \end{scope}
  \filldraw[shift={(214.909, 324.651)}, xscale=0.6364, yscale=0.5814, ipe pen heavier, ipe dash dashed, ->, fill=darkmagenta]
    (0, 0)
     -- (172, 0);
  \filldraw[shift={(273.978, 343.245)}, xscale=0.6364, yscale=0.5814, ipe pen heavier, ipe dash dashed, ->, fill=darkmagenta]
    (0, 0)
     -- (79.177, 0.019);
  \node[ipe node, font=\large]
     at (328.673, 338.825) {\cite{PilipczukSSTV22algor}};
  \node[ipe node, font=\large]
     at (328.673, 318.825) {\cite{FlumG01fixe}};
  \node[ipe node, font=\large]
     at (328.673, 358.825) {\cite{GolovachST22model}};
  \node[ipe node, font=\large]
     at (328.673, 378.825) {\cite{FominGSST23}};
  \node[ipe node, font=\large]
     at (328.673, 398.825) {~This paper};
  \filldraw[ipe pen heavier, ipe dash dashed, ->, fill=darkmagenta]
    (292.189, 403.071)
     -- (324.416, 403.018);
  \filldraw[shift={(287.645, 382.992)}, xscale=0.6364, yscale=0.5814, ipe pen heavier, ipe dash dashed, ->, fill=darkmagenta]
    (0, 0)
     -- (57.225, 0);
  \filldraw[shift={(282.728, 362.997)}, xscale=0.6364, yscale=0.5814, ipe pen heavier, ipe dash dashed, ->, fill=darkmagenta]
    (0, 0)
     -- (65.026, 0.014);
\end{tikzpicture}

}
\caption{The Hasse diagram of the poset of some logics between $\mathsf{FO}$ and $\mathsf{CMSO}$ along with the corresponding known and recent results on AMTs whose combinatorial condition is minor-exclusion. An edge between two logics means that the formulas of the lower one are expressible by the formulas of the higher one.}
\label{fig-Hasse}
\end{figure}

Note that \Cref{main_theorem_intro} can be rephrased as the property that model checking for $\mathsf{CMSO/tw}\!+\!\mathsf{dp}$ is tractable on non-trivial
minor-closed graph classes, and subsumes {\sl all} known meta-algorithmic results on the tractability of logics on non-trivial minor-closed graph classes, such as those in~\cite{FlumG01fixe,GolovachST22model,FominGSST23,FominGST22}.
This is schematized in the Hasse diagram of \autoref{fig-Hasse}.
On the other hand, recall that Courcelle's theorem, stated in \autoref{question},
asserts that $\mathsf{CMSO}$ model checking is tractable on graphs of bounded treewidth, and note that, when restricted to this class of graphs, the logics $\mathsf{CMSO}$ and $\mathsf{CMSO/tw}$ {(as well as $\mathsf{CMSO/tw}\!+\!\mathsf{dp}$)}
have the same expressive power,
since on bounded-treewidth graphs every set of vertices/edges trivially has bounded annotated treewidth.
Therefore, \Cref{main_theorem_intro} can be interpreted as a far-reaching generalization of Courcelle's theorem under a new paradigm in which the bounded-treewidth restriction is not anymore applied to the {\sl input graph}, but is instead applied, in its annotated form, to the quantification of the considered {\sl fragment of $\mathsf{CMSO}$}.
Our logic allows to lift the tractability from bounded-treewidth graphs to non-trivial minor-closed graph classes, and with the additional ``bonus'' of the $\mathsf{dp}$-extension of the $\mathsf{CMSO/tw}$ logic.

We should stress that the  proof of \Cref{main_theorem_intro}
is arguably more comprehensive, both from the conceptual and the technical point of view, {than
those of the  AMTs in \cite{GolovachST22model,FominGSST23,FominGST22}.}

\subsection{Our techniques}
\Cref{main_theorem_intro} capitalizes on the main algorithmic potential of the Graph Minors series of Robertson and Seymour.
Namely, it exploits the ``locality'' of sets of bounded annotated treewidth  and follows the scheme of the irrelevant-vertex technique of~\cite{RobertsonS95XIII}.
This approach is part of a greater framework that we introduce, concerning \emph{reductions between meta-theorems}, where
the task of finding irrelevant vertices for a problem is ``localized'' in a way that makes it (meta-)algorithmically tractable.
In what follows, we provide a brief exposition of the principal ideas of our proofs.

\vspace{-2mm}
\paragraph{Simplifying instances.}
Our approach to model checking is based on the irrelevant-vertex technique, which was introduced by Robertson and Seymour for the design of an algorithm solving the \textsc{Disjoint Paths} problem~\cite{RobertsonS95XIII}.
The core idea of this technique, from its original setting~\cite{RobertsonS95XIII} to a striking majority of its subsequent applications
in algorithm design for minor-free graphs, is that a candidate solution for the problem can be ``rerouted away'' from some sufficiently internal part of a grid, typically emerging from a structure called \emph{flat wall}.
This allows the characterization of the ``avoidable'' vertices as \emph{irrelevant}, in the sense that their removal does not change the {\sf yes}/{\sf no} status of the instance.
The repeated removal of such vertices results in an instance of bounded treewidth that is equivalent to the original instance. In the reduced instance, bounded treewidth can be used
algorithmically, typically via explicit dynamic programming or using the meta-result of~\Cref{question}, in order to solve the problem.

However, the {above approach}  does not readily work for model checking. The reason is that, even for \textsf{FO},
the removal of irrelevant vertices from a graph that has unbounded treewidth (or is even just planar) for some given fixed formula $\varphi$ would result to a bounded-treewidth graph that is equivalent, in terms of satisfaction of $\varphi$, this 
would contradict the undecidability of the satisfiability problem
(i.e., whether, for a given formula $\varphi$,
there is some graph satisfying $\varphi$)
for \textsf{FO};
see e.g.~\cite[Section~9.1]{Libkin04elem}.
Such an obstacle can be avoided by ``shifting'' from the initial problem with input $(G,\varphi)$
to an alternative problem in which the input graph $G$
comes together with an annotation $R_1,\ldots,R_h \subseteq V(G)$ of its vertex set, for some $h\in\mathbb{N}$ depending on $\varphi$,
and structures of the form $(G,R_1,\ldots,R_h)$, called \emph{annotated graphs},
are asked to satisfy some new formula $\hat{\varphi}$.
Initially, each $R_i$ is set to be equal to $V(G)$ and $(G,R_1,\ldots,R_h)$ satisfies $\hat{\varphi}$ if and only if $G$ satisfies $\varphi$, and our goal is to obtain some ``simplified'' instance $(G',R_1',\ldots,R_h')$ that is equivalent to $(G,R_1,\ldots,R_h)$
with respect to the satisfaction of $\hat{\varphi}$.
This idea, originating in~\cite{GolovachKMT17thep}, plays a key role in the proof of the AMTs in~\cite{GolovachST22model,FominGSST23,FominGST22,SchirrmacherSSTV24mode} and is also used in this work.

\paragraph{Annotated types of annotated graphs.}
In this paper, the formulas for which we find irrelevant vertices are the ones expressing the \emph{annotated type} of an annotated graph.
Given an $h\in\mathbb{N}$,
the annotated $h$-type of an annotated graph $(G,R_1,\ldots,R_h)$ can be seen as a ``tree-like'' encoding of all sentences of $h$
(first-order and monadic second-order)
variables in prenex normal form\footnote{For the current section, we assume that all considered formulas are in \emph{prenex normal form}, i.e., they are ``split'' into a prefix part that contains only quantifiers and a suffix part that is quantifier-free.
The quantifier-free part can be seen as a boolean combination of \emph{atomic formulas}, for some appropriate definition of atomic formulas depending on the logic we consider. For the corresponding definitions for $\mathsf{CMSO}/\tw,\!+\!\mathsf{dp}$, see~\Cref{def_twMSO} and~\Cref{lem_trans_CMSO/tw_MSO}.} that $G$ satisfies, where the $i$-th quantified variable is asked to be interpreted in $R_i$.
This corresponds to the annotated variant of \emph{$h$-rank} \textsf{FO}/\textsf{MSO}-\emph{types} given, for instance, by Shelah~\cite{Shelah75them}, Ebbinghaus and Flum~\cite[Definition~2.2.5]{EbbinghausF99fini}, and Libkin~\cite[Section~3.4]{Libkin04elem}; see also~\cite{BonnetKTW22twinI,gradel2003automata,GajarskyGK20diff,gajarsky16algorithmic,GolovachST22model,SchirrmacherSSTV24mode}.
{Our notation builds on the one used in~\cite{GolovachST22model,GolovachST22model_arXiv,SchirrmacherSSTV24mode},
since this facilitates the exposition of our proofs.}

Since annotated $h$-types fully determine which formulas of $h$ quantified variables
are satisfied by
an annotated graph $(G,R_1,\ldots,R_h)$, the \textsc{Model Checking} problem for any formula $\varphi$ on a input graph $G$
is translated to the computation of the (annotated) $h$-type of the graph $G$ annotated with $h$ copies of its vertex set.
To compute the annotated $h$-type, rather than computing the atomic type of every $h$-tuple of sets/vertices (which yields an algorithm running in time $n^{\mathcal{O}(h)}$),
we opt for the design of a procedure that \textsl{detects} and \textsl{removes} sets $S_0,\ldots,S_h$ of vertices that are ``irrelevant'',
in the sense that they can be safely discarded from the sets $V(G),R_1,\ldots,R_h,$ respectively, so as to obtain an annotated (sub){graph}  $(G',R_1',\ldots,R_h')$
that has the \textsl{same annotated type} as $(G,R_1,\ldots,R_h)$. By repeating this process, we obtain an equivalent annotated instance (of the same annotated type as the original instance) which has small treewidth and, using Courcelle's theorem, we can finally compute its annotated type (which indeed is \textsf{MSO}-expressible).
For the irrelevant-vertex arguments, it is important to note that the number of different possible annotated types, over all annotated graphs $(G,R_1,\ldots,R_h)$, depends only on $h$ and the number of different (up to logical equivalence) atomic types on $h$ variables.
A preliminary version of the above scheme {(restricted to \textsf{FO}{+\textsf{dp}} formulas)} appears in~\cite{GolovachST22model,GolovachST22model_arXiv}.

\paragraph{Reductions between AMTs.}
{In our setting},
irrelevant vertices for the annotated type are obtained via a ``local-global'' argument.
We detect an area of the graph where ``\textsl{locally}'' irrelevant vertices for the $R_i$'s, for some appropriate ``local variant'' of the annotated type, are also ``\textsl{globally}'' irrelevant for the (global) annotated type.
In other words, we describe a way to localize the task of finding irrelevant vertices in a two-fold manner:
``locally'' irrelevant vertices are also ``globally'' irrelevant, and such vertices can be found \textsl{locally} using black-box application of some meta-algorithmic result.

We make the above idea more abstract by introducing a notion of \emph{reductions between meta-theorems}.
In the viewpoint we adopt, graph classes can be expressed as families of graphs where some (structural) parameter is bounded.
Therefore, AMTs are stated for \textsl{Logic/Parameter pairs}, also called \emph{LP-pairs}, of the form $(\mathsf{L},\p)$, where  $\mathsf{L}$ is a logic and $\p$ is a graph parameter, mapping graphs to non-negative integers.
In its simplest form, reductions between AMTs are defined as follows.
We say that $(\mathsf{L}_1,\p_1)$ \emph{is reduced to} $(\mathsf{L}_2,\p_2)$ if the following conditions hold:
\begin{description}
\item \textbf{(Definability Condition)} For every $h\in\mathbb{N}$, annotated $h$-types $\Psi_h$ for $\mathsf{L}_1$ can be defined in $\mathsf{L}_2$,
\item \textbf{(Local-Global-Irrelevancy Condition)}
There is an effectively computable set $\Phi_h$ of $\mathsf{L}_2$-formulas, for every $h\in\mathbb{N}$,
and an algorithm $\mathbb{A}$ that given an $h\in\mathbb{N}$, a graph~$G$,
and an $h$-tuple $\bar{R}=(R_1,\ldots,R_{h})$ of subsets of $V(G)$,
if $\p_2(G)$ is ``large''  as a function of $h$ and $\p_1(G)$,
then $\mathbb{A}$ outputs, in time $f(h,\p_1(G))\cdot |G|^{\mathcal{O}(1)}$, a set $Y\subseteq V(G)$ such that
\begin{description}
\item (a) $\p_2(G[Y])$ is ``small'' , i.e., upper-bounded by a function of $h$ and $\p_1(G)$,
\item (b) irrelevant vertices for $\Phi_h$, in $(G[Y],\bar{R}\cap Y)$, are also irrelevant for $\Psi_h$ in $(G,\bar{R})$,
\item (c) such irrelevant vertices for $\Phi_h$ exist in $Y$.
\end{description}
\end{description}
If in the running time of $\mathbb{A}$, the dependency on $|G|$ is linear, then we say that $(\mathsf{L}_1,\p_1)$ \emph{is linear-time reduced to} $(\mathsf{L}_2,\p_2)$.
Reductions are defined formally and in a more general setting in~\Cref{reduct_sec} (see~\Cref{def_red}). We prove that if $(\mathsf{L}_1,\p_1)$ can be reduced to $(\mathsf{L}_2,\p_2)$ and the \textsl{evaluation} of $\mathsf{L}_2$-expressible formulas is tractable on graphs of bounded $\p_2$, then $\mathsf{L}_1$ model checking is tractable on graphs of bounded $\p_1$; see~\Cref{thm_reduce}.
In this case, the model checking algorithm for $\mathsf{L}_1$ has two stages.
First, it consists of an iterative application of algorithm $\mathbb{A}$ of the Local-Global-Irrelevancy Condition, in order to detect and remove locally irrelevant vertices which we know that are also globally irrelevant. This is done ``locally'' using the evaluation meta-algorithm for $(\mathsf{L}_2,\p_2)$.
The final outcome, after repeated applications of that procedure, is an annotated graph $(G',R_1',\ldots,R_h')$
with the same annotated $h$-type as $G$ (annotated with $h$ copies of $V(G)$) where $\p_2(G')$ is ``small'', i.e, upper-bounded by some  function of $\p_1(G)$ and $h$.
Therefore, using again the evaluation meta-algorithm for $(\mathsf{L}_2,\p_2)$, we can compute the annotated $h$-type of $(G',R_1',\ldots,R_h')$ and therefore conclude whether the initial graph $G$ satisfies a given formula $\varphi$ of $h$ quantifiers.

Let us mention here that reductions between AMTs already exist in the literature in the form of \textsl{transductions}; see e.g.~\cite[Theorem~4.3]{ArnborgLS91easy} and~\cite[Section 5]{GroheK09}.
In our case, reductions focus on iteratively translating ``local parts'' of the input (and the ``candidate'' annotated type) in order to gradually simplify the original instance,
rather than using syntactic interpretations to completely shift the \textsc{Model Checking} problem into another class of inputs (which is typically the case in the setting of syntactic interpretations; see e.g. the recent powerful results of~\cite{GroheK09,Bonnet22model,GajarskyHOLR20anew,GajarskyKNMPST20firs,DreierMS23}).

To prove that an LP-pair $(\mathsf{L}_1,\p_1)$ is reduced to an LP-pair $(\mathsf{L}_2,\p_2)$,
one has to come up with some appropriate set of formulas $\Phi_h$ and provide a suitable algorithm $\mathbb{A}$
that, when combined, yield the  \textsl{Local-Global-Irrelevancy Condition}.
In such an enterprise, the most demanding task is to prove that
``locally'' irrelevant vertices are also ``globally'' irrelevant.
{For this, one has to develop a ``composition'' argument on how the global annotated type can be split into two \textsl{partial} parts, one of them corresponding to some part of the graph with ``small $\p_2$'' (which can be detected by $\mathbb{A}$) and explain how computing this partial part of the annotated type can give us information about the global annotated type. These ideas date back to the work of Feferman and Vaught~\cite{FefermanV59firs}.}
For the setting where annotated types are ``glued'' in a common set of unbounded size, previous work was done in~\cite{FominGST22,FominGSST23,GolovachST22model,GolovachST22model_arXiv}.

After defining the above reduction framework  between AMTs, our goal is to prove that
\begin{equation}
\emph{$(\mathsf{CMSO/tw}\!+\!\mathsf{dp},\mathsf{hw})$ is linear-time reduced to $(\mathsf{CMSO},\mathsf{tw})$},\label{eq_red}
\end{equation}
where $\mathsf{hw}(G)$ is the \emph{Hadwiger number} of $G$, i.e., the minimum $k$ such that $G$ does not contain $K_k$ as a minor, while $\tw(G)$ is the treewidth of $G$.
This, combined with~\Cref{thm_reduce} and the \textsl{evaluation variant} of Courcelle's theorem~\cite[Theorem 2.10]{CourcelleM93} (see also,~\cite{ArnborgLS91easy} and~\cite[Corollary 4.15]{FlumFG02quer})
implies tractable model checking for $\mathsf{CMSO/tw}\!+\!\mathsf{dp}$ on classes of bounded Hadwiger number (\Cref{main_theorem_intro}).
Using the framework of reductions, one can also prove that if
$\mathsf{L}\in\{\mathsf{FO},\text{\textsf{FO}+\textsf{dp}},\mathsf{CMSO}/\tw,\mathsf{CMSO}/\tw\!+\!\mathsf{dp}\}$ and
$\p\in\{\mathsf{eg},\mathsf{hw}\}$ (here $\mathsf{eg}$ corresponds to Euler genus),
then $(\mathsf{L},\p)$ is linear-time reduced to $(\mathsf{CMSO},\mathsf{tw})$ (for $(\mathsf{L},\p) = (\text{\textsf{FO}+\textsf{dp}},\mathsf{hw})$,
the proof is in~\cite{GolovachST22model,GolovachST22model_arXiv}).
However, the statement of \eqref{eq_red} subsumes all these statements.
Note that, using the same framework, one can show that if
$\mathsf{L}\in\{\text{\textsf{FO}+\textsf{sdp}},\mathsf{CMSO}/\tw\!+\!\mathsf{sdp}\}$,
then $(\mathsf{L},\mathsf{eg})$ is linear-time reduced to $(\mathsf{CMSO},\mathsf{tw})$\footnote{$\text{\textsf{FO}+\textsf{dsp}}$ and $\mathsf{CMSO/tw}\!+\!\mathsf{sdp}$ are the extensions of $\mathsf{FO}$ and $\mathsf{CMSO/tw}$ with \emph{scattered disjoint-paths predicates}, as defined in~\cite{GolovachST22model,GolovachST22model_arXiv}.}.

\paragraph{Novelties of our approach.}
As discussed earlier, the notion of reduction between AMTs that we introduce is a new concept that stems from the approach of~\cite{FominGSST23,GolovachST22model} and evolves it in a more abstract and general way, aiming to establish it as a base for the study and the proof of other AMTs. To the best of our knowledge, it is the first time that the irrelevant-vertex technique is abstracted as a \emph{method} for the derivation of AMTs in such a general setting.

The crucial property of having bounded annotated treewidth is not only allowing us to obtain a fragment of $\mathsf{CMSO}$ for which model checking \textsl{can} be tractable on unbounded-treewidth inputs (without contradicting general complexity assumptions; see~\cite[Section 9]{GroheK09}),
but also permits to see sets of bounded annotated treewidth as ``local'' sets inside a grid-like structure. In fact, sets of bounded annotated treewidth are \textsl{local} inside \emph{flat railed annuli}: systems of concentric cycles in a graph together with a collection of disjoint paths orthogonally traversing them, in such a way that there are no long ``jumps'' in the graph that can connect distant cycles or paths (see~\Cref{fig:railedann}).
The locality of sets of  bounded annotated treewidth  allow us to lift locality arguments for \textsf{FO} and, in particular, argue that for every tuple of sets of bounded annotated treewidth inside a flat railed annulus, we can find a ``strip'' of the railed annulus that is disjoint of these sets and crops out some part of them; see~\autoref{lemma:buffer}. This property lies in the very base of our arguments toward showing the Local-Global-Irrelevancy Condition in our reduction.

In fact, this property is essential to prove both items (b) and (c) of the \textsl{Local-Global-Irrelevancy Condition}. Let us particularly stress that, in order to prove (c), one has to show that there is a way to ``locally'' reduce the $R_i$'s while not changing the (partial) annotated type of the (partial) annotated graph.
For $R_i$'s that correspond to first-order quantification, one may simply locally restrict $R_i$ only to ``representative'' vertices that represent
the satisfaction of all different annotated types of one level less and discard duplicates, eventually leading to a ``small'' number of representative vertices for each $R_i$. This was the approach of~\cite{GolovachST22model,GolovachST22model_arXiv}.
In the case of \MSO-quantification,
this approach fails because the (union of) these representative sets have unbounded size and therefore we should not hope to bound the total number of sets in a locally representative $R_i'$.
To deal with this, we use the fact that the union of representative sets has bounded annotated treewidth (\Cref{prop:union}).
In a large flat railed annulus, we can find a strip disjoint from these sets. We use this observation to iteratively ``zoom'' inside such a strip to
prove that we can restrict each $R_i$ after in an area that is disjoint from the already restricted $R_1,\ldots,R_{i-1}$; see~\Cref{lemma:reducing_annotation}.

Another cornerstone of our approach is that the quantifier $\exists_k X$ (as well as $\forall_k X$) can be expressed in $\text{\textsf{FO}+\textsf{dp}}$, via the excluded-minor characterization of bounded annotated treewidth (see~\Cref{formulaanntw}). Such a translation of the bounded annotated treewidth demand for a set $X$ into a query on the satisfiability of certain disjoint-paths queries between terminals in $X$, allows to focus our study on ``just'' (more complicated) queries on disjoint paths between terminals that are picked in certain (quantified) sets; see~\Cref{lem_trans_CMSO/tw_MSO}.
Shifting to disjoint paths, we can use existing tools to achieve the \textsl{Local-Global-Irrelevancy Condition} in the presence of disjoint-paths predicates,
as the ones developed in~\cite{GolovachST22model,GolovachST22model_arXiv}.
In our work,
when it comes to dealing with the disjoint-paths predicates,
we further develop the technical lemmas of~\cite{GolovachST22model,GolovachST22model_arXiv}, which we  use in our proofs either in a black-box manner (\Cref{lem:translation-dp-apices,lem_colomodelsrerout,lemma:irrelevant_area}) or by proving them in a more general setting (\Cref{lem_inter2equiv,lemma:intermediate1}).
The framework of flat railed annuli that we present in this paper, which serves as the main combinatorial base of the \textsl{Local-Global-Irrelevancy Condition}, evolves the one of~\cite{GolovachST22model,GolovachST22model_arXiv,FominGSST23} and originates from the combinatorial tools studied in~\cite{SauST24amor}.

\paragraph{Organization of the paper.}
The rest of the paper is structured as follows.
\Cref{sec_two} is dedicated to the presentation of the main definitions of this paper, including the formal definition of the logic $\mathsf{CMSO/tw}\!+\!\mathsf{dp}$ and of the reductions between AMTs, as well as the formal statement of our results.
The proof of~\eqref{eq_red} spans \Cref{sec_gametrees,sec_flatannuli,sec_enhancedstruct,sec_exchangability}.
In particular, in~\Cref{sec_gametrees} we prove the Definability Condition,
and in~\Cref{sec_enhancedstruct,sec_exchangability} we prove the Local-Global-Irrelevancy Condition.
\Cref{sec_flatannuli} contains the framework of flat railed annuli, which is used as the underlying combinatorial base of the definitions and proofs of~\Cref{sec_enhancedstruct,sec_exchangability}.
Concluding discussions and directions for future research can be found in~\Cref{sec_conclusions}.

\section{Main definitions and formal statement of the results}
\label{sec_two}

In this section we provide the main definitions and statements involved in the proof of \Cref{main_theorem_intro}.
In this paper, we consider only \textsl{finite relational} {\em vocabularies} consisting of relation and constant symbols (we do not use function symbols).
Every relation symbol $R$ is associated with a positive integer, its {\em arity}.
A {\em structure $\mathfrak{A}$ of vocabulary $\sigma$}, in short a {\em $\sigma$-structure}, consists of a non-empty \textsl{finite} set  $V(\mathfrak{A}),$ called the {\em universe} of $\mathfrak{A},$
an $r$-ary relation $R^{\mathfrak{A}}\subseteq {V(\mathfrak{A})}^r$ for each relation symbol $R\in \sigma$ of arity $r\geq  1,$
and an element
 $c^{\mathfrak{A}}\in V(\mathfrak{A})$ for each constant symbol $c\in\sigma$.
We refer to $R^\mathfrak{A}$ (resp. $c^{\mathfrak{A}}$) as the {\em interpretation of the symbol $R$ (resp. $c$) in the structure $\mathfrak{A}$}.
When there is no ambiguity, we will not distinguish between relation symbols and their interpretations.
We use $|\mathfrak{A}|$ to denote $|V(\mathfrak{A})|$.
An undirected graph  without loops can be seen as an $\{E\}$-structure $\mathfrak{G} = (V(\mathfrak{G}),E^{\mathfrak{G}}),$ where $E^{\mathfrak{G}}$ is a binary relation that is symmetric and anti-reflexive.
A \emph{colored-graph vocabulary} contains a binary relation symbol $E$ that is always interpreted as a symmetric and anti-reflexive binary relation (corresponding to edges), a collection of unary relation symbols $Y_1,\ldots, Y_h$ (corresponding to colors on the vertices), and a collection of constant symbols $c_1,\ldots,c_l$ (corresponding to ``root'' vertices).
Given a $\sigma$-structure $\mathfrak{G}$, we use $\mathsf{Gaifman}(\mathfrak{G})$ to denote the Gaifman graph\footnote{For a structure $\mathfrak{G}$ of vocabulary $\sigma$, its \emph{Gaifman graph} is the graph whose vertex set is $V(\mathfrak{G})$ and two vertices are adjacent if they both appear in a tuple of the same relation. For colored-graph vocabularies, this is $(V(\mathfrak{G}),E^{\mathfrak{G}})$.}
 of $\mathfrak{G}$.
Given two integers $p,q$ such that $p\leq q$, we set $[p,q]:=\{p,\ldots,q\}$
and, if $p\geq 1$, $[p]:=[1,p]$.

\subsection{Basic definitions from logic}
\label{basic_logic}
We now define the syntax and the semantics of first-order logic and monadic second-order logic of a vocabulary $\sigma.$ We refer the reader to~\cite{EbbinghausF99fini,Libkin04elem,CourcelleE12grap}
for a broader discussion on logical structures, first-order logic, and monadic second-order logic, from the viewpoint of graphs.
We assume the existence of a countable infinite set of {\em first-order variables},
usually denoted by lowercase symbols $x_1,x_2,\ldots,$
and of a countable infinite set of {\em set variables},
usually denoted by uppercase symbols $X_1,X_2, \ldots.$
A {\em first-order term} is either a first-order variable or a constant symbol.

A {\em first-order logic formula}, in short {\em {\sf FO}-formula}, of vocabulary $\sigma$ is built from atomic formulas of the form $x=y$
and $(x_1, \ldots, x_r)\in R,$ where $R\in \sigma$ and has arity $r\geq  1,$ on first-order terms $x,y,x_1,\ldots, x_r,$
by using the logical connectives $\vee,$ $\neg$ and the
quantifier $\exists$ on first-order variables.

A \emph{monadic second-order logic formula}, in short \textsf{MSO}-\emph{formula}, of vocabulary $\sigma$ is obtained by enhancing
the syntax of {\sf FO}-formulas by allowing atomic formulas of the form $x\in X,$
for some first-order term $x$ and some set variable $X,$ and allowing quantification both on first-order and set variables\footnote{
{We clarify that
\textsf{MSO} is commonly referred in the literature as {\sf MSO}$_1$.
For the vocabulary of graphs, {\sf MSO}$_2$ is defined by also allowing quantification over edges and edge sets.
Every $\mathsf{MSO}_2$-formula on $G$ can be turned into an $\mathsf{MSO}_1$-formula on the \emph{incidence graph} $I(G)$ of $G$, defined as $(V(G)\cup E(G), \big\{\{v,e\}\mid v\in e\big\})$.
In this paper, for the vocabulary of graphs, we use \textsf{MSO} to denote $\mathsf{MSO}_2$. Note that for a vertex/edge set $X$ in $G$, $\tw(G,X) = \tw(I(G),X)$.
}
}.
We use $\forall x \varphi$ and $\forall X \varphi$ as an abbreviation for $\neg (\exists x\ (\neg \varphi))$ and $\neg (\exists X\ (\neg \varphi))$, respectively.
A {\em counting monadic second-order logic formula}, in short {\em {\sf CMSO}-formula}, of vocabulary $\sigma$ is obtained by enhancing the syntax of \textsf{MSO}-formulas by allowing predicates
of the form ${\sf card}_p(X),$ expressing that $|X|$ is a multiple of an integer $p>1$, for a fixed integer $p>1$.

Given a logic $\mathsf{L}$, we denote by $\mathsf{L}[\sigma]$ the set of all \textsf{L}-formulas of vocabulary $\sigma.$
Given a formula $\varphi$ of some logic $\mathsf{L}$,  the {\em free variables} of $\varphi$ are its variables
that are not in the scope of any quantifier.
We write $\varphi(x_1,\ldots, x_k, X_1,\ldots, X_\ell)$ to indicate that the free variables of the formula $\varphi$ are among $x_1, \ldots,x_k,X_1, \ldots,X_\ell$.
A {\em sentence} is a formula without free variables.
Let $\mathsf{L}$
be a logic and $\sigma$ be a vocabulary.
Given a $\sigma$-structure $\mathfrak{A},$ a formula $\varphi(x_1,\ldots, x_k, X_1,\ldots, X_\ell)\in \mathsf{L}[\sigma],$
$a_1,\ldots, a_k\in V(\mathfrak{A}),$ and $A_1,\ldots,A_\ell \subseteq V(\mathfrak{A})$,
we write $\mathfrak{A}\models \varphi(a_1, \ldots, a_k, A_1,\ldots,A_\ell)$ to denote that $\varphi(x_1,\ldots, x_k, X_1,\ldots, X_\ell)$ holds in $\mathfrak{A}$ if, for every $i\in[k],$ the variable $x_i$ is interpreted as $a_i$ and, for every $j\in [\ell]$, the variable $X_j$ is interpreted as $A_j$.
If $\varphi$ is a sentence, then we write $\mathfrak{A}\models \varphi$ to denote that $\mathfrak{A}$ satisfies $\varphi$.
For every considered set of formulas $\Phi$, we assume that we have fixed an encoding
$\gamma:\Phi\to\{0,1\}^*$, and we set the \emph{length of $\varphi$}, denoted by $|\varphi|$, to be the length of $\gamma(\varphi)$. Note that this encoding takes into account the maximum value of $p$ and $k$, where ${\sf card}_p(X)$ and $\exists_k X$ are used in $\varphi$.

\subsection{Definition of \texorpdfstring{$\mathsf{CMSO/tw}$}{CMSO/tw}}\label{def_twMSO}

\paragraph{Annotated treewidth.}
We refer the reader to~\cite{ThilikosW23excl}, where annotated treewidth was first defined.
Let $G$ be a graph and $X\subseteq V(G)$.
We say that a graph $H$ is an \emph{$X$-rooted minor} of $G$ (or, simply, an \emph{$X$-minor of $G$})
if there is a collection $\mathcal{S}=\{S_{v}\mid v\in V(H))\}$ of pairwise disjoint subsets of $V(G),$ each inducing a connected subgraph of $G$ and containing at least one vertex of $X$, and such that, for every edge $xy\in E(H),$ the set $S_{x}\cup S_{y}$ is connected in $G.$ %
A graph $H$ is a \emph{minor} of $G$ if it is a $V(G)$-rooted minor of $G$.
Given a graph $G$ and a set $X\subseteq V(G),$ we define $\tw(G,X)$ as the maximum treewidth of an $X$-minor of $G$ (see~\Cref{subsec:findingflat} for a formal definition of treewidth).
Note that $\tw(G,V(G))=\tw(G)$ and that $\tw(G,X)≤\tw(G)$, for every $X\subseteq V(G)$.
See~\cite[Subsection~1.5]{ThilikosW23excl} for a comparison with other types of annotated extensions of treewidth.

\paragraph{Syntax and semantics of $\mathsf{CMSO/tw}$.}
The terms, (atomic) formulas, and free variables of $\mathsf{CMSO/tw}$ are defined as the ones of $\mathsf{CMSO}$
with the only exception that instead of using the quantifier $\exists X$, where $X$ is a set variable,
we use the quantifier $\exists_{k} X$ (semantically defined below) for some $k\in\mathbb{N}$.
Also, we use $\forall_k X \varphi$ as an abbreviation of $\neg (\exists_k X \neg \varphi)$.

The semantics of $\mathsf{CMSO/tw}$ are as the ones of $\mathsf{CMSO}$, where for a given vocabulary $\sigma$,
and $\varphi\in\mathsf{CMSO/tw}[\sigma]$,
a $\sigma$-structure $\mathfrak{G}$ satisfies
$\exists_{k}X \varphi$, denoted by $\mathfrak{G}\models \exists_{k}X \varphi$, if and only if
there is an $S\subseteq V(\mathfrak{G})$ such that
$\tw(\mathsf{Gaifman}(\mathfrak{G}),S)\leq k$ and
$\mathfrak{G}$ satisfies $\varphi$ when $X$ is interpreted as $S$.

\paragraph{The logic \texorpdfstring{$\mathsf{CMSO/tw}\!+\!\mathsf{dp}$}{CMSO/tw+dp}.}
For every positive integer $k$,
we define the $2k$-ary \emph{disjoint-paths predicate} ${\sf dp}_k(x_1, y_1, \ldots, x_k, y_k)$,
where $x_1, y_1, \ldots, x_k, y_k$ are first-order terms.
We use ${\sf dp}$ instead of ${\sf dp}_k$ when $k$ is clear from the context.
A formula of $\mathsf{CMSO/tw}\!+\!\mathsf{dp}$
is obtained by enhancing the syntax of $\mathsf{CMSO/tw}$ by allowing atomic formulas of the form ${\sf dp}_k(x_1, y_1, \ldots, x_k, y_k)$ on first-order terms $x_1, y_1, \ldots, x_k, y_k$, for $k\geq 1$.
The satisfaction relation between a $\sigma$-structure $\mathfrak{G}$ and formulas $\varphi$ of $\mathsf{CMSO/tw}\!+\!\mathsf{dp}[\sigma]$ is as for $\mathsf{CMSO/tw}$,
where ${\sf dp}_k(x_1, y_1, \ldots, x_k, y_k)$ evaluates true in $\mathfrak{G}$ if and only if
there are paths
$P_1,\ldots, P_k$ of $\mathsf{Gaifman}(\mathfrak{G})$
between (the interpretations of) $x_i$ and $y_i$ for all $i\in[k]$ such that for every distinct $i,j\in[k]$, $V(P_i)\cap V(P_j)=\emptyset.$
It is easy to see that ${\sf dp}_k$ is $\mathsf{MSO}_1$-expressible.
Also, note that $\tw(G,X)$ is $\mathsf{MSO}_1$-expressible via the obstruction characterization of treewidth and hence $\mathsf{CMSO/tw}\!+\!\mathsf{dp}$ is ``contained'' in $\mathsf{CMSO}$; see~\Cref{sec_gametrees}.
{Let us also stress that the problem of deciding whether a given formula in \textsf{CMSO} has an equivalent formula in $\mathsf{CMSO/tw}\!+\!\mathsf{dp}$ is undecidable.}

Another feature of our proofs is that they may produce AMTs for more expressive extensions of $\mathsf{CMSO/tw}$. For this, we may enhance $\mathsf{CMSO/tw}$
with more general predicates of the form ${\sf dp}^+_k(Z,x_1, y_1, \ldots, x_k, y_k)$, where
$Z$ is a vertex set variable and $x_1, y_1, \ldots, x_k, y_k$ are vertex variables, that evaluates true in $\mathfrak{G}$ if and only if
there are paths
$P_1,\ldots, P_k$ in $\mathsf{Gaifman}(\mathfrak{G})$
between (the interpretations of) $x_i$ and $y_i$ for all $i\in[k]$ such that for every distinct $i,j\in[k]$, $V(P_i)\cap V(P_j)=\emptyset,$ and additionally none of these paths contains any of the vertices in the interpretation  of $Z$. As we argue in \autoref{sec_conclusions}, minor modifications
 of our proofs can also provide a proof of \Cref{main_theorem_intro}
 for $\mathsf{CMSO/tw}\!+\!\mathsf{dp}^+$ (see the conclusion section for other possible extensions).

\subsection{A framework for the derivation of AMTs}
\label{reduct_sec}

Given a vocabulary $\sigma$ and a set $\Phi$ of formulas of vocabulary $\sigma$,
we say that two $\sigma$-structures $\mathfrak{G}_1$ and $\mathfrak{G}_2$ are \emph{$\Phi$-equivalent}, and we denote it by
$\mathfrak{G}_1 \equiv_{\Phi} \mathfrak{G}_2$,
if for every $\varphi\in\Phi$,
$\mathfrak{G}_1\models \varphi \iff \mathfrak{G}_2 \models \varphi$.

Let $\sigma$ be a vocabulary, let $R_1,\ldots, R_{h},h\in\mathbb{N},$ be a collection of unary relation symbols not contained in $\sigma$.
Also, let $\Phi$ be a set of formulas of vocabulary $\sigma \cup\{R_1,\ldots,R_h\}$.
Given a $\sigma$-structure $\mathfrak{G}$ and
sets $R_1,\ldots,R_h\subseteq V(\mathfrak{G})$,
we say that a tuple $(S_0,\ldots,S_h)$ where $S_0,\ldots,S_h\subseteq V(\mathfrak{G})$
is \emph{$\Phi$-irrelevant} in $(\mathfrak{G},R_1,\ldots,R_h)$
if for every $i\in[h]$, $R_i\setminus S_i \subseteq V(\mathfrak{G})\setminus S_0$ and
$$(\mathfrak{G}\setminus S_0, R_1\setminus S_1,\ldots, R_h\setminus S_h)
\equiv_{\Phi}
(\mathfrak{G},R_1,\ldots,R_h).$$
Note that, given a $\Phi\subseteq \mathsf{L}[\sigma]$ such that $\mathsf{FO}\subseteq \mathsf{L}$,
$\Phi$-irrelevancy is $\mathsf{L}[\sigma]$-expressible; see~\Cref{obs_irrformula}.

\vspace{-2mm}
\paragraph{Reductions between meta-theorems.}
A \emph{logic-parameter pair}, in short  \emph{an LP-pair}, is a pair $\mathbf{M}=(\mathsf{L},\p)$, where $\mathsf{L}$ is a logic
and $\p$ is a graph parameter, mapping graphs to non-negative integers.

\begin{importantdef}\label{def_red}
Given a positive integer $c$, and two LP-pairs  $\mathbf{M}_1=(\mathsf{L}_1,\p_1)$ and $\mathbf{M}_2=(\mathsf{L}_2,\p_2)$,
where for every graph $G$, $\p_2(G) = \mathcal{O}(|G|)$,
we say that
\emph{$\mathbf{M}_1$ is $\mathsf{n}^{c}$-reducible to $\mathbf{M}_2$}, denoted by $\mathbf{M}_1 \leq_{\mathsf{n}^{c}} \mathbf{M}_2$, if the following conditions are satisfied:
\begin{description}
\item \textbf{(Definability Condition)}
For every $d\in\mathbb{N}$ and every formula $\varphi\in\mathsf{L}_{1}[\{E\}]$ of length $d$,
there is an $h_d\in\mathbb{N}$, a set $\Phi_{d}\subseteq \mathsf{L}_2[\{E\}]$, and a formula $\delta_{\varphi}\in \mathsf{L}_2[\{E,R_1,\ldots,R_{h_d}\}]$ such that for every
two graphs $G,G'$ and every $h_d$-tuple $\bar{R}$ of subsets of $V(G')$,
\[ \text{if $(G,V(G)^{h_d})\equiv_{\Phi_{d}} (G',\bar{R})$, then $G\models \varphi \iff (G',\bar{R})\models \delta_{\varphi}$}.
\]
Moreover, $h_d,\Psi_d,$ and $\delta_\varphi$ are effectively computable from $d$ and $\varphi$.
\item \textbf{(Local-Global-Irrelevancy Condition)}
There is a colored-graph vocabulary $\sigma_{d,l}$ and a finite set
$\Psi_{d,l}\subseteq \mathsf{L}_2 [\sigma_{d,l}\cup\{R_1,\ldots,R_{h_d}\}],$ both effectively computable for given $d,l\in\mathbb{N}$,
and
there is a function $f:\mathbb{N}^2\to\mathbb{N}$ and an algorithm $\mathbb{A}$ that, given a $d\in\mathbb{N}$, a graph $G$, where $l:=\p_1(G)$,
and an $h_d$-tuple $\bar{R}$ of subsets of $V(G)$,
either reports that $\p_{2}(G) \leq f(d,l)$,
or outputs a set $Y\subseteq V(G)$, a $\sigma_{d,l}$-structure $\mathfrak{B}$, where $Y\subseteq V(\mathfrak{B})$,
an $h_d$-tuple $\bar{R}^\mathfrak{B}$ of subsets of $V(\mathfrak{B})$, and a set $Z\subseteq Y$ such that
\begin{description}
\item (a) $|\mathfrak{B}| = \mathcal{O}(|Y|)$
and $\p_{2}(\mathsf{Gaifman}(\mathfrak{B}))\leq f(d, l)$,

\item (b) for every vertex subsets $S_0,\ldots, S_{h_d}\subseteq Z$
it holds that
\begin{eqnarray}
& & \text{if $(S_0,\ldots,S_{h_d})$ is $\Psi_{d,l}$-irrelevant in $(\mathfrak{B},\bar{R}^\mathfrak{B})$,}\notag
\\
& & \text{then $(S_0,\ldots,S_{h_d})$ is
$\Phi_{d}$-irrelevant in $(G,\bar{R})$.}\notag
\end{eqnarray}
\vspace{-7mm}

\item (c) there exist vertex subsets $S_0,\ldots, S_{h_d}\subseteq Z$,
with $S_0\neq\emptyset$, such that $(S_0,\ldots,S_{h_d})$ is $\Psi_{d,l}$-irrelevant in $(\mathfrak{B},\bar{R}^\mathfrak{B})$.
\end{description}
Moreover, algorithm $\mathbb{A}$ runs in time $\mathcal{O}_{d,l}(|G|^{c})$.
\end{description}
\end{importantdef}

\paragraph{Decision and evaluation AMTs.}
Let $\mathsf{L}$ be a logic and $\sigma$ be a vocabulary.
The \textsc{Model Checking} problem for $\mathsf{L}$
asks, given as input a $\sigma$-structure $\mathfrak{G}$ and a $φ\in\mathsf{L}[\sigma]$, whether $\mathfrak{G}\models φ$.
The  \textsc{Evaluation} problem for $\mathsf{L}$ asks, given as input a $\sigma$-structure $\mathfrak{G}$ and a $\varphi(\bar{X})\in\mathsf{L}[\sigma],$
for a tuple $\bar{S}$ such that $\mathfrak{G}\models \varphi(\bar{S})$,
or a report that no such a tuple exists.
Given an integer $c\geq 1$, we say that
an LP-pair $(\mathsf{L},\p)$ is
an $\mathsf{n}^{c}$-\emph{\FPT-decision} (resp.
\emph{evaluation}) Algorithmic Meta-Theorem,
in short, an $\mathsf{n}^{c}$-\emph{\FPT-decision AMT}
(\emph{resp.~an} $\mathsf{n}^{c}$-\emph{\FPT-evaluation AMT}), if there is a computable function $f:\mathbb{N}^2\to\mathbb{N}$ and an algorithm solving the \textsc{Model Checking} (resp. \textsc{Evaluation}) problem for $\mathsf{L}$ in time $f(|φ|,\p(\mathsf{Gaifman}(\mathfrak{G})))\cdot |\mathfrak{G}|^{c}$, where $(φ,\mathfrak{G})$ is the input instance.

\paragraph{Our results.} We prove the following two results.
Given $p\in\mathbb{N}$, we set $\mathbb{N}_{\geq p}:=\mathbb{N}\setminus[0,p-1]$.
\begin{result}\label{thm_reduce}
If $c,c'\in\mathbb{N}_{\ge 1}$ and  $(\mathsf{L}_1,\p_1)$ and   $(\mathsf{L}_2,\p_2)$ are  LP-pairs, where
$(\mathsf{L}_1,\p_1)\leq_{\mathsf{n}^{c'}}  (\mathsf{L}_2,\p_2)$
and $(\mathsf{L}_2,\p_2)$ is an $\mathsf{n}^{c}$-\FPT-evaluation AMT, then $(\mathsf{L}_1,\p_1)$
is an $\mathsf{n}^{\max\{c,c'\}+1}$-\FPT-decision AMT.
\end{result}

\begin{result}\label{thm_CMSO/tw_MSO}
$(\mathsf{CMSO/tw}\!+\!\mathsf{dp},\mathsf{hw})\leq_{\mathsf{n}^1}  (\mathsf{CMSO},\mathsf{tw})$.
\end{result}

Since
$(\mathsf{CMSO},\mathsf{tw})$ is an $\mathsf{n}$-\FPT-evaluation AMT~\cite[Theorem 2.10]{CourcelleM93} (see also~\cite{ArnborgLS91easy} and~\cite{FlumFG02quer}),
\Cref{thm_reduce,thm_CMSO/tw_MSO},
imply that $(\mathsf{CMSO/tw}\!+\!\mathsf{dp},\mathsf{hw})$ is an $\mathsf{n}^{2}$-$\FPT$-decision AMT (\Cref{main_theorem_intro}).

\subsection{Proof of \texorpdfstring{\Cref{thm_reduce}}{Theorem 4}}
In this subsection, we prove~\Cref{thm_reduce}. We start with the observation that, given a logic $\mathsf{L}$ such that $\mathsf{FO}\subseteq \mathsf{L}$ and a $\Phi\subseteq \mathsf{L}[\sigma]$,
$\Phi$-irrelevancy is $\mathsf{L}[\sigma]$-expressible.

\begin{observation}\label{obs_irrformula}
Let $\mathsf{L}$ be a logic such that $\mathsf{FO}\subseteq\mathsf{L}$,
let $h\in\mathbb{N}$, let $\sigma$ be a vocabulary, let $R_1,\ldots, R_h$ be a collection of unary relation symbols,
and let $\Phi\subseteq\mathsf{L}[\sigma \cup\{R_1,\ldots, R_h\}]$.
There is a formula $\gamma_{\Phi\text{-}\mathsf{irr}}(X_0,X_1,\ldots,X_{h})$
in $\mathsf{L}[\sigma\cup \{R_1,\ldots,R_h\}]$ such that
for every $\sigma$-structure $\mathfrak{G}$, every $\bar{R}\in (2^{V(\mathfrak{G})})^h$, and every $S_0,\ldots,S_h\subseteq V(\mathfrak{G})$,
it holds that $(S_0,\ldots,S_h)$ is $\Phi$-irrelevant in $(\mathfrak{G},\bar{R})$ if and only if $(\mathfrak{G},\bar{R})\models\gamma_{\Phi\text{-}\mathsf{irr}}(S_0,\ldots,S_h)$.
Moreover, the length of $\gamma_{\Phi\text{-}\mathsf{irr}}$ depends only on $|\Phi|$.
\end{observation}

The proof of~\Cref{thm_reduce} is split into two parts.
First, assuming that we are given an annotated graph $(G,\bar{R})$,
where $\p_2(G)$ is ``large'',
we show how to use the \textsl{Local-Global-Irrelevancy Condition} in order to obtain an annotated graph $(G',\bar{R}')$ such that
$\p_2(G')$ is ``small'' and $(G,\bar{R})\equiv_{\Phi_{d}} (G',\bar{R}')$.

\begin{lemma}\label{lemma_irr}
Let
$c,c'\in\mathbb{N}_{\ge 1}$ and let $(\mathsf{L}_1,\p_1)$ and   $(\mathsf{L}_2,\p_2)$ be two LP-pairs, where
$(\mathsf{L}_1,\p_1)\leq_{\mathsf{n}^{c'}}  (\mathsf{L}_2,\p_2)$
and $(\mathsf{L}_2,\p_2)$ is an $\mathsf{n}^{c}$-{\FPT-evaluation AMT}.
There is an algorithm that given a $d\in\mathbb{N}$, a graph $G$,
where $l:=\p_1(G)$,
and an $h_d$-tuple  $\bar{R}$ of subsets of $V(G)$,
outputs
a subgraph $G'$ of $G$ and an $h_d$-tuple $\bar{R}'$ of subsets of $V(G')$,
such that $\p_2(G')\leq
f(d,l)$ and
$(G,\bar{R})\equiv_{\Psi_{d}} (G',\bar{R}')$.
Moreover, this algorithm runs in time
$\mathcal{O}_{d,\p_{1}(G)}(|G|^{\max\{c,c'\}+1})$.
\end{lemma}

\begin{proof}
Let  $\mathbb{A}$ be the algorithm that certifies that $(\mathsf{L}_1,\p_1)\leq_{\mathsf{n}^{c'}}  (\mathsf{L}_2,\p_2)$.
Also, let $\mathbb{B}$ be an algorithm certifying that
$(\mathsf{L}_2,\p_2)$ is an $\mathsf{n}^{c}$-{\FPT-evaluation AMT}.
We start by setting
$(G',\bar{R}'):=
(G,\bar{R})$.
We also set $(R_1',\ldots,R_{h_d}'):=\bar{R}'$.
The algorithm proceeds in two steps:
\begin{description}
\item[Step 1.] Apply algorithm $\mathbb{A}$ with input $d$ and $(G',\bar{R}')$.
If $\mathbb{A}$ reports that $\p_2(G') \leq f(d,l)$,
then stop.
If it outputs a set $Y\subseteq V(G)$, a $\sigma_{d,l}$-structure $\mathfrak{B}$, where $Y\subseteq V(\mathfrak{B})$,
an $h_d$-tuple $\bar{R}^\mathfrak{B}$ of subsets of $V(\mathfrak{B})$, and a set $Z\subseteq Y$ as in~\Cref{def_red},
then proceed to Step 2.

\item[Step 2.]
Let $\gamma_{\Psi\text{-}\mathsf{irr}}$ be the formula of~\Cref{obs_irrformula} for $\Psi =\Psi_{d,l}$.
Apply algorithm $\mathbb{B}$ with input $(\mathfrak{B},\bar{R}^\mathfrak{B})$ to find some $S_0,\ldots,S_{h_d}\subseteq Z$
where $S_0\neq \emptyset$ and $(S_0,\ldots,S_{h_d})$ is
$\Psi_{d,l}$-irrelevant in $(\mathfrak{B},\bar{R}^\mathfrak{B})$.
Set $(G',\bar{R}'):=(G'\setminus S_0,R_1'\setminus S_1,\ldots, R_{h_d}'\setminus S_{h_d})$ and repeat Step 1.
\end{description}
The correctness of the algorithm follows from the definition of $(S_0,\ldots,S_{h_d})$ being $\Psi_{d,l}$-irrelevant in $(G',\bar{R}')$.
The algorithm terminates because $\p_2(G) = \mathcal{O}(|G|)$.
For the running time, note that Step 1 is performed in time $\mathcal{O}_{d,\p_{1}(G)}(|G|^{c})$.
Also, $|\gamma_{\Psi\text{-}\mathsf{irr}}|$ depends only on $|\Psi_{d,l}|$,
which, in turn, depends only on $d$
and $l$ (which is equal to $\p_1(G)$).
Moreover, $\p_{2}(G')\leq
f(d,l)$ and $|\mathfrak{B}| = \mathcal{O}(|G'|)$.
Thus, Step 2 is performed in time $\mathcal{O}_{d,\p_1(G)}(|G|^{c'})$ and the overall running time is
$\mathcal{O}_{d,\p_1(G)}(|G|^{\max\{c,c'\}+1})$.
\end{proof}

\begin{proof}[Proof of~\Cref{thm_reduce}]
Let
$c,c'\in\mathbb{N}_{\ge 1}$ and let $(\mathsf{L}_1,\p_1)$ and   $(\mathsf{L}_2,\p_2)$ be two LP-pairs, where
$(\mathsf{L}_1,\p_1)\leq_{\mathsf{n}^{c'}}  (\mathsf{L}_2,\p_2)$.
Also,
let $\mathbb{B}$ be an algorithm certifying that $(\mathsf{L}_2,\p_2)$ is an $\mathsf{n}^{c}$-{\FPT-evaluation AMT}.
We present an algorithm that certifies that
$(\mathsf{L}_1,\p_1)$ is an
$\mathsf{n}^{\max\{c,c'\}+1}$-\FPT-decision AMT.
First, given a sentence $\varphi\in\mathsf{L}_1[\{E\}]$,
we compute the integer $h_{|\varphi|}$,
the set $\Phi_{|\varphi|}\subseteq \mathsf{L}_2[\{E\}]$ and the formula $\delta_\varphi\in\mathsf{L}_2[\{E,R_1,\ldots,R_{h_{|\varphi|}}\}]$.
Note that $|\delta_{\varphi}|$ depends only on $|\varphi|$.
To decide whether a given graph $G$ satisfies a given sentence $\varphi\in\mathsf{L}_1[\{E\}]$, we perform the following two steps.
\begin{description}

\item[Step 1.] Apply the algorithm of~\Cref{lemma_irr} with input $|\varphi|,G,(V(G))^{h_{|\varphi|}}$.
This algorithm outputs a subgraph $G'$ of $G$ and an $h_d$-tuple $\bar{R}'$ of subsets of $V(G')$,
such that $\p_2(G')\leq
f(|\varphi|,\p_1(G))$ and
$(G,V(G)^{h_{|\varphi|}})\equiv_{\Phi_{|\varphi|}} (G',\bar{R}')$.

\item[Step 2.]
Apply algorithm $\mathbb{B}$ with input $(G',\bar{R}')$, which decides whether $(G',\bar{R}')\models \delta_\varphi$ and therefore whether $G\models\varphi$.
\end{description}
Note that $G\models\varphi\iff(G',\bar{R}')\models \delta_\varphi$, because $(G,V(G)^{h_{|\varphi|}})\equiv_{\Phi_{|\varphi|}} (G',\bar{R}')$; see the \textsl{Definability Condition} of~\Cref{def_red}.
Step 1 is performed in time $\mathcal{O}_{|\varphi|,\p_{1}(G)}(|G|^{\max\{c,c'\}+1})$, while Step 2 in time
$\mathcal{O}_{|\varphi|,\p_{1}(G)}(|G|^{c'})$, and thus the overall running time is
$\mathcal{O}_{|\varphi|,\p_{1}(G)}(|G|^{\max\{c,c'\}+1})$.
\end{proof}

\section{Annotated types for \texorpdfstring{$\mathsf{CMSO/tw}\!+\!\mathsf{dp}$}{CMSO/tw+dp}}
\label{sec_gametrees}

The following sections are dedicated to the proof of~\Cref{thm_CMSO/tw_MSO}.
In this section, we show the \textsl{Definability Condition}
(see~\Cref{def_red})
for $(\mathsf{CMSO/tw}\!+\!\mathsf{dp},\mathsf{hw})\leq_{\mathsf{n}^1}  (\mathsf{CMSO},\mathsf{tw})$.
We start by showing that every formula of $\mathsf{CMSO/tw}\!+\!\mathsf{dp}$ can be expressed in $\mathsf{CMSO}$.

\begin{lemma}\label{lem_CMSO/twdp_MSO}
Let $\sigma$ be a vocabulary.
For every $\varphi\in\mathsf{CMSO/tw}\!+\!\mathsf{dp}[\sigma]$ there is a $\psi\in\mathsf{CMSO}[\sigma]$, where $|\psi|$ is
depends only on $|\varphi|$,
such that for every $\sigma$-structure $\mathfrak{G}$, it holds that
$\mathfrak{G}\models \varphi \iff \mathfrak{G}\models \psi.$
\end{lemma}

To prove~\Cref{lem_CMSO/twdp_MSO}, it suffices to show that the property of bounded annotated treewidth and the disjoint-paths predicate are $\mathsf{CMSO}$-expressible.
For the disjoint-paths predicate, it is easy to observe that it is $\mathsf{MSO}_1$-expressible.

\begin{observation}\label{lemma_dpMSO}
Let $k\in\mathbb{N}$. There is a $\xi_k(x_1,y_1,\ldots,x_k,y_k)\in\mathsf{MSO}_1[\{E\}]$ such that
for every graph $G$ and every $x_1,y_1,\ldots,x_k,y_k\in V(G)$,
$G\models\mathsf{dp}(x_1,y_1,\ldots,x_k,y_k)\iff G\models\xi_k(x_1,y_1,\ldots,x_k,y_k)$.
\end{observation}

Also, bounded annotated treewidth is definable in the weaker logic $\text{\textsf{FO}+\textsf{dp}}$; see~\cite{SchirrmacherSV23first,GolovachST22model}.
\Cref{lem_CMSO/twdp_MSO} follows from~\Cref{lemma_dpMSO} and~\Cref{formulaanntw}.

\begin{lemma}\label{formulaanntw}
For every $k\in\mathbb{N}$, there is a formula $\zeta_k\in\text{\emph{\textsf{FO}+\textsf{dp}}}$ such that
for every graph $G$ and every $X\subseteq V(G)$, it holds that
$\tw(G,X)\leq k$ if and only if $G\models \zeta_k(X)$.
\end{lemma}

\begin{proof}
First, observe that containment of a given graph $H$ as an $X$-rooted minor can be expressed in $\text{\textsf{FO}+\textsf{dp}}$.
To see this, note that
for every graph $H$, there is a finite family $\mathcal{F}_H$
of graphs (their size depends on $|H|$)
such that
$H$ is a $X$-rooted minor in $G$ if and only if
$G$ contains some $F\in\mathcal{F}$ as an \emph{$X$-rooted topological minor}, i.e., there is a set $V$ of $|F|$ vertices in $X$, a mapping $\mu:V\to V(F)$, and a collection $\mathcal{P}$ of disjoint paths in $G$ connecting the vertices in $V$ such that $u,v\in V$ are connected by a path in $\mathcal{P}$ if and only if $\mu(u)$ and $\mu(v)$ are adjacent in $F$.
Therefore, $X$-rooted topological minor (and, thus, $X$-rooted minor) containment is definable in $\text{\textsf{FO}+\textsf{dp}}$.
Then, note that $\tw(G,X)\leq k$ is characterized by forbidding a finite (depending on $k$) set of graphs of at most $f(k)$ vertices, for some function $f$, as $X$-minors; see~\cite[Theorem 5.9]{Lagergren87uppe}.
\end{proof}

\paragraph{Containment of $\tilde{\Theta}^{\sf dp}$ in $\mathsf{CMSO/tw}\!+\!\mathsf{dp}$.}
At this point, let us mention that it can be shown that every formula of $\tilde{\Theta}^{\sf dp}$ is definable in $\mathsf{CMSO/tw}\!+\!\mathsf{dp}$.
The essential part of such proof is to show that
existence of a path between two vertices is expressed in $\mathsf{CMSO/tw}\!+\!\mathsf{dp}$
and that
there is a formula $\psi_k\in \mathsf{CMSO/tw}\!+\!\mathsf{dp}$
such that for every $(G,X)$, $\mathsf{tw}(\mathsf{torso}(G,X))\leq k$ if and only if $G\models \psi_k (X)$;
see~\cite{FominGSST23,FominGSST21comp_arxiv} for a formal definition of $\tilde{\Theta}^{\sf dp}$ and $\mathsf{tw}(\mathsf{torso}(G,X))$.
The former is (trivially) expressible with the use of the disjoint-paths predicate, while definability of $\mathsf{tw}(\mathsf{torso}(G,X))\leq k$ in $\mathsf{CMSO/tw}\!+\!\mathsf{dp}$ follows from the fact that $\mathsf{tw}(G,X)\leq\mathsf{tw}(\mathsf{torso}(G,X))$; see~\cite[Observation 1.16]{ThilikosW23excl}.

\subsection{Atomic types for \texorpdfstring{$\mathsf{CMSO/tw}\!+\!\mathsf{dp}$}{CMSO/tw+dp}}
\label{subsec:atp}
Recall that, given a colored-graph vocabulary $\sigma$,
the atomic formulas of $\mathsf{CMSO/tw}\!+\!\mathsf{dp}[\sigma]$
are of the form
$x=y$, $\{x,y\}\in E$, $x\in X$,  ${\sf card}_p(X),$ and $\mathsf{dp}_k(x_1,y_1,\ldots,x_k,y_k)$ for $k\in\mathbb{N}$,
where $x,y,x_1,y_1,\ldots,x_k,y_k$ denote first-order terms and $X$ denotes either a set variable or a unary predicate symbol in $\sigma$.
Given $m,r\in\mathbb{N}$ and $p\in\mathbb{N}_{\geq 2}$,
we use $\mathsf{Atp}^{m,r}[\sigma]$ to denote the set of all formulas $\varphi(X_1,\ldots,X_m,x_1,\ldots,x_r)$ (up to logical equivalence)
that are boolean combinations of atomic formulas of $\mathsf{CMSO/tw}\!+\!\mathsf{dp}[\sigma]$ on the variables $X_1,\ldots,X_m,x_1,\ldots,x_r$.

\paragraph{The set $\mathcal{L}_{p,t}^{m,r}[\sigma]$.}
We use $\mathcal{L}_{p,t}^{m,r}[\sigma]$ to denote the set of all different (up to logical equivalence) formulas of the form
$\psi(\bar{X},\bar{x})\in \mathsf{Atp}^{m,r}[\sigma]$
such that for every $\sigma$-structure $\mathfrak{G}$, every $\bar{V}=(V_1, \ldots,V_m)\in (2^{V(\mathfrak{G})})^m$ and every $\bar{v}\in V(\mathfrak{G})^r$, if $\mathfrak{G}\models \psi(\bar{V},\bar{v})$ then
for each $i\in[m]$, if $\mathsf{card}_j(X_i)$ appears in $\varphi$, then $j\in[2,p]$,
and for each $i\in[r]$, there is some $t_i\in[0,t]$ such that $\tw(\mathfrak{G},V_i)\leq t_{i}$
(for this, alternatively, one may conjoin the formulas in $\mathsf{Atp}^{m,r}$ with the quantifier-free part of $\zeta_{t_i}(X_i)$ from~\Cref{formulaanntw}, for every $i\in[m]$).
Note that $\mathcal{L}_{p,t}^{m,r}[\sigma]\subseteq \mathsf{CMSO}[\sigma]$ and $|\mathcal{L}_{p,t}^{m,r}[\sigma]|$
depends (only) on $|\sigma|$, $m,$ $r,$ $p$, and $t$.

Given a $\varphi\in\mathsf{CMSO/tw}\!+\!\mathsf{dp}[\sigma]$,
we use $p(\varphi)$ (resp. $t(\varphi)$) to denote
the maximum $p$ (resp. $t$) such that a predicate of the form $\mathsf{card}_p(X)$ (resp. $\exists_t X$) appears in $\varphi$.
Using standard arguments for translating a given formula to one of prenex normal from, we can effectively translate every  $\mathsf{CMSO/tw}\!+\!\mathsf{dp}$-formula to a $\mathsf{CMSO}$-formula of a form like in the following statement.

\begin{lemma}\label{lem_trans_CMSO/tw_MSO}
Let $\sigma$ be a vocabulary.
For every $\varphi\in\mathsf{CMSO/tw}\!+\!\mathsf{dp}[\sigma]$,
there are $m,r\in\mathbb{N}$ and $\varphi'\in\mathsf{CMSO}[\sigma]$ such that
$\varphi'$ can be written as
\[Q_1 X_1\ldots Q_m X_m Q_{m+1}x_1\ldots Q_{m+r}x_r\ \psi(\bar{X},\bar{x}),\]
where for every $i\in[m+r]$ $Q_i \in \{\forall,\exists\}$ and
$\psi(\bar{X},\bar{x})\in \mathcal{L}_{p,t}^{m,r}[\sigma]$,
where $p:=p(\varphi)$ and $t:=t(\varphi)$,
and for every $\sigma$-structure $\mathfrak{G}$ it holds that $\mathfrak{G}\models \varphi \iff \mathfrak{G}\models \varphi'$.
Moreover, $m,r,$ and $\varphi'$ can be effectively computed from $\varphi,p,t$.
\end{lemma}

\paragraph{Some conventions on notation.}
Let us note that, throughout the paper,
given a set $A$ and elements $a_1,\ldots,a_r\in A$,
we use either $\bar{a}$ or $a_1a_2\ldots a_r$ to denote the tuple $(a_1,\ldots,a_r)$, whenever it is clear from the context, and $V(\bar{a})$ to denote the set $\{a_1,\ldots,a_r\}$.
Given a set $A$, we use $2^A$ to denote the set of all of its subsets and $A^c$ to denote the product $A\times \cdots \times A$ of $c$ copies of $A$.
If ${\cal S}$ is a collection of objects where the operation $\cup$ is defined, then we denote $\cupall \mathcal{S}=\bigcup_{X\in \mathcal{S}}X.$
Also, to facilitate reading, in the rest of the paper we assume that we always deal with formulas where $p(\varphi)$ is bounded by a fixed  constant $\hat{c}$.
For this reason, we use  $\mathcal{L}_{t}^{m,r}[\sigma]$ to denote $\mathcal{L}_{\hat{c},t}^{m,r}[\sigma]$.

\subsection{Annotated types for \texorpdfstring{$\mathsf{CMSO/tw}\!+\!\mathsf{dp}$}{CMSO/tw+dp}}\label{subsec:anntypes}

In this subsection we define \emph{annotated types}.
These correspond to the annotated variant of $h$-rank types, as defined e.g.,
by Shelah~\cite{Shelah75them}, Ebbinghaus and Flum~\cite[Definition~2.2.5]{EbbinghausF99fini}, and Libkin~\cite[Section~3.4]{Libkin04elem}.
We refer the reader to~\cite{BonnetKTW22twinI,gradel2003automata,GajarskyGK20diff,gajarsky16algorithmic,GolovachST22model,GolovachST22model_arXiv,SchirrmacherSSTV24mode}
for similar ``game-like'' definitions of the same notion.
The framework presented in this section is inspired from the one of~\cite{GolovachST22model_arXiv,GolovachST22model}.
In the following definition, we use $\mathsf{tp}$ as an abbreviation of \emph{annotated type} of $\mathsf{CMSO/tw}\!+\mathsf{dp}$.
Keep in mind that
$\mathcal{L}_{t}^{m,r}[\sigma]$ denotes $\mathcal{L}_{\hat{c},t}^{m,r}[\sigma]$, for some fixed $\hat{c}\in\mathbb{N}$.

\begin{importantdef}
\label{def:tp}
Let $\sigma$ be a colored-graph vocabulary and let $m,r,t\in\mathbb{N}$.
Let $\mathfrak{G}$ be a $\sigma$-structure and let $\bar{R} = (R_1,\ldots, R_{m+r})$ be a tuple of subsets of $V(\mathfrak{G})$.

For every $m$-tuple $\bar{V}\in (2^{V(\mathfrak{G})})^m$ and every $r$-tuple $\bar{v}\in {V(\mathfrak{G})}^r$,
we set
$$\mathsf{tp}_{m,r,t}^0(\mathfrak{G},\bar{R},\bar{V},\bar{v})
:=
\{\psi(\bar{X},\bar{x}) \in \mathcal{L}_{t}^{m,r}[\sigma]
\mid \mathfrak{G}\models \psi(\bar{V},\bar{v})\}.$$
For every $i\in[r]$, every $m$-tuple $\bar{V}\in (2^{V(\mathfrak{G})})^m$, and every  $(r-i)$-tuple $\bar{v}\in V(\mathfrak{G})^{r-i}$,
we set
$$\mathsf{tp}_{m,r,t}^i(\mathfrak{G},\bar{R},\bar{V},\bar{v}) := \{\mathsf{tp}_{m,r,t}^{i-1}(\mathfrak{G},\bar{R},\bar{V},\bar{v}u) \mid {u\in R}_{m+r-(i-1)}\}.$$
Also, for every $i\in[1+r,m+r]$ and every $(m+r-i)$-tuple $\bar{V}\in (2^{V(\mathfrak{G})})^{m+r-i}$,
we set
$$\mathsf{tp}_{m,r,t}^i(\mathfrak{G},\bar{R},\bar{V}) := \{\mathsf{tp}_{m,r,t}^{i-1}(\mathfrak{G},\bar{R},\bar{V}U) \mid U\subseteq R_{m+r-(i-1)}\}.$$
\end{importantdef}
Observe that for every $i\in[1+r, m+r]$ and every $\bar{V}\in (2^{V(\mathfrak{G})})^{m+r-i}$,
if $\mathsf{tp}_{m,r,t}^i(\mathfrak{G},\bar{R},\bar{V})$ contains $\mathsf{tp}_{m,r,t}^{i-1}(\mathfrak{G},\bar{R},\bar{V}U)$, then
$U\subseteq R_{m+r-(i-1)}$ and $U$ satisfies one of the formulas $\zeta_0,\ldots, \zeta_t$ from~\Cref{formulaanntw}.

\paragraph{Definability Condition for $\mathsf{CMSO/tw}\!+\!\mathsf{dp}$.}
The \textsl{Definability Condition}
(see~\Cref{def_red})
for $(\mathsf{CMSO/tw}\!+\!\mathsf{dp},\mathsf{hw})\leq_{\mathsf{n}^1}  (\mathsf{CMSO},\mathsf{tw})$ is implied by the two following results (\Cref{prop:patterns,lemgametreesCMSO/tw}).
Their proof is a reformulation of folklore results on $\mathsf{CMSO}$-types, written for annotated types of $\mathsf{CMSO/tw}\!+\!\mathsf{dp}$;
see~\cite[Lemma 2.1]{Shelah75them},~\cite[Definition 2.2.5]{EbbinghausF99fini},~\cite[Section 3.4]{Libkin04elem}, and~\cite[Subsection~4.2]{GajarskyGK20diff}. The first states that the number of different annotated types of $\mathsf{CMSO/tw}\!+\!\mathsf{dp}$ for fixed $m,r,t,\sigma$ is upper-bounded by a function of $m,r,t,|\sigma|$ and they can be defined in $\mathsf{CMSO}$ (the latter follows from~\Cref{lem_trans_CMSO/tw_MSO}).

\begin{lemma}\label{prop:patterns}
There is a function $\newfun{fun_ipatterns}:\mathbb{N}^4\to\mathbb{N}$ such that
for every colored-graph vocabulary $\sigma$ and every $m,r,t\in\mathbb{N}$,
it holds that the set
$$\mathcal{B}_{\sigma}^{m,r,t}:=\{
\mathsf{tp}_{m,r,t}^{m+r}(\mathfrak{G},\bar{R})\mid
\mathfrak{G}\in\mathbb{STR}[\sigma] \text{ and }
\bar{R}\in (2^{V(\mathfrak{G})})^{m+r}\}
$$ has size at most $\funref{fun_ipatterns}(|\sigma|,m,r,t)$.
Moreover,  if $P_1,\ldots,P_{c}$ is an enumeration of all sets in $\mathcal{B}_{\sigma}^{m,r,t}$,
then there is a set $\Phi_{m,r,t}:=\{\alpha_1,\ldots,\alpha_c\}\subseteq \mathsf{CMSO}[\sigma\cup\{R_1,\ldots, R_{m+r}\}]$
such that
for every $i\in[c]$, every $\sigma$-structure $\mathfrak{G}$, and every $\bar{R}\in (2^{V(\mathfrak{G})})^{m+r}$,
it holds that
$(\mathfrak{G},\bar{R})\models \alpha_i \iff \mathsf{tp}_{m,r,t}^{m+r}(\mathfrak{G},\bar{R}) = P_i.$
Moreover, $\Phi_{m,r,t}$ can be effectively computed from $m,r,$ and $t$.
\end{lemma}

By using~\Cref{lem_trans_CMSO/tw_MSO},
we can also prove the following result, formulated for the set $\Phi_{m,r,t}$ of~\Cref{prop:patterns}.

\begin{lemma}\label{lemgametreesCMSO/tw}
Let $\sigma$ be a colored-graph vocabulary.
For every sentence $\varphi\in\mathsf{CMSO/tw}\!+\!\mathsf{dp}[\sigma]$, where $t:=t(\varphi)$,
there are $m,r\in\mathbb{N}$ and a sentence $\delta_\varphi\in\mathsf{CMSO}[\sigma\cup\bar{R}]$ such that for every two $\sigma$-structures $\mathfrak{G},\mathfrak{G}'$ and every $\bar{R}\in (2^{V(\mathfrak{G}')})^{m+r}$,
\[
\text{if $(\mathfrak{G},(V(\mathfrak{G}))^{m+r}) \equiv_{\Phi_{m,r,t}} (\mathfrak{G}',\bar{R})$, then
$\mathfrak{G}\models \varphi \iff (\mathfrak{G}',\bar{R})\models \delta_\varphi.$}\]
Moreover, $m,r,$ and $\delta_\varphi$ can be effectively computed from $\varphi$, $p(\varphi)$, and $t(\varphi)$.
\end{lemma}

\section{Flat railed annuli}\label{sec_flatannuli}
In this section, we define flat railed annuli, a notion analogous to flat walls.
Definitions and results on this section follow the framework of~\cite{SauST24amor} (see also~\cite{KawarabayashiTW18anew,GolovachST22model_arXiv}).
Our aim is to provide a combinatorial background of the proof of
the \textsl{Local-Global-Irrelevancy Condition}
(see~\Cref{def_red})
for $(\mathsf{CMSO/tw}\!+\!\mathsf{dp},\mathsf{hw})\leq_{\mathsf{n}^1}  (\mathsf{CMSO},\mathsf{tw})$.

This section is organized as follows.
All necessary definitions of the flat railed annuli framework that we use in this paper,
including the definitions of \emph{paintings}, \emph{renditions},
\emph{(flat) railed annuli}, \emph{influences}, \emph{levelings}, and the property of~\emph{well-alignedness} are gradually presented in~\Cref{subsec:paint-rend,subsec:defflatra,subsec:influence,subsec_levelings}.
The analogous algorithmic variant of the \textsl{Flat Wall theorem}~\cite{SauST24amor,RobertsonS95XIII,KawarabayashiTW18anew,Chuzhoy15impr,GiannopoulouT13opti,KawarabayashiKR12thedis}
for flat railed annuli that we need in this paper is \Cref{prop:flatann} and is presented in~\Cref{subsec:findingflat}.
The last two subsections contain useful results on how disjoint paths and bounded annotated treewidth sets behave inside a flat railed annulus.
In particular, in~\Cref{subsec:linkages}, we define \emph{linkages} (i.e., collections of disjoint paths) and
present a result derived from~\cite{GolovachST22model_arXiv} (see~\Cref{lem_levelingpaths}) showing that linkages are preserved when ``leveling'' some strip of a flat railed annulus away from its terminals.
In~\Cref{subsec:buffer}, we show that inside a large enough flat railed annulus, bounded annotated treewidth sets always leave some strip of the annulus intact (\Cref{lemma:buffer}).

\subsection{Paintings and renditions}
\label{subsec:paint-rend}

We first present the notions of renditions and paintings, originating in the work of Robertson and Seymour~\cite{RobertsonS95XIII}.
The definitions presented here were introduced by Kawarabayashi, Thomas, and Wollan~\cite{KawarabayashiTW18anew} (see also~\cite{SauST24amor}).
In~\cite{KawarabayashiTW18anew,SauST24amor}, renditions and paintings are defined only on disks, while here we extend these definitions to annuli, as in~\cite{GolovachST22model_arXiv}.

A \emph{circle} is a subset of $\mathbb{R}^2$ homeomorphic to
 $\{(x,y)\in \mathbb{R}^2\mid x^2 + y^2 =1\}$.
A \emph{closed} (resp.
\emph{open}) \emph{disk} is a a subset of $\mathbb{R}^2$ homeomorphic to $\{(x,y)\in \mathbb{R}^2\mid x^2 + y^2 \leq 1\}$ (resp.
$\{(x,y)\in \mathbb{R}^2\mid x^2 + y^2 < 1\}$) and a \emph{closed annulus} is a subset of $\mathbb{R}^2$ homeomorphic to $\{(x,y)\in \mathbb{R}^2\mid 1 \leq x^2 + y^2 \leq 2\}$.
Let $X$ be a subset of $\mathbb{R}^2.$
We define the \emph{boundary} of $X$, which we denote by $\mathsf{bd}(X)$, as follows: if $X$ is an open set, then $\mathsf{bd}(X)$ is the set of all points $y$ in $\mathbb{R}^2$ such that every neighborhood of $y$ in $\mathbb{R}^2$ intersects both $X$ and $\mathbb{R}^2\setminus X$,
while if $X$ is a closed set then $\mathsf{bd}(X)$ is the set of all points $y\in X$ such that every neighborhood of $y$ in $\mathbb{R}^2$ intersects $\mathbb{R}^2\setminus X$.
The \emph{closure} of an open subset $X$ of $\mathbb{R}^2$,
which is denoted by $\bar{X}$, is the set $X\cup \mathsf{bd}(X)$.
If $X$ is a closed annulus then $\mathsf{bd}(X) = B_1\cup B_2$,
where $B_1,B_2$ are the two unique connected components
of $\mathsf{bd}(X)$,
and note that $B_1,B_2$ are two disjoint circles.
We call $B_1$ and $B_2$ \emph{boundaries} of $X$.

\paragraph{Paintings.}
Let $\Delta$ be a closed annulus.
A \emph{$\Delta$-painting} is a pair $\Gamma=(U,N)$
where $N$ is a finite set of points of $\Delta,$
$N \subseteq U \subseteq \Delta,$ and
$U \setminus  N$ has finitely many arcwise-connected components, called \emph{cells}, where, for every cell $c,$
the closure $\bar{c}$ of $c$ is a closed disk and
$|\tilde{c}|\leq 3,$ where $\tilde{c}:=\mathsf{bd}(c)\cap N.$

We use the notation $U(\Gamma) := U,$
$N(\Gamma) := N$, and denote the set of cells of $\Gamma$
by $\mathsf{cells}(\Gamma).$
For convenience, we may assume that each cell  of $\Gamma$ is an open disk of $\Delta.$
Note that, given a $\Delta$-painting $\Gamma,$
the pair $(N(\Gamma),\{\tilde{c}\mid c\in \mathsf{cells}(\Gamma)\})$  is a hypergraph whose hyperedges have cardinality at most three, and  $\Gamma$ can be seen as a plane embedding of this hypergraph in $\Delta.$

\paragraph{Renditions.}
Let $G$ be a graph, let $V_1,V_2$ be two subsets of $V(G)$,
and let $\Omega_1$ (resp.
$\Omega_2$) be a cyclic permutation of $V_1$ (resp.
$V_2$).
By an \emph{$(\Omega_1, \Omega_2)$-rendition} of $G$ we mean a triple $(\Gamma, \sigma, \pi),$
where
\begin{itemize}
	\item[(a)] $\Gamma$ is a $\Delta$-painting for some closed annulus $\Delta,$
	\item[(b)] $\pi: N(\Gamma)\to V(G)$ is an injection, and
	\item[(c)] $\sigma$ assigns to each cell $c \in  \mathsf{cells}(\Gamma)$ a subgraph $\sigma(c)$ of $G,$ such that
	      \begin{enumerate}
		      \item[(1)] $G=\bigcup_{c\in \mathsf{cells}(\Gamma)}\sigma(c),$
		      \item[(2)] for distinct $c, c' \in \mathsf{cells}(\Gamma),$ $\sigma(c)$ and $\sigma(c')$ are edge-disjoint,
		      \item[(3)] for every cell $c \in \mathsf{cells}(\Gamma),$ $\pi(\tilde{c}) \subseteq V(\sigma(c)),$
		      \item[(4)] for every cell $c \in \mathsf{cells}(\Gamma),$
		            $V(\sigma(c)) \cap \bigcup_{c' \in  \mathsf{cells}(\Gamma) \setminus  \{c\}} V(\sigma(c')) \subseteq \pi(\tilde{c}),$ and
		      \item[(5)]  $\pi(N(\Gamma)\cap \bd(\Delta))=V_1\cup V_2,$ such that, if $\mathsf{bd}(\Delta) = B_1\cup B_2$, then, for $i\in[1,2]$, the points in $N(\Gamma)\cap B_i$ appear in $B_i$ in the same ordering as their images, via $\pi,$ in $\Omega_i$.
	      \end{enumerate}
\end{itemize}
Moreover, we ask the additional property that for every $c\in \mathsf{cells}(\Gamma),$ every two vertices in $\pi(\tilde{c})$ belong to some path of $\sigma(c)$ (property (ii) of the tightness condition of~\cite{SauST24amor}) and because of~\cite[Lemma~3]{SauST24amor}, we can assume that all renditions we consider satisfy this additional property.

\subsection{Definition of flat railed annuli}
\label{subsec:defflatra}

\paragraph{Railed annuli.}
Let $G$ be a graph and let $p,q\in\mathbb{N}_{\geq 3}$.
A $(p,q)$-\emph{railed annulus} of $G$ is a pair $\mathcal{A} = (\mathcal{C},\mathcal{P})$ where $\mathcal{C} = [C_1,\ldots,C_p]$ is a sequence of pairwise vertex-disjoint cycles of $G$
and $\mathcal{P} = [P_1,...,P_q]$ is a sequence of pairwise vertex-disjoint paths in $G$ such that for every $(i,j) \in [p] \times [q]$, $C_i \cap P_j$ is a non-empty path,
which we denote by $P_{i,j}$.
We refer to the paths of $\mathcal{P}$ as the \emph{rails} of $\mathcal{A}$ and to the cycles of $\mathcal{C}$ as the \emph{cycles} of $\mathcal{A}$.
We use $V(\mathcal{A})$ to denote the vertex set $\bigcup_{i\in[p]}V(C_i)\cup \bigcup_{j\in[q]}V(P_j)$ and $E(\mathcal{A})$ to denote the edge set $\bigcup_{i\in[p]}E(C_i)\cup \bigcup_{j\in[q]}E(P_j).$
A \emph{railed annulus} of $G$ is a $(p,q)$-railed annulus for some $p,q\in \mathbb{N}_{\ge 3}$.

We always assume an ordering of the vertices of each $P$ in $\mathcal{P}$ according to the linear ordering where the vertices of $P\cap C_1$ appear before the vertices of $P\cap C_p$.
For every $P\in\mathcal{P}$, we denote by $s_P$ (resp.
$t_P$) the endpoint of $P$ that appears first (resp.
last) in this linear ordering.

\paragraph{Flat railed annuli.}
Let $G$ be a graph.
Let $\mathcal{A} = (\mathcal{C},\mathcal{P})$ be a $(p,q)$-railed annulus of $G,$ for some $p,q\in \mathbb{N}_{\geq 3}$.
We say that $\mathcal{A}$ is a \emph{flat  $(p,q)$-railed annulus}
of $G$ if there exist two laminar\footnote{Given a graph $G$, a pair $(X,Y)$ of subsets of $V(G)$ is a \emph{separation} of $G$ if $X\neq Y$, $X\cup Y=V(G)$, and there is no edge in $G$ between $X\setminus Y$ and $Y\setminus X$.
We say that two separations $(X_1, Y_1)$ and $(X_2, Y_2)$ of $G$ are \emph{laminar} if either $Y_1 \subseteq Y_2$ and $X_2 \subseteq X_1$ or $Y_2 \subseteq Y_1$ and $X_1 \subseteq X_2$.} separations $(X_1,Y_1),(X_2,Y_2)$ of $G$, a non-empty set $Z_1\subseteq V(C_1) \cap V(\cupall\mathcal{P})$, and a non-empty set $Z_2\subseteq V(C_p) \cap V(\cupall \mathcal{P})$ such that:
\begin{itemize}
	\item $V(\mathcal{A})\subseteq Y_1\cap X_2,$
	\item  $Z_1\subseteq X_1\cap Y_1\subseteq V(C_1),$
	\item  $Z_2\subseteq X_2\cap Y_2\subseteq V(C_p),$ and
	\item if $\Omega_1$ (resp.
$\Omega_2$) is the cyclic ordering of the vertices $X_1\cap Y_1$ (resp.
$X_2\cap Y_2$) as they appear in the cyclic ordering of $C_1$ (resp.
$C_p$),
then there is an $(\Omega_1,\Omega_2)$-rendition $(\Gamma, \sigma, \pi)$ of $G[Y_1\cap X_2]$.
\end{itemize}
We say that $\mathcal{A}$ is a \emph{flat railed annulus}
of $G$ if it is a flat $(p,q)$-railed annulus for some $p,q\in\mathbb{N}_{\geq 3}$.

\paragraph{Railed annuli flatness pairs.}
Given the above, we  say that  the choice of the 9-tuple $\mathfrak{R}=(X_1,Y_1,X_2,Y_2,Z_1,Z_2,\Gamma,\sigma,\pi)$
\emph{certifies that $\mathcal{A}$ is a flat railed annulus of $G$}.
We call the pair $(\mathcal{A},\mathfrak{R})$ a \emph{railed annulus flatness pair} of $G$.
We use the term \emph{cell of} $\mathfrak{R}$ in order to refer to the cells of $\Gamma.$

\begin{figure}[ht]
\begin{center}
\includegraphics[scale=1]{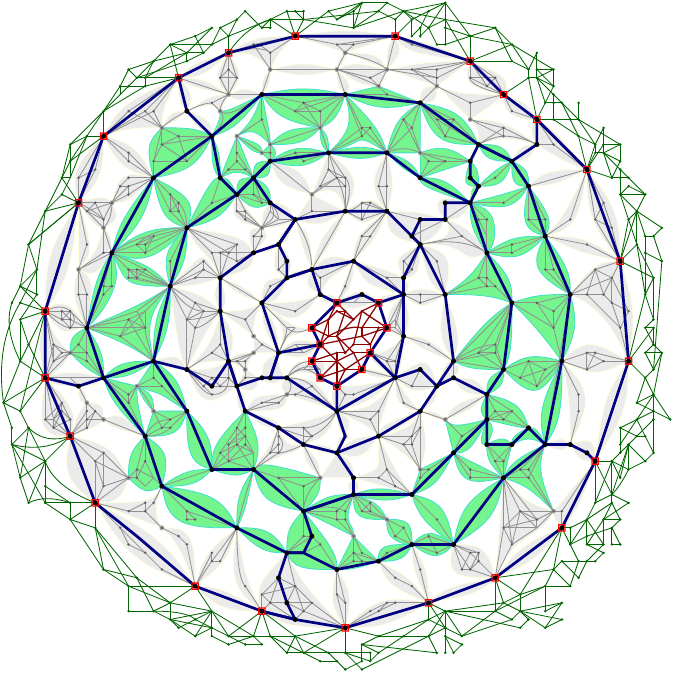}
\end{center}
\caption{A graph $G$ and a $(6,5)$-railed annulus flatness pair $(\mathcal{A},\mathfrak{R})$ of $G$.
The $(6,5)$-railed annulus $\mathcal{A}$ is depicted in deep blue, where we assume that the ``outer'' cycle in the figure is $C_1$ and the ``inner'' cycle is $C_6$.
The vertices in $X_1\cap Y_1$ and $X_2\cap Y_2$ are the ``outer'' and the ``inner'' red squared vertices, respectively.
The green vertices are the vertices in $X_1\setminus Y_1$ and the red vertices are the vertices in $Y_2\setminus X_2$.
The grey and green disks are the cells of $\mathfrak{R}$ and $\mathsf{Influence}_\mathfrak{R}([C_2,C_3])$ is the set of all graphs drawn in the green cells.
}\label{fig:railedann}
\end{figure}

We call the graph $G[Y_1\cap X_2]$
the \emph{$\mathfrak{R}$-compass} of $\mathcal{A}$ in $G,$
denoted by ${\sf Compass}_{\mathfrak{R}}(\mathcal{A}).$
We also use $\mathsf{inn}_\mathfrak{R}(\mathcal{A})$ to denote the set $Y_2$.
It is easy to see that there is a connected component of ${\sf Compass}_{\mathfrak{R}}(\mathcal{A})$ that contains the cycles and the paths of $\mathcal{A}$ as subgraphs.
We can assume that ${\sf Compass}_{\mathfrak{R}} (\mathcal{A})$ is connected, updating $\mathfrak{R}$ by removing from $Y_1\cap X_2$ the vertices of all the connected components of ${\sf Compass}_\mathfrak{R} (\mathcal{A})$
except for the one that contains $\mathcal{A}$, and including them in $X_1$ ($\Gamma, \sigma, \pi$ can also be easily modified according to the removal of the aforementioned vertices from $X_1\cap Y_2$).
We define the  \emph{flaps} of the railed annulus $\mathcal{A}$ in $\mathfrak{R}$ as
${\sf Flaps}_{\mathfrak{R}}(\mathcal{A}):=\{\sigma(c)\mid c\in \mathsf{cells}(\Gamma)\}.$
Given a flap $F\in {\sf Flaps}_{\mathfrak{R}}(\mathcal{A}),$ we define its \emph{base}
as $\partial F:=V(F)\cap \pi(N(\Gamma)).$

\subsection{Influence of sequences of cycles in flat railed annuli}
\label{subsec:influence}

Let $G$ be a graph and let $(\mathcal{A},\mathfrak{R})$ be a railed annulus flatness pair of $G$.
Given a cycle $C$ of $\mathsf{Compass}_{\mathfrak{R}}(\mathcal{A}),$ we say that
$C$ is \emph{$\mathfrak{R}$-normal} if it is \textsl{not} a subgraph of a flap $F\in \mathsf{Flaps}_{\mathfrak{R}}(\mathcal{A}).$
Note that every cycle $C$ of $G$ whose edge set is a subset of $E(\mathcal{A})$ is $\mathfrak{R}$-normal, since it has at least four vertices and $\mathcal{A}$ contains at least four disjoint paths connecting these vertices to the ``first'' and ``last'' cycles $C_1,C_p$ of $\mathcal{A}$.

Given an $\mathfrak{R}$-normal cycle $C$ of $\mathsf{Compass}_{\mathfrak{R}}(\mathcal{A}),$
we call a cell $c$ of $\mathfrak{R}$ \emph{$C$-perimetric} if
$\sigma(c)$ contains some edge of $C.$
We denote by $\Delta$ the closed annulus in the $(\Omega_1,\Omega_2)$-rendition of $\mathsf{Compass}_{\mathfrak{R}}(\mathcal{A}).$
A cell $c$ of $\mathfrak{R}$ is called \emph{$C$-external} if
it is contained in an arcwise connected component of $\Delta\setminus \cupall\{c\in \mathsf{cells}(\Gamma)\mid \mbox{$c$ is a $C$-perimetric cell of $\mathfrak{R}$}\}$
that intersects the boundary of $\Delta$ that contains $V(C_1)$.
A cell $c$ of $\mathfrak{R}$ is called \emph{$C$-internal} if it is neither $C$-perimetric nor $C$-external.
Given an $\mathfrak{R}$-normal cycle $C$ of $\mathsf{Compass}_{\mathfrak{R}}(\mathcal{A})$ we define the set
$$\mathsf{Influence}_{\mathfrak{R}}(C)=\{\sigma(c)\mid \mbox{$c$ is a cell of $\mathfrak{R}$ that is either $C$-perimetric or $C$-internal}\}.$$
Given a subsequence $\mathcal{C}' =[C_1',\ldots,C_r']$ of the cycles of $\mathcal{A}$,
we define $\mathsf{Influence}_{\mathfrak{R}}(\mathcal{C}')$
as the set $\mathsf{Influence}_{\mathfrak{R}}(C_1')\setminus \{\sigma(c)\mid \text{$c$ is a cell of $\mathfrak{R}$ that is $C_r'$-internal}\}$. See~\Cref{fig:railedann}.

\subsection{Levelings and well-aligned flat railed annuli}
\label{subsec_levelings}

Let $G$ be a graph and let $(\mathcal{A},\mathfrak{R})$ be a railed annulus flatness pair of $G.$
We define
the \emph{ground set} of $\mathsf{Compass}_{\mathfrak{R}}(\mathcal{A})$
to be the set ${\sf ground}_{\mathfrak{R}}(\mathcal{A}):=\bigcup_{F\in\mathsf{Flaps}_{\mathfrak{R}}(\mathcal{A})} \partial F$
and we refer to the vertices of this set as the \emph{ground vertices} of  $\mathsf{Compass}_{\mathfrak{R}}(\mathcal{A})$ in $G.$
Similarly, given a subsequence $\mathcal{C}'$
of the cycles of $\mathcal{A}$,
we define the \emph{ground set} of $\mathsf{Influence}_{\mathfrak{R}}(\mathcal{C}')$
to be the set ${\sf ground}_{\mathfrak{R}}(\mathcal{C}'):=\bigcup_{F\in\mathsf{Influence}_{\mathfrak{R}}(\mathcal{C}')} \partial F$,
and we refer to the vertices of this set as the \emph{ground vertices} of  $\mathsf{Influence}_{\mathfrak{R}}(\mathcal{C}')$ in $G.$
Note  that ${\sf ground}_{\mathfrak{R}}(\mathcal{A})$ and ${\sf ground}_{\mathfrak{R}}(\mathcal{C}')$ may  contain vertices
of ${\sf Compass}_{\mathfrak{R}}(\mathcal{A})$ that are not necessarily vertices in $V(\mathcal{A}).$

\paragraph{Levelings.}
We define  the $\mathfrak{R}$-\emph{leveling} of $\mathcal{A}$ in $G,$
denoted by ${\mathcal{A}}_{\mathfrak{R}},$ as the bipartite graph
where  one part is the ground set of $\mathsf{Compass}_{\mathfrak{R}}(\mathcal{A})$, the  other part is a set ${\sf vflaps}_{\mathfrak{R}}(\mathcal{A})=\{v_{F}\mid F\in\mathsf{Flaps}_{\mathfrak{R}}(\mathcal{A})\}$ containing one new vertex $v_{F}$ for each flap  $F$ in $\mathsf{Flaps}_{\mathfrak{R}}(\mathcal{A})$,
and, given  a pair $(x,F)\in {\sf ground}_{\mathfrak{R}}(\mathcal{A})\times \mathsf{Flaps}_{\mathfrak{R}}(\mathcal{A}),$
the set $\{x,v_F\}$ is an edge of $\mathcal{A}_\mathfrak{R}$ if and only if
$x\in \partial F.$
We call the vertices of ${\sf ground}_{\mathfrak{R}}(\mathcal{A})$ (resp.
${\sf vflaps}_{\mathfrak{R}}(\mathcal{A})$) \emph{ground-vertices} (resp.
\emph{flap-vertices}) of $\mathcal{A}_{\mathfrak{R}}.$

Also, given a subsequence $\mathcal{C}'$ of the cycles of $\mathcal{A}$, we define  the $(\mathfrak{R},\mathcal{C}')$-\emph{leveling} of $\mathcal{A}$ in $G,$
denoted by ${\mathcal{A}}_{\mathfrak{R}}^{\langle \mathcal{C}'\rangle},$ as the bipartite graph
where  one part is the ground set of $\mathsf{Influence}_{\mathfrak{R}}(\mathcal{C}')$, the  other part is the set ${\sf vflaps}_{\mathfrak{R}}(\mathcal{C}')=\{v_{F}\mid F\in\mathsf{Influence}_{\mathfrak{R}}(\mathcal{C}')\}$ containing one new vertex $v_{F}$ for each flap  $F$ in $\mathsf{Influence}_{\mathfrak{R}}(\mathcal{C}')$,
and, given  a pair $(x,F)\in {\sf ground}_{\mathfrak{R}}(\mathcal{C}')\times \mathsf{Influence}_{\mathfrak{R}}(\mathcal{C}'),$
the set $\{x,v_F\}$ is an edge of ${\mathcal{A}}_{\mathfrak{R}}^{\langle\mathcal{C}'\rangle}$ if and only if
$x\in \partial F.$
We call the vertices of ${\sf ground}_{\mathfrak{R}}(\mathcal{C}')$ (resp.
${\sf vflaps}_{\mathfrak{R}}(\mathcal{C}')$) \emph{ground-vertices} (resp.
\emph{flap-vertices}) of $\mathcal{A}_{\mathfrak{R}}^{\langle\mathcal{C}'\rangle}.$
Note that  $\mathcal{A}_{\mathfrak{R}}^{\langle\mathcal{C}'\rangle}=\mathcal{A}_\mathfrak{R}[\mathsf{Influence}_{\mathfrak{R}}(\mathcal{C}') \cup {\sf vflaps}_{\mathfrak{R}}(\mathcal{C}')]$.

\paragraph{Well-aligned railed annulus flatness pairs.}
An edge $e$ of ${\sf Compass}_{\mathfrak{R}}(\mathcal{A})$ is \emph{short} if there is some flap $F\in\mathsf{Flaps}_\mathfrak{R}(\mathcal{A})$ whose vertex set is the endpoints of $e$, and $e$ is its only edge.
{We denote by $\mathcal{A}^{\bullet}$ the  graph obtained from $\mathcal{A}$ if we subdivide \emph{once} every
edge in $E(\mathcal{A})$ that is short in ${\sf Compass}_{\mathfrak{R}}(\mathcal{A}).$}
The graph $\mathcal{A}^\bullet$ is a ``slightly richer variant'' of $\mathcal{A}$  that is necessary for our definitions and  proofs, namely to be able to associate  every flap-vertex of  an appropriate subgraph of $\mathcal{A}_{\mathfrak{R}}$ (that we will denote by $R_{\mathcal{A}}$) with  a non-empty path of $\mathcal{A}^\bullet,$ as we proceed to formalize.
We say that $(\mathcal{A},\mathfrak{R})$ is \emph{well-aligned} if the following holds:
\begin{quote}
	$\mathcal{A}_{\mathfrak{R}}$ contains as a subgraph a railed annulus $R_\mathcal{A}$
	such that $\mathcal{A}^{\bullet}$ is isomorphic to some subdivision of $R_\mathcal{A}$
	via an isomorphism that maps each ground vertex to itself.
\end{quote}
Suppose now that the railed annulus flatness pair $(\mathcal{A},\mathfrak{R})$ is well-aligned.
We call the wall  $R_\mathcal{A}$ in the above condition  a \emph{representation} of $\mathcal{A}$ in $\mathcal{A}_{\mathfrak{R}}.$ Note that, as $R_\mathcal{A}$ is a subgraph of $\mathcal{A}_{\mathfrak{R}},$ it is bipartite as well.
The above property gives us a way to represent a flat railed annulus by a railed annulus of its leveling,
in a way that ground vertices are not altered.
Note that both $\mathcal{A}_{\mathfrak{R}}$ and its subgraph $R_\mathcal{A}$ can be seen as $\Delta$-embedded graphs, and there is a bijection $\delta$ from the set of cycles of $\mathcal{A}$  to the set of cycles of $R_\mathcal{A}.$
Moreover, observe that, given a subsequence $\mathcal{C}'$ of the cycles of $\mathcal{A}$,
since all the cycles in $\mathcal{C}'$ are subgraphs of $\cupall {\sf Influence}_{\mathfrak{R}}(\mathcal{C}')$,
 the graph $\mathcal{A}_{\mathfrak{R}}^{\langle\mathcal{C}'\rangle}$ contains a sequence of $|\mathcal{C}'|$ cycles $\mathcal{C}''$ that is a
subsequence of the cycles of the railed annulus $R_\mathcal{A}$.

\paragraph{The graph ${\sf Leveling}_{(\mathcal{A},\mathfrak{R})}(G)$.}
Let $G$ be a graph and let $(\mathcal{A},\mathfrak{R})$ be a well-aligned railed annulus flatness pair of $G$.
We set ${\sf Leveling}_{(\mathcal{A},\mathfrak{R})}(G)$ to be the graph obtained from $G$ after replacing ${\sf Compass}_{\mathfrak{R}}(\mathcal{A})$ with $\mathcal{A}_{\mathfrak{R}}$
and, given a subsequence $\mathcal{C}'$ of the cycles of $\mathcal{A}$,
we set ${\sf Leveling}_{(\mathcal{A},\mathfrak{R})}^{\langle\mathcal{C}'\rangle}(G)$ to be the graph obtained from $G$ after replacing $\cupall{\sf Influence}_{\mathfrak{R}}(\mathcal{C}')$ with $\mathcal{A}_{\mathfrak{R}}^{\langle\mathcal{C}'\rangle}$.

\subsection{Finding a flat railed annulus}
\label{subsec:findingflat}

We now formally define the notion of treewidth and of a tree decomposition.
A \emph{tree decomposition} of a graph~$G$
is a pair~$(T,\chi)$ where $T$ is a tree and $\chi: V(T)\to 2^{V(G)}$
such that
\begin{enumerate}
	\item $\bigcup_{t \in V(T)} \chi(t) = V(G),$
	\item for every edge~$e$ of~$G$ there is a $t\in V(T)$ such that
	      $\chi(t)$
	      contains both endpoints of~$e,$ and
	\item for every~$v \in V(G),$ the subgraph of~${T}$
	      induced by $\{t \in V(T)\mid {v \in \chi(t)}\}$ is connected.
\end{enumerate}
The \emph{width} of $(T,\chi)$ is defined as
$\textsf{w}(T,\chi):=
	\max\big\{\left|\chi(t)\right|-1 \bigmid t\in V(T)\big\}.$
The \emph{treewidth of $G$} is defined as
$$\tw(G):=\min\big\{\textsf{w}(T,\chi) \bigmid (T,\chi) \text{ is a tree decomposition of }G\big\}.$$

The following result is a consequence of the version of the \textsl{Flat Wall theorem}~\cite{SauST24amor,RobertsonS95XIII,KawarabayashiTW18anew,Chuzhoy15impr,GiannopoulouT13opti,KawarabayashiKR12thedis}
appeared in~\cite[Theorem~8]{SauST24amor}, stated for (well-aligned) flat railed annuli.
A (flat) railed annulus can be obtained from a slightly larger (flat) wall given by~\cite[Theorem~8]{SauST24amor} -- see~\cite[Proposition 5.1]{BasteST20acom}.
Well-alignedness of the railed annulus flatness pair can be achieved because of the regularity (and therefore well-alignedness -- see~\cite[Lemma 16]{SauST24amor}) of the flat wall given in~\cite[Theorem 8]{SauST24amor}.

\begin{proposition}\label{prop:flatann}
There are functions $\newfun{label_confrontation},\newfun{label_hierarchical}:\mathbb{N}\to \mathbb{N}$ and
an algorithm that receives as input a graph $G$, a $p\in\mathbb{N}_{\geq 3}$, and a $t\in\mathbb{N}_{\geq 1}$,
and outputs, in time $2^{{\cal O}_{t}(p^2)}\cdot n$ time, one of the following:
\begin{itemize}
\item a report  that $K_{t}$ is a minor of $G,$ or

\item a tree decomposition of $G$ of width at most $\funref{label_confrontation}(t)\cdot p,$ or

\item a set $A\subseteq V(G)$,  where $|A|\leq \funref{label_hierarchical}(t),$ a well-aligned $(p,p)$ railed annulus flatness pair $(\mathcal{A},\mathfrak{R})$ of $G\setminus A$,
and a tree decomposition of $G[V(\mathsf{Compass}_\mathfrak{R}(\mathcal{A}))\cup\mathsf{inn}_\mathfrak{R}(\mathcal{A})]$
of width at most $\funref{label_confrontation}(t)\cdot p.$

\end{itemize}
Moreover, $\funref{label_confrontation}(t)=2^{\mathcal{O}(t^2 \log t)}$ and $\funref{label_hierarchical}(t) =  \mathcal{O}(t^{24})$.
\end{proposition}

We stress that~\Cref{prop:flatann} is the main building block of the algorithm $\mathbb{A}$ in the \textsl{Local-Global-Irrelevancy Condition} of the reduction between
$(\mathsf{hw},\mathsf{CMSO}/\tw+\!\mathsf{dp})$ and $(\mathsf{tw},\mathsf{CMSO})$. In fact, the algorithm $\mathbb{A}$ uses~\Cref{prop:flatann} to find a well-aligned railed annulus from which it constructs a new colored graph; see~\Cref{sec_enhancedstruct} and~\Cref{subsec:itema}.

\subsection{Linkages in railed annuli}
\label{subsec:linkages}

A \emph{linkage} in a graph $G$ is a subgraph $L$ of $G$ whose connected components are paths of at least one edge.
The \emph{paths} of a linkage are its connected components and we denote them by $\mathcal{P}(L)$.
The \emph{size} of $L$ is the number of its paths, i.e., $|\mathcal{P}(L)|$.
The \emph{terminals} of a linkage $L$, denoted by $T(L)$, are the endpoints of the paths of $L$, and the \emph{pattern} of $L$ is the set $\{\{s,t\} \mid \mathcal{P}(L) \text{ contains some $(s,t)$-path}\}$.
Two linkages $L_1,L_2$ of $G$ are \emph{equivalent}, which we denote by $L_1 \equiv L_2$, if they have the same pattern.
By definition, given two linkages $L_1,L_2$ such that $V(L_1)\cap V(L_2) = T(L_1)\cap T(L_2)$, the graph $L_1\cup L_2$ is a linkage and $T(L_1\cup L_2)$ is the symmetric difference of $T(L_1)$ and $T(L_2)$, i.e.,  $T(L_1\cup L_2) = (T(L_1)\setminus T(L_2)) \cup (T(L_2)\setminus T(L_1)).$
In other words, when taking the union of two linkages, the resulting vertices of degree two are not considered as terminals anymore.

The proof of~\Cref{lem_levelingpaths} is presented in~\cite[Lemma 12]{GolovachST22model_arXiv} using different notation.

\begin{proposition}\label{lem_levelingpaths}
Let $G$ be a graph, let $(\mathcal{A},\mathfrak{R})$ be a well-aligned railed annulus flatness pair of $G$,
and let $\mathcal{C}'$ be a subsequence of the cycles of $\mathcal{A}$.
Also, let  $s_1,t_1,\ldots, s_k,t_k\in V(G)$ such that for every $i\in[k]$, $s_i,t_i\notin V(\cupall{\sf Influence}_{\mathfrak{R}}(\mathcal{C}'))$.
Then there is a linkage $L$ in $G$ such that $T(L)=\{s_1,t_1,\ldots, s_k,t_k\}$ if and only if there is a linkage $L'$ in ${\sf Leveling}_{(\mathcal{A},\mathfrak{R})}^{\langle\mathcal{C}'\rangle}(G)$ such that $L\equiv L'$.
\end{proposition}

\subsection{Bounded annotated treewidth sets are local inside flat railed annuli}
\label{subsec:buffer}
A crucial element of our proofs is the following lemma,
which states that in a large enough flat railed annulus,
sets of small annotated treewidth ``avoid'' the influence of a still large enough sequence of cycles of the railed annulus.

\begin{lemma}\label{lemma:buffer}
There is a function $\newfun{fun:bid}:\mathbb{N}\to\mathbb{N}$,
whose images are even numbers,
such that for every $t\in\mathbb{N}$,
if $G$ is a graph,
$A\subseteq V(G)$,
$(\mathcal{A},\mathfrak{R})$ is an $(\funref{fun:bid}(t) \cdot p,q)$-railed annulus flatness pair of $G\setminus A$,
where $p,q\in\mathbb{N}_{\geq 3}$,
then for every set $X\subseteq V(G)$ such that $\tw(G,X)\leq t$,
there exists a subsequence $\mathcal{C}'$ of the cycles of $\mathcal{A}$ of size $p$
such that $X$ does not intersect $V(\cupall\mathsf{Influence}_\mathfrak{R}(\mathcal{C}'))$.
\end{lemma}

The proof of~\Cref{lemma:buffer} is a consequence of the fact that after partitioning the set of cycles $\mathcal{C}$ of a flat $(\funref{fun:bid}(t) \cdot p,q)$-railed annulus $(\mathcal{A},\mathfrak{R})$ into $\funref{fun:bid}(t)$ sequences of $p$ consecutive cycles  $\mathcal{C}_i$,
if a set $X$ intersects, for every $i$, the influence of $\mathcal{C}_i$, then we can find a large enough (as a function of $t$) grid as an $X$-minor of $G$.
This is justified by~\cite[Lemma 3.1]{demaine2004bidimensional}, which intuitively states that if $\Gamma$ is a ``sufficiently large'' grid and a set $X$ contains ``many enough'' vertices of some ``sufficiently central'' subgrid of $\Gamma$, then $\Gamma$ contains a still large enough grid as an $X$-minor.
In our setting, this can be applied after focusing on ``sufficiently central'' collections $\mathcal{C}_i$.
Therefore, if $G$ contains a large grid as an $X$-minor, via the \textsl{grid-exclusion theorem}~\cite{RobertsonS86GMV,ChuzhoyT21towa},
$\tw(G,X)$ should also be large. See also the proof of~\cite[Lemma 30]{FominGSST21comp_arxiv} for an alternative approach involving brambles.
We omit the proof of~\Cref{lemma:buffer} as it follows from the above standard arguments.

\section{Building the new structures and formulas}
\label{sec_enhancedstruct}
This section contains some preliminary definitions and results used in order to show \textsl{Local-Global-Irrelevancy Condition} of the reduction between
$(\mathsf{CMSO/tw}\!+\!\mathsf{dp},\mathsf{hw})$ and $(\mathsf{CMSO},\mathsf{tw})$,
as well as the definitions of the colored-graph vocabulary $\sigma_{d,l}$ and formulas $\Phi_{d,l}$ needed in this condition (see~\Cref{def_red}).
The structures of the aimed colored-graph vocabulary ($\sigma_{d,l}$) that we consider are built from a flat railed annulus.
Since such a railed annulus of a given graph may be flat only after the removal of a set $A$ of few \emph{apices} (see~\Cref{prop:flatann}), we find a way to ``interpret them away'' by modifying the formulas and the considered structures. This trick comes from~\cite[Lemma 26]{FlumG01fixe} (see also~\cite{FominGSST23,GolovachST22model}) and is presented in~\Cref{subsec_apices}.
Then, we proceed to the definition of the new vocabulary
(\Cref{def_vocab}) and the new set of formulas on this vocabulary (\Cref{def:ltp}), which appear in~\Cref{subsec:voc,subsec:layeredtypes}, respectively.
Before providing these definitions
and in order to help the reader better understand how these will be used in the proofs,
we provide some additional definitions and notations on vertex sets arising from different ``layerings'' of a graph with respect to a flat railed annulus; see~\Cref{subsec_conventions}.
These will be used to define the structures that
the algorithm $\mathbb{A}$ of the \textsl{Local-Global-Irrelevancy Condition} will output, which is done in~\Cref{def:str} in~\Cref{subsec:voc}.

\subsection{Dealing with apices}
\label{subsec_apices}

We first show how to deal with apex vertices that can be adjacent to vertices in a flat area of the graph in a completely arbitrary way.
We use a generalization of the idea of~\cite[Lemma 26]{FlumG01fixe}, already appearing in~\cite{FominGSST23,GolovachST22model}.

\begin{definition}\label{def:apices}
Let $l\in\mathbb{N}$.
Given a graph $G$ and a $\bar{a} =(a_1,\ldots, a_l)\in V(G)^l$,
we define the structure
${\sf ap}_{l}(G,\bar{a})$ to
be the structure on vocabulary $\{E\}^{\langle l \rangle} :=\{E,c_1,\ldots,c_l,C_1,\ldots,C_l\}$ obtained as follows:

\begin{itemize}
\item its universe is $V(G),$

\item $E$ is interpreted as all $\{u,v\}\in E(G)$
such that either $u,v\in V(G)\setminus V(\bar{a})$ or $u,v\in V(\bar{a})$,

\item for every $i\in [l]$, $c_i$ is interpreted as $a_i$ and $C_i$ is interpreted as $N_{G}(a_i)\setminus V(\bar{a}).$
\end{itemize}
\end{definition}
Intuitively, we keep the same universe,
we keep only edges that are either between apices or between non-apices, and
we color the neighbors of $a_i$ by color $C_i$.
The following result from~\cite{GolovachST22model_arXiv} shows how to interpret the disjoint-paths predicate in ${\sf ap}_{l}(G,\bar{a})$.

\begin{proposition}[\!\!\cite{GolovachST22model_arXiv}]\label{lem:translation-dp-apices}
There is a function~$\newfun{@transversales}:\mathbb{N}^2\to\mathbb{N}$ such that for every $l,k\in\mathbb{N}$, there is a formula $\zeta_{{\sf dp}}(x_{1},y_{1},\ldots, x_{k}, y_{k})\in \mathsf{FO}\!+\!\mathsf{dp}[\{E\}^{\langle  l \rangle}]$ of {at most $\funref{@transversales}(l,k)$ first-order quantifiers} such that
for every graph $G,$ every $v_1,u_1,\ldots, v_k,u_k\in V(G)$, and every $\bar{a}\in V(G)^l$, it holds that
$G\models {\sf dp}(v_1, u_{1},\ldots, v_{k}, u_{k}) \iff {\sf ap}_{l}(G,\bar{a})\models
\zeta_{{\sf dp}}(v_{1},  u_{1},\ldots, v_{k}, u_{k}).$
\end{proposition}

Due to~\Cref{lem:translation-dp-apices},
one can ``shift'' the problem of satisfaction of the disjoint-paths predicate from $G$ to ${\sf ap}_{l}(G,\bar{a})$
by considering the $\mathsf{FO}\!+\!\mathsf{dp}$-formula $\zeta_{{\sf dp}}$. This formula has extra quantifiers on first-order variables and in order to deal with them we have to consider annotated types of larger ``depth''.
However, the extra quantified variables of $\zeta_{{\sf dp}}$ are not bounded by the annotation $\bar{R}$ of a given annotated graph $(G,\bar{R})$. For this reason, we define \emph{semi-annotated} types, which differ from the annotated types in~\Cref{def:tp} only on where the considered vertices of the $i$-th level of the definition should belong. In particular,
while in~\Cref{def:tp} we ask that a vertex on the $i$-th level of the definition should belong to $R_{m+d-i}$,
we now ask this \textsl{only} when $i>d-r$, for some given $r\leq d$. See the next paragraph for the formal definition of semi-annotated patterns.

\paragraph{A variant of annotated types: semi-annotated types.}
Let $\sigma$ be a colored-graph vocabulary and let $m,d,r\in\mathbb{N}$
such that  $r\leq d$.
Let $\mathfrak{G}$ be a $\sigma$-structure and let $\bar{R}=(R_1,\ldots, R_{m+r})$ be an $(m+r)$-tuple of subsets of $V(\mathfrak{G})$.
For every
$m$-tuple $\bar{V}\in (2^{V(\mathfrak{G})})^m$ and every $d$-tuple
$\bar{v}\in V(\mathfrak{G})^d$,
we define
$$\mathsf{tp}_{m,d,t}^{(0,r)}(\mathfrak{G},\bar{R},\bar{V},\bar{v}) :=
\mathsf{tp}_{m,d,t}^0(\mathfrak{G},\bar{R},\bar{V},\bar{v}).$$

For every $i\in[d]$, every $m$-tuple $\bar{V}\in (2^{V(\mathfrak{G})})^m$
and every $(d-i)$-tuple
$\bar{v} \in V(\mathfrak{G})^{d-i}$,
we define
\[
\mathsf{tp}_{m,d,t}^{(i,r)}(\mathfrak{G},R,\bar{V},\bar{v})
:=
\big\{\mathsf{tp}_{m,d,t}^{(i-1,r)}
(\mathfrak{G},\bar{R},\bar{V},\bar{v}u)\mid u\in A_{i-1}\big\},
\]
where $A_{i} =
V(\mathfrak{G})$, if $i\in[d-r]$,
while $A_i = R_{m+d-i}$ if $i\in[d-r+1,\ldots, d]$.

Also, for every $i\in[1+d,m+d]$ and every $(m+d-i)$-tuple $\bar{V}\in (2^{V(\mathfrak{G})})^{m+d-i}$,
we set
\[
\mathsf{tp}_{m,d,t}^{(i,r)}
(\mathfrak{G},\bar{R},\bar{V}) :=
\{\mathsf{tp}_{m,d,t}^{(i-1,r)}(\mathfrak{G},\bar{R},\bar{V}U) \mid U\subseteq R_{m+d-(i-1)}\}.
\]

\paragraph{From the semi-annotated type of $({\sf ap}_{l}(G,\bar{a}),\bar{R})$ to the annotated type of $(G,\bar{R})$.}
We show that if the ``gap'' between semi-annotation and full annotation is the number of extra quantifiers obtained by $\zeta_{{\sf dp}}$ from~\Cref{lem:translation-dp-apices},
then equality of (semi)-annotated types for $({\sf ap}_{l}(G,\bar{a}),\bar{R})$ and $({\sf ap}_{l}(G',\bar{a}),\bar{R}')$ implies equality of (fully-)annotated types for $(G,\bar{R})$ and $(G',\bar{R}')$.

\begin{lemma}\label{lem_inter2equiv}
Let $l,m,r\in\mathbb{N}$ and let $d=r+\funref{@transversales}(l,r)$.
Also, let two graphs $G,G'$,
let $\bar{R}\in (2^{V(G)})^{m+r}, \bar{R}'\in (2^{V(G')})^{m+r}$,
and let $\bar{a}\in V(G)^l, \bar{a}'\in V(G')^l$.
Then
\[\mathsf{tp}^{(m+d,r)}_{m,d,t}
({\sf ap}_{l}(G,\bar{a}),\bar{R})
 = \mathsf{tp}^{(m+d,r)}_{m,d,t}({\sf ap}_{l}(G',\bar{a}'),\bar{R}')
\implies
\mathsf{tp}^{m+r}_{m,r,t}(G,\bar{R})=\mathsf{tp}^{m+r}_{m,r,t}(G',\bar{R}').\]
\end{lemma}

\begin{proof}
We set $\bar{R} = (R_1,\ldots, R_{m+r})$ and $\bar{R}'=(R_1',\ldots, R_{m+r}')$.
Suppose that $\mathsf{tp}^{(m+d,r)}_{m,d,t}
({\sf ap}_{l}(G,\bar{a}),\bar{R})
 = \mathsf{tp}^{(m+d,r)}_{m,d,t}({\sf ap}_{l}(G',\bar{a}'),\bar{R}')$.
We prove inductively, for increasing values of $i$, that the following two statements hold for every $i\in[m+r]$, where $j:=\min\{m,m+r-i\}$ and $\ell:=\max\{0,r-i\}$:

\begin{itemize}
\item[(i)] Let a $j$-tuple $\bar{V}\in 2^{R_1}\times \cdots \times 2^{R_j}$ and an $\ell$-tuple $\bar{v}\in R_{m+1}\times \cdots \times R_{m+\ell}$.
There exists a $j$-tuple $\bar{U}\in 2^{R_1'}\times \cdots \times 2^{R_j'}$
and a $\bar{u}\in R_{m+1}'\times \cdots \times R_{m+\ell}'$ such that if $\mathsf{tp}^{(d-r+i,r)}_{m,d,t}({\sf ap}_{l}(G,\bar{a}),\bar{R},\bar{V},\bar{v})
=
\mathsf{tp}^{(d-r+i,r)}_{m,d,t}({\sf ap}_{l}(G',\bar{a}'),\bar{R}',\bar{U},\bar{u})$, then
$\mathsf{tp}^{i}_{m,r,t}(G,\bar{R},\bar{V},\bar{v})
= \mathsf{tp}^{i}_{m,r,t}(G',\bar{R}',\bar{U},\bar{u})
$.

\item[(ii)] Let a $j$-tuple $\bar{U}\in 2^{R_1'}\times \cdots \times 2^{R_j'}$ and an $\ell$-tuple $\bar{u}\in R_{m+1}'\times \cdots \times R_{m+\ell}'$.
There exists a $j$-tuple $\bar{V}\in 2^{R_1}\times \cdots \times 2^{R_j}$
and a $\bar{v}\in R_{m+1}\times \cdots \times R_{m+\ell}$ such that if $\mathsf{tp}^{(d-r+i,r)}_{m,d,t}({\sf ap}_{l}(G,\bar{a}),\bar{R},\bar{V},\bar{v})
=
\mathsf{tp}^{(d-r+i,r)}_{m,d,t}({\sf ap}_{l}(G',\bar{a}'),\bar{R}',\bar{U},\bar{u})$, then
$\mathsf{tp}^{i}_{m,r,t}(G,\bar{R},\bar{V},\bar{v})
= \mathsf{tp}^{i}_{m,r,t}(G',\bar{R}',\bar{U},\bar{u})
$.
\end{itemize}
First we show that (i) holds for $i=0$, i.e., $j=m$ and $\ell=r$.
We consider an $m$-tuple $\bar{V}\in 2^{R_1}\times \cdots \times 2^{R_m}$ and an
$r$-tuple $\bar{v}\in R_{m+1}\times \cdots \times R_{m+r}$.
Since
$\mathsf{tp}^{(m+d,r)}_{m,d,t}
({\sf ap}_{l}(G,\bar{a}),\bar{R})
 = \mathsf{tp}^{(m+d,r)}_{m,d,t}({\sf ap}_{l}(G',\bar{a}'),\bar{R}')$,
 there is a  $\bar{U}\in 2^{R_1'}\times \cdots \times 2^{R_m'}$ and a
 $\bar{u}\in R_{m+1}'\times \cdots \times R_{m+r}'$
 such that
 \begin{equation}\label{eq:pattern-chop}
 \mathsf{tp}^{(d-r,r)}_{m,d,t}({\sf ap}_{l}(G,\bar{a}),\bar{R},\bar{V},\bar{v})
=
\mathsf{tp}^{(d-r,r)}_{m,d,t}({\sf ap}_{l}(G',\bar{a}'),\bar{R}',\bar{U},\bar{u}).
\end{equation}

Note that~\eqref{eq:pattern-chop} implies that
for every $i,j\in[r]$,
$\{v_i,v_j\}\in E(G)\iff \{u_i,u_j\}\in E(G')$.
For this, it suffices to observe that $v_i$ and $v_j$ are adjacent in $G$ if and only if either they are adjacent in ${\sf ap}_{l}(G,\bar{a})$, or there is an $h\in[l]$ such that
either ``$v_i = a_h$ and $v_j\in  N_G(a_h)$'' or
``$v_j = a_h$ and $v_i\in N_G(a_h)$''.

We also show that,
for every $i_1,\ldots,i_k,j_1,\ldots,j_k\in [r]$,
$G\models  {\sf dp}_k(v_{i_1},v_{j_1},\ldots,v_{i_k},v_{j_k})$ if and only if $G'\models {\sf dp}_k(u_{i_1},u_{j_1},\ldots,u_{i_k},u_{j_k})$.
Indeed,
by~\Cref{lemgametreesCMSO/tw},~\eqref{eq:pattern-chop} implies that $
{\sf ap}_{l}(G,\bar{a})
\models  \zeta_{\sf dp}(v_{i_1},v_{j_1},\ldots,v_{i_k},v_{j_k})
\iff
{\sf ap}_{l}(G',\bar{a}')
\models  \zeta_{\sf dp}(u_{i_1},u_{j_1},\ldots,u_{i_k},u_{j_k}),$
where $\zeta_{\sf dp}$ is given by~\Cref{lem:translation-dp-apices}.
Due to~\Cref{lem:translation-dp-apices},
$G\models  {\sf dp}_k(v_{i_1},v_{j_1},\ldots,v_{i_k},v_{j_k})\iff G'\models {\sf dp}_k(u_{i_1},u_{j_1},\ldots,u_{i_k},u_{j_k})$.

Therefore, for every
$\psi(\bar{X},\bar{x}) \in \mathcal{L}^{m,r}_t[\{E\}]$,
$G\models \psi(\bar{V},\bar{v}) \iff G'\models \psi(\bar{U},\bar{u})$, which in
turn implies that $\mathsf{tp}^{0}_{m,r,t}(G,\bar{R},\bar{V},\bar{v})
= \mathsf{tp}^{0}_{m,r,t}(G',\bar{R}',\bar{U},\bar{u})$, which concludes the proof of (i) for $i=0$.
The proof of (ii) for $i=0$ is completely symmetric.

To conclude the proof by induction, let us
show how to prove (i) for every $i\geq 1$, assuming that (i) and (ii) hold for $i-1$.
Using the assumption that $\mathsf{tp}^{(m+d,r)}_{m,d,t}
({\sf ap}_{l}(G,\bar{a}),\bar{R})
 = \mathsf{tp}^{(m+d,r)}_{m,d,t}({\sf ap}_{l}(G',\bar{a}'),\bar{R}')$,
we can find a $j$-tuple $\bar{U}$ and an $\ell$-tuple $\bar{u}$ in $G'$, such that
\begin{equation}\label{eq_ipgsatt}
\mathsf{tp}^{(d-r+i,r)}_{m,d,t}
({\sf ap}_{l}(G,\bar{a}),\bar{R},\bar{V},\bar{v})
 = \mathsf{tp}^{(d-r+i,r)}_{m,d,t}({\sf ap}_{l}(G',\bar{a}'),\bar{R}',\bar{U},\bar{u}).
 \end{equation}
 We distinguish two cases, depending on the value of $i$.
If  $i>r$, then $\bar{v}$ has length zero.
Consider a set $V_{j+1}\subseteq V(G)$ and use~\eqref{eq_ipgsatt} to obtain a $U_{j+1}\subseteq V(G')$ such that
$\mathsf{tp}^{(d-r+i,r)}_{m,d,t}
({\sf ap}_{l}(G,\bar{a}),\bar{R},\bar{V}V_{j+1})
 = \mathsf{tp}^{(d-r+i,r)}_{m,d,t}({\sf ap}_{l}(G',\bar{a}'),\bar{R}',\bar{U}U_{j+1}).$
Applying the induction hypothesis, we get
$\mathsf{tp}^{i-1}_{m,r,t}(G,\bar{R},\bar{V}V_{j+1})
= \mathsf{tp}^{i-1}_{m,r,t}(G',\bar{R}',\bar{U}U_{j+1})$. Symmetrically, we can consider a set
$U_{j+1}\subseteq V(G')$  and
use~\eqref{eq_ipgsatt}
to obtain a set $V_{j+1}\subseteq V(G)$
such that $\mathsf{tp}^{i-1}_{m,r,t}(G,\bar{R},\bar{V}V_{j+1})
= \mathsf{tp}^{i-1}_{m,r,t}(G',\bar{R}',,\bar{U}U_{j+1}).$
Therefore, $\mathsf{tp}^{i}_{m,r,t}(G,\bar{R},\bar{V})
= \mathsf{tp}^{i}_{m,r,t}(G',\bar{R}',\bar{U})$.
For $i\leq r$, the arguments are for $v_{j+1}$ and $u_{j+1}$ and are identical to the above ones.
The proof of the lemma is concluded by observing that, since (i) and (ii) hold for $i=m+d$ and $\mathsf{tp}^{(m+d,r)}_{m,d,t}
({\sf ap}_{l}(G,\bar{a}),\bar{R})
 = \mathsf{tp}^{(m+d,r)}_{m,d,t}({\sf ap}_{l}(G',\bar{a}'),\bar{R}')$,
we have that $\mathsf{tp}^{m+r}_{m,r,t}(G,\bar{R})=\mathsf{tp}^{m+r}_{m,r,t}(G',\bar{R}')$.
\end{proof}

\subsection{Layering with respect to flat railed annuli}\label{subsec_conventions}

We proceed to define a graph with colors and roots obtained from a given graph $G$ and a railed annulus flatness pair $(\mathcal{A},\mathfrak{R})$ of $G$.
To define such a colored graph with roots, we use the {\sl leveling} ${\sf Leveling}_{(\mathcal{A}, \mathfrak{R})} (G)$ of $(\mathcal{A},\mathfrak{R})$, that is the ``planar representation'' of $(\mathcal{A},\mathfrak{R})$,
as defined in~\Cref{subsec_levelings}.
In order for $(\mathcal{A},\mathfrak{R})$ to have this ``planar representation''
property, it has to be {\sl well-aligned} (see also~\Cref{subsec_levelings} and~\cite[Subsection 3.4]{SauST24amor}).

\paragraph{Refinements of railed annuli.}
Let $\mathcal{A}$ be an $(s^h \cdot p,q)$-railed annulus, for some $h\in\mathbb{N}$, $s\in\mathbb{N}_{\geq 2}$, and $p,q\in\mathbb{N}_{\geq 3}$.
For every $\bar{w}\in [0,s-1]^h$, we use $\mathcal{A}_{\bar{w}}$ to denote the $(p,q)$-railed annulus that is cropped by the cycles $C_{1 + p\cdot (n_{\bar{w}}-1)}$ and $C_{p\cdot n_{\bar{w}}}$ of $\mathcal{A}$, where $n_{\bar{w}}=1+\sum_{i\in[r]} w_i\cdot s^{h-i+1}$.
We call $\{\mathcal{A}_{\bar{w}}\mid \bar{w}\in[0,s-1]^{h}\}$ the \emph{$(s,h)$-refinement} of $\mathcal{A}$.
Given a $\bar{w}\in [0,s-1]^{h}$, we use $\mathcal{C}_{\bar{w}}$ to denote the cycles of $\mathcal{A}_{\bar{w}}$.
Keep in mind that for every $\bar{w}\in [0,s-1]^{h}$, $|\mathcal{C}_{\bar{w}}| = p$.
We denote by $C_{\bar{w}}^1,\ldots, C_{\bar{w}}^{p}$ the elements of $\mathcal{C}_{\bar{w}}$.
Also, for every $i\in[h-1]$ and every $\bar{w}\in [0,s-1]^{i}$, we denote by $\mathcal{A}_{\bar{w}}$ the railed annulus
$\bigcup_{w\in[0,s-1]}\mathcal{A}_{\bar{w}w}$ and by $\mathcal{C}_{\bar{w}}$ its cycles (i.e., $\mathcal{C}_{\bar{w}}= \bigcup_{w\in[0,s-1]}\mathcal{C}_{\bar{w}w}$).

\paragraph{The extra boundary vertices ${r}_{i}^{\langle\bar{w}\rangle}$ and the graphs  $G_{\bar{w}}$ and $G_{\bar{w}}^{\sf out}$.}
Let $G$ be a graph and let $(\mathcal{A},\mathfrak{R})$ be a well-aligned $(s^h \cdot p,q)$-railed annulus flatness pair of $G$,
for some $h\in\mathbb{N}$, $s\in\mathbb{N}_{\geq 2}$,
and $p,q\in\mathbb{N}_{\geq 3}$.
Let $\{\mathcal{A}_{\bar{w}}\mid \bar{w}\in [0,s-1]^{h}\}$ be the $(s,h)$-refinement of $\mathcal{A}$.
Let $\mathcal{C}_{\bar{w}} =(C_{\bar{w}}^1,\ldots, C_{\bar{w}}^{p})$ be the cycles of $\mathcal{A}_{\bar{w}}$, for each $\bar{w}\in [0,s-1]^{h}$.

We consider the representation $R_{\mathcal{A}}$ of $\mathcal{A}$ in $\mathcal{A}_\mathfrak{R}$ (see~\Cref{subsec_levelings}), which is an $(s^h \cdot p ,q)$-railed annulus.
Let $\hat{P}_1,\ldots,\hat{P}_q$ be the rails of  $R_{\mathcal{A}}$.
Also, for every  $\bar{w}\in [0,s-1]^{h}$, let $\hat{C}_{\bar{w}}^\mathsf{mid}$ be the cycle of $R_{\mathcal{A}}$
corresponding to the cycle $C_{\bar{w}}^{\lceil p/2\rceil}$ of $\mathcal{C}_{\bar{w}}$.
As described in~\Cref{subsec_levelings}, the graph $R_{\mathcal{A}}$
is $\Delta$-embedded in some closed  annulus $\Delta$ and
contains a sequence of $p$ cycles $\hat{\mathcal{C}}_{\bar{w}}$.
This sequence $\hat{\mathcal{C}}_{\bar{w}}$ of cycles is bijectively mapped to the cycles $C_{\bar{w}}^1,\ldots, C_{\bar{w}}^{p}$ of $\mathcal{A}_{\bar{w}}$.
Intuitively, the cycle $\hat{C}_{\bar{w}}^\mathsf{mid}$ is the ``middle'' cycle among the $p$ cycles of $\hat{\mathcal{C}}_{\bar{w}}$.

We see each rail $\hat{P}_i$ of $R_\mathcal{A}$ as being oriented towards the ``inner'' part of $R_\mathcal{A}$, i.e., it first crosses $\hat{C}_{0^h}^\mathsf{mid}$ and then crosses $\hat{C}_{(s-1)^h}^{\mathsf{mid}}$.
For every $i\in[q]$,
we define ${r}_{i}^{\langle\bar{w}\rangle}$ as the first vertex of $\hat{P}_i$ that appears in the intersection of the cycle $\hat{C}_{\bar{w}}^\mathsf{mid}$ and the rail $\hat{P}_i$,
while traversing $\hat{P}_i$ according to the aforementioned orientation.
Also, we denote by $\Delta_{\bar{w}}^\mathsf{mid}$ the closed annulus of $\Delta$ whose boundaries are $\hat{C}_{\bar{w}}^\mathsf{mid}$ and $\hat{C}_{\bar{w}}^{p}$.
We use
\begin{itemize}
\item $V_{\mathcal{A}_{\bar{w}}}^{\sf lev}$ to denote the set of all vertices of ${\sf Leveling}_{(\mathcal{A}, \mathfrak{R})}^{\langle \mathcal{C}_{\bar{w}}\rangle}(G)$ embedded in $\Delta_{\bar{w}}^\mathsf{mid}$,

\item  $V_{\mathcal{A}_{\bar{w}}}^{\mathsf{inn}}$ to denote the set
$(V(\cupall\mathsf{Influence}_\mathfrak{R}(C_{\bar{w}}^{1}))\cup \mathsf{inn}_\mathfrak{R}(\mathcal{A}))\setminus V(\cupall\mathsf{Influence}_\mathfrak{R}(\mathcal{C}_{\bar{w}}))$,

\item $G_{\bar{w}}$ to denote the graph
${\sf Leveling}_{(\mathcal{A}, \mathfrak{R})}^{\langle \mathcal{C}_{\bar{w}}\rangle}(G)
[V_{\mathcal{A}_{\bar{w}}}^{\sf lev}\cup V_{\mathcal{A}_{\bar{w}}}^{\mathsf{inn}}]$, and

\item $G_{\bar{w}}^{\sf out}$ to denote the graph
${\sf Leveling}_{(\mathcal{A}, \mathfrak{R})}^{\langle \mathcal{C}_{\bar{w}}\rangle}(G)
\setminus \Big(V_{\mathcal{A}_{\bar{w}}}^{\mathsf{inn}} \cup \big(V_{\mathcal{A}_{\bar{w}}}^{\sf lev}\cap \inter(\Delta_{\bar{w}}^\mathsf{mid})\big)\Big)$.
\end{itemize}
Keep in mind that $\{r_{1}^{\langle\bar{w}\rangle},\ldots,r_{q}^{\langle\bar{w}\rangle}\}\subseteq V(G_{\bar{w}})\cap V(G_{\bar{w}}^{\sf out})$ and $V_{\mathcal{A}_{\bar{w}}}^{\mathsf{inn}}\subseteq V({\sf Leveling}_{(\mathcal{A}, \mathfrak{R})}^{\langle \mathcal{C}_{\bar{w}}\rangle}(G))\setminus V(\mathcal{A}_{\mathfrak{R}}^{\langle\mathcal{C}_{\bar{w}}\rangle}).$

\paragraph{Rerouting linkages on a flat part of the graph.}
We conclude this subsection with the following result.
It intuitively shows that every linkage traversing a large enough collection of consecutive cycles of a flat railed annulus (whose endpoints are not inside the influence of the cycles) can be ``rerouted'' so as to cross the central cycle of this collection only in the intersection of this cycle with a small number of rails of the annulus.
Moreover, the parts of the rerouted linkage that are inside the influence of the collection of cycles can be ``projected'' to the leveling of the graph corresponding to this collection of cycles.

\begin{proposition}\label{lem_colomodelsrerout}
There are functions $\newfun{@aristocracias},\newfun{fun_ulinkastar}:\mathbb{N}\to\mathbb{N}$ such that for every $h,d,s\in\mathbb{N}$, if $G$ is a graph,
$(\mathcal{A},\mathfrak{R})$ is a well-aligned $(p,q)$-railed annulus flatness pair of $G$, where $p=s^{h}\cdot\funref{@aristocracias}(d)$ and $q\geq 5/2\cdot \funref{fun_ulinkastar}(d)$,
and $v_1,\ldots, v_d\subseteq V(G)$, then for every
$\bar{w}\in[0,s-1]^h$ such that $v_1,\ldots,v_d\notin V(\cupall{\sf Influence}_{\mathfrak{R}}(\mathcal{C}_{\bar{w}}))$ and every linkage $L$ of $G$ where $T(L)\subseteq \{v_1,\ldots,v_d\}$,
there is
 a linkage $\tilde{L}_1$ of $G_{\bar{w}}$,
a linkage $\tilde{L}_2$ of $G_{\bar{w}}^{\sf out}$,
and a set $Z\subseteq \{r_{1}^{\langle\bar{w}\rangle},\ldots,r_{\funref{fun_ulinkastar}(d)}^{\langle\bar{w}\rangle}\}$ such that
\begin{itemize}
\item $T(\tilde{L}_1)  =  (T(L)\cap V(G_{\bar{w}})) \cup Z$,
\item $T(\tilde{L}_2)  =  (T(L)\cap V(G_{\bar{w}}^{\sf out})) \cup Z$,
and
\item $\tilde{L}_1 \cup \tilde{L}_2$
is a linkage of ${\sf Leveling}_{(\mathcal{A}, \mathfrak{R})}^{\langle \mathcal{C}_{\bar{w}}\rangle}(G)$ that is equivalent to $L$.
\end{itemize}
\end{proposition}

The proof of~\Cref{lem_colomodelsrerout} uses the Linkage Combing Lemma~\cite[Corollary 3.4]{GolovachST23comb}
and is presented in~\cite[Lemma 4]{GolovachST22model_arXiv}, using different notation.
The function $\funref{fun_ulinkastar}$ comes from the \textsl{Unique Linkage Theorem}~\cite{KawarabayashiW10asho,RobertsonS09XXI}.
Also, $\funref{@aristocracias}(k)= \mathcal{O}((\funref{fun_ulinkastar}(k))^2)$.

\subsection{The colored-graph vocabulary for the Local-Global-Irrelevancy Condition}
\label{subsec:voc}

We now proceed to define a way of transforming a vocabulary $\sigma$ to a new vocabulary $\sigma^{\langle s,h,b\rangle}$, for non-negative integers $s,h,b$.
This will be used to construct the claimed vocabulary for the \textsl{Local-Global-Irrelevancy Condition} of the reduction between
$(\mathsf{CMSO/tw}\!+\!\mathsf{dp},\mathsf{hw})$ and $(\mathsf{CMSO},\mathsf{tw})$.
\medskip

\begin{importantdef}\label{def_vocab}
Let $\sigma$ be a colored-graph vocabulary and let $h,b\in\mathbb{N}$ and $s\in\mathbb{N}_{\geq 2}$.
We set $\sigma^{\langle s,h,b\rangle}$ to be the vocabulary obtained from $\sigma$ after adding:
\begin{itemize}

\item $s^{h+1}$ unary predicate symbols
$Y_{\bar{w}}, \bar{w}\in[0,s-1]^h$
and
$Z_{\bar{w}}, \bar{w}\in[0,s-1]^h$, and
\item $s^h\cdot b$ constant symbols $c^{\bar{w}}_1,\ldots,c^{\bar{w}}_{b}, \bar{w}\in[0,s-1]^h$.
\end{itemize}
\end{importantdef}

In our proofs, we will construct $\sigma^{\langle s,h,b\rangle}$ for $\sigma=\{E\}^{\langle  l \rangle}$ (see~\Cref{def:apices}), for a given $l\in\mathbb{N}$, while $s,h,b$ will depend on some given $m,d\in\mathbb{N}$ (see~\Cref{subsec:itema}).
In the next paragraph,
we define certain structures on the vocabulary of~\Cref{def_vocab}. While these structures will not be used in our proofs, we present them here as an intermediate step to understand~\Cref{def:str} and the semantic meaning of the layered annotated patterns defined in~\Cref{subsec:layeredtypes}.

\paragraph{Enhanced boundaried structures from flat railed annuli.}
Let $\sigma$ be a colored-graph vocabulary and let $h\in\mathbb{N}$, let $s\in\mathbb{N}_{\geq 2}$, let $p,q\in\mathbb{N}_{\geq 3}$, and let $b\in[q]$.
Let $\mathfrak{G}$ be a $\sigma$-structure (and let $G$ be its Gaifman graph)
and let $(\mathcal{A},\mathfrak{R})$ be a well-aligned
$(s^{h} \cdot p ,q)$-railed annulus flatness pair of $G$.
We define the $\sigma^{\langle s,h,b\rangle}$-structure $\mathsf{Layered}_{\mathcal{A}}^{(s,h,p,b)}(\mathfrak{G})$
as follows:
\begin{itemize}
\item
its universe is
$V(\mathfrak{G})\cup\bigcup_{\bar{w}\in[0,s-1]^h}V({\sf Leveling}_{(\mathcal{A}, \mathfrak{R})}^{\langle \mathcal{C}_{\bar{w}}\rangle}(G))$,

\item
$E$ is interpreted as $E(\mathfrak{G})\cup\bigcup_{\bar{w}\in[0,s-1]^h}E({\sf Leveling}_{(\mathcal{A}, \mathfrak{R})}^{\langle \mathcal{C}_{\bar{w}}\rangle}(G))$,

\item every $c\in \sigma\setminus\{E\}$ (resp. $C\in \sigma\setminus\{E\}$) is interpreted as $c^\mathfrak{G}$ (resp. $C^\mathfrak{G}$),

\item for every $\bar{w}\in[0,s-1]^h$,
$Y_{\bar{w}}$ is interpreted
as $V({\sf Leveling}_{(\mathcal{A}, \mathfrak{R})}^{\langle \mathcal{C}_{\bar{w}}\rangle}(G))\setminus V(\mathcal{A}_{\mathfrak{R}}^{\langle\mathcal{C}_{\bar{w}}\rangle})$,
\item for every $\bar{w}\in[0,s-1]^h$, $Z_{\bar{w}}$  is interpreted
as
$\mathcal{A}_{\mathfrak{R}}^{\langle\mathcal{C}_{\bar{w}}\rangle}$, and

\item for every $\bar{w}\in[0,s-1]^h$, $c^{\bar{w}}_1,\ldots,c^{\bar{w}}_{b}$ are interpreted
as the vertices $r_{1}^{\langle\bar{w}\rangle},\ldots,r_{b}^{\langle\bar{w}\rangle}$.
\end{itemize}

\paragraph{The structure $E_{\mathcal{A}}^{(s,h,p,b,l)}(G,\bar{a})$.}
The next structure contains a new element ``$\mathspace$'' in its universe. This will be used to intuitively expresses the absence of a vertex for the corresponding entry.
In the next definition we use notation introduced in~\Cref{subsec_conventions}. See also~\Cref{def:apices}
for the definition of the $\mathfrak{R}$-leveling ${\mathcal{A}}_{\mathfrak{R}}$ of $\mathcal{A}$.
\medskip

\begin{importantdef}\label{def:str}
Let $h,l\in\mathbb{N}$, let $s\in\mathbb{N}_{\geq 2}$, let $p,q\in\mathbb{N}_{\geq 3}$, and let $b\in[q]$.
Also,
let $G$ be a graph,
let $\bar{a}=(a_1,\ldots,a_l)\in V(G)^l$,
and let $(\mathcal{A},\mathfrak{R})$ be a well-aligned
$(s^h \cdot p,q)$-railed annulus flatness pair of $G\setminus V(\bar{a})$.
We set $\tau:=\{E,c_1,\ldots,c_l,C_1,\ldots,C_l\}$ (see~\Cref{def:apices})
and define the $\tau^{\langle s,h,b\rangle}$-structure $E_{\mathcal{A}}^{(s,h,p,b,l)}(G,\bar{a})$
as follows:
\begin{itemize}
\item
its universe is
$V(\mathsf{Compass}_\mathfrak{R}(\mathcal{A}))\cup \mathsf{inn}_\mathfrak{R}(\mathcal{A})\cup V({\mathcal{A}}_{\mathfrak{R}})\cup V(\bar{a})\cup\{\mathspace\}$,

\item
$E$ is interpreted as $E(G[V(\mathsf{Compass}_\mathfrak{R}(\mathcal{A}))\cup \mathsf{inn}_\mathfrak{R}(\mathcal{A})])\cup E({\mathcal{A}}_{\mathfrak{R}}) \cup E(G[V(\bar{a})])$,

\item For $i\in[l]$, $c_i$ (resp. $C_i$) is interpreted as $a_i$
(resp. $N_{G}(a_i)\cap (V(\mathsf{Compass}_\mathfrak{R}(\mathcal{A}))\cup \mathsf{inn}_\mathfrak{R}(\mathcal{A}))$),

\item for every $\bar{w}\in[0,s-1]^h$,
$Y_{\bar{w}}$ is interpreted as the set $V_{\mathcal{A}_{\bar{w}}}^{\mathsf{inn}}\cup\{\mathspace\}$,

\item  for every $\bar{w}\in[0,s-1]^h$,
 $Z_{\bar{w}}$ is interpreted as the set
 $V_{\mathcal{A}_{\bar{w}}}^{\sf lev}\cup\{\mathspace\}$,

\item for every $\bar{w}\in[0,s-1]^h$, $c^{\bar{w}}_1,\ldots,c^{\bar{w}}_{b}$ are interpreted
as $r_{1}^{\langle\bar{w}\rangle},\ldots,r_{b}^{\langle\bar{w}\rangle}$.

\end{itemize}
\end{importantdef}
Keep in mind that $E_{\mathcal{A}}^{(s,h,p,b,l)}(G,\bar{a})$ contains $V(\bar{a})\cup\bigcup_{w\in[0,s-1]^h}\{r_{1}^{\langle\bar{w}\rangle},\ldots,r_{b}^{\langle\bar{w}\rangle}\}$ as roots.
The structure $E_{\mathcal{A}}^{(s,h,p,b,l)}(G,\bar{a})$, can be seen as a substructure of
$\mathsf{Layered}_{\mathcal{A}}^{(s,h,p,b)}({\sf ap}_l(G,\bar{a}))$,
intuitively obtained from the latter after focusing on the part ``cropped'' by the middle cycle of each $\mathcal{C}_{\bar{w}}$,
for every $\bar{w}\in [0,s-1]^h$.

\subsection{Layered annotated types}
\label{subsec:layeredtypes}
We now present the definition of layered annotated types for $\sigma^{\langle s,h,b\rangle}$-structures.
For $\sigma = \{E\}^{\langle  l \rangle}$ (see~\Cref{subsec_apices}),
these correspond to the formulas in the \textsl{Local-Global-Irrelevancy Condition} of the reduction between
$(\mathsf{CMSO/tw}\!+\!\mathsf{dp},\mathsf{hw})$ and $(\mathsf{CMSO},\mathsf{tw})$.

\paragraph{The set  of formulas $\Psi_{\bar{w}}^{m,d,t}[\sigma^{\langle s,h,b\rangle}]$.}
Let $\sigma$ be a colored-graph vocabulary.
Let $m,d,t\in\mathbb{N}$.
We describe some modifications on the set $\mathcal{L}^{m,d}_{t}[\sigma]$.
Let $s,b\in\mathbb{N}$ and set $h:=m+d$.
Given a $\bar{w}\in [0,s-1]^{h}$,
we define $\Psi_{\bar{w}}^{m,d,t}[\sigma^{\langle s,h,b\rangle}]$ to be the set obtained from $\mathcal{L}^{m,d}_{t}[\sigma]$ after modifying each formula $\psi(\bar{X},\bar{x})\in \mathcal{L}^{m,d}_{t}[\sigma]$ as follows:
\begin{itemize}
\item ask that $\bar{X}$ and $\bar{x}$ are in $Y_{\bar{w}}\cup V(\bar{c})$, where $V(\bar{c})$ are the constant symbols of $\sigma^{\langle s,h,b\rangle}$,
and
\item
replace every predicate $\mathsf{dp}(x_1,y_1,\ldots,x_k,y_k)$ in $\psi(\bar{X},\bar{x})$
with a formula $\xi(x_1,y_1,\ldots,x_k,y_k)\in \mathsf{CMSO}[\sigma^{\langle s,h,b\rangle}]$,
demanding that the vertex sets of the (vertex-disjoint) paths $P_1,\ldots,P_k$ should be subsets of $Y_{\bar{w}}\cup Z_{\bar{w}}\cup V(\bar{c})$ (see~\Cref{lemma_dpMSO}).
\end{itemize}
Observe that  $\Psi_{\bar{w}}^{m,d,t}[\sigma^{\langle s,h,b\rangle}]\subseteq \mathsf{CMSO}[\sigma^{\langle s,h,b\rangle}]$.\medskip

We now give the definition of the layered annotated types.
\medskip

\begin{importantdef}[Layered Annotated Types]
\label{def:ltp}
Let $m,d,r,s,t,b\in\mathbb{N}$ where $r\leq d$. Let
$\sigma$ be a colored-graph vocabulary, let $\mathfrak{H}$ be a $\sigma^{\langle s,m+d,b\rangle}$-structure, and let $R_1,\ldots,R_{m+r}\subseteq V(\mathfrak{H})$.
Given a $\bar{w}\in[0,s-1]^{m+d}$,
an $m$-tuple $\bar{V} \in (2^{V(\mathfrak{H})})^m$ and
a $d$-tuple $\bar{v}\in V(\mathfrak{H})^d$,
we define $$\mathsf{ltp}_{m,d,t}^{(0,r)}(\mathfrak{H},\bar{R},\bar{w},\bar{V},\bar{v})=\{\psi(\bar{X},\bar{x})\in \Psi_{\bar{w}}^{m,d,t}[\sigma^{\langle s,m+d,b\rangle}]
\mid \mathfrak{H}\models \psi(\bar{V},\bar{v})\}.$$
Given an $i\in[d]$,
a $\bar{w}\in[0,s-1]^{m+d-i}$,
an $m$-tuple $\bar{V}  \in (2^{V(\mathfrak{H})})^m$,
and
a $(d-i)$-tuple $\bar{v}\in V(\mathfrak{H})^{d-i}$,
we define
\[\mathsf{ltp}_{m,d,t}^{(i,r)}(\mathfrak{H},\bar{R},\bar{w},\bar{V},\bar{v})=\bigcup_{w\in[0,s-1]} \bigcup_{u\in A^{\bar{w}w}}
\{\mathsf{ltp}_{m,d,t}^{(i-1,r)}(\mathfrak{H},\bar{R},\bar{w}w,\bar{V},\bar{v}u)\},
\]
where for each $w\in[0,s-1]$,
$A^{\bar{w}w} =
Y_{\bar{w}w},$ if $i\leq d-r$, while otherwise
$A^{\bar{w}w}=Y_{\bar{w}w} \cap R_{m+d-(i-1)}$.
For $i=d$, we use $\mathsf{ltp}_{m,d,t}^{(d,r)}(\mathfrak{H},\bar{R},\bar{w},\bar{V})$ to denote $\mathsf{ltp}_{m,d,t}^{(d,r)}(\mathfrak{H},\bar{R},\bar{w},\bar{V},\bar{v})$, where $\bar{v}$ has length zero.
Given an $i\in[1+d,m+d]$, a $\bar{w}\in[0,s-1]^{m+d-i}$, and an $(m+d-i)$-tuple $\bar{V}\in (2^{V(\mathfrak{H})})^{m+d-i}$,
we define
\[\mathsf{ltp}_{m,d,t}^{(i,r)}(\mathfrak{H},\bar{R},\bar{w},\bar{V})=
\bigcup_{w\in[0,s-1]}
\bigcup_{U\subseteq Y_{\bar{w}w}\cap R_{m+d-(i-1)}}\{\mathsf{ltp}_{m,d,t}^{(i-1,r)}(\mathfrak{H},\bar{R},\bar{w}w,\bar{V}U)\}.
\]
For $i=m+d$, we set $\mathsf{ltp}_{m,d,t}^{(m+d,r)}(\mathfrak{H},\bar{R}):=\mathsf{ltp}_{m,d,t}^{(m+d,r)}(\mathfrak{H},\bar{R},\bar{w},\bar{V})$, where $\bar{w},\bar{V}$ have length zero.
\end{importantdef}

As in~\Cref{prop:patterns}, it is easy to observe that the number of different layered annotated types is bounded and that they can be defined in $\mathsf{CMSO}$.
 \begin{lemma}\label{lemma_number_partial_patterns}
There is some function $\newfun{fun:partial-patterns}:\mathbb{N}^7 \to \mathbb{N}$ such that for every $m,d,r,s,t,b\in\mathbb{N}$, where $r\leq d$, every colored-graph vocabulary $\sigma$,
and every $i\in[m+d]$, the set
\begin{equation*}
\mathcal{D}_{m,d,r,s,t,b}^{(i)}(\sigma):=
\left\{\!
\mathsf{ltp}_{m,d,t}^{(i,r)}(\mathfrak{H},\bar{R},\bar{w},\bar{V},\bar{v})~\left|~
\begin{aligned}
&  \mathfrak{H}\in\mathbb{STR}[\sigma^{\langle s,m+d,q\rangle}],\\
& \bar{R}\subseteq V(\mathfrak{H})^{m+r},\\
& \bar{w}\in[0,s-1]^{m+d-i},\\
& \bar{V}\subseteq V(\mathfrak{H})^{\min\{m,m+d-i\}},\\
& \bar{v}\in V(\mathfrak{H})^{\max\{0,d-i\}}
\end{aligned}\right.\right\}
\end{equation*}
has size at most $\funref{fun:partial-patterns}(i,m,d,s,t,b,|\sigma|)$.
Moreover,  if $P_1,\ldots,P_{c}$ is an enumeration of all sets in $\mathcal{D}_{m,d,r,s,t,b}^{(m+d)}(\sigma)$,  where $c\in\{1,\ldots,\funref{fun:partial-patterns}(m+d,m,d,s,t,b,|\sigma|)\}$,
then there exists a set $\Psi_{m,d,s,t,b}^\sigma:=\{\beta_1,\ldots,\beta_c\}\subseteq \mathsf{CMSO}[\sigma^{\langle s,m+d,b\rangle}\cup\{R_1,\ldots,R_{m+r}\}]$
such that for every $i\in\{1,\ldots,c\}$, every $\sigma^{\langle s,m+d,b\rangle }$-structure $\mathfrak{H}$ and every $\bar{R}\subseteq V(\mathfrak{H})^{m+r}$,
$(\mathfrak{H},\bar{R})\models \beta_i \iff \mathsf{ltp}_{m,d,t}^{(m+d,r)}(\mathfrak{H},\bar{R}) = P_i.$
\end{lemma}

\section{Local-global-irrelevance}\label{sec_exchangability}
The goal of this section is to present the proof of the \textsl{Local-Global-Irrelevancy Condition} of the reduction
between
$(\mathsf{CMSO/tw}\!+\!\mathsf{dp},\mathsf{hw})$ and $(\mathsf{CMSO},\mathsf{tw})$ (see~\Cref{def_red}).

\subsection{Item (a) of the Local-Global-Irrelevancy Condition}
\label{subsec:itema}

\paragraph{The algorithm $\mathbb{A}$ of the \textsl{Local-Global-Irrelevancy Condition}.}
Let us describe the algorithm $\mathbb{A}$ of the \textsl{Local-Global-Irrelevancy Condition}.
It finds a structure of the form $E_{\mathcal{A}}^{(s,h,p,b,l)}(G,\bar{a})$, for appropriate values of $s,h,p,b,l$
defined as follows.
Let $m,r,t\in\mathbb{N}$.
We set $l:=\funref{label_hierarchical}(\mathsf{hw}(G))$,
$d:=r+\funref{@transversales}(r,l)$,
$s:=\funref{fun:bid}(t)$,
$h:=m+d$,
$p :=\funref{@aristocracias}(d)$,
and $b:=\funref{fun_ulinkastar}(d)$.
The algorithm $\mathbb{A}$ first calls the algorithm of~\Cref{prop:flatann}
with input $(G,s^h\cdot p,\mathsf{hw}(G))$,
which outputs
\begin{itemize}
\item either a report that $\tw(G)\leq c_{m,r,t,l}$, for some $c_{m,r,t,l}$ depending on $m,r,t,$ and $l$,
or
\item an apex set $A$ of size at most $l$ and a well-aligned $(s^h\cdot p,s^h\cdot p)$-railed annulus flatness pair $(\mathcal{A},\mathfrak{R})$ of $G\setminus A$,
where $\tw(G[V(\mathsf{Compass}_\mathfrak{R}(\mathcal{A}))\cup\mathsf{inn}_\mathfrak{R}(\mathcal{A})])\leq c_{m,r,t,l}$.
\end{itemize}
In the latter case, it considers an arbitrary ordering $\bar{a}$ of $A$ (if $|A|<l$, enhance $\bar{a}$ by adding $l- |A|$ copies of an arbitrary vertex of $G$ and update $(\mathcal{A},\mathfrak{R})$ accordingly).
Using $(\mathcal{A},\mathfrak{R})$ and $\bar{a}$, we construct  $E_{\mathcal{A}}^{(s,h,p,q,l)}(G,\bar{a})$ (in linear time).

\paragraph{Bounding the treewidth of $\mathsf{Gaifman}(E_{\mathcal{A}}^{(s,h,p,b,l)}(G,\bar{a}))$.}
The bound on the treewidth of the Gaifman graph of $E_{\mathcal{A}}^{(s,h,p,q,l)}(G,\bar{a})$,
follows from the following observation, which is based on the fact that when considering the union of a graph $G$ and the $\mathfrak{R}$-leveling of $\mathcal{A}$ in $G$ (see~\Cref{subsec_levelings} for the definition), the treewidth of the obtained graph is at most the treewidth of $G$ plus one.

\begin{observation}\label{obs:enha}
Let $s,h,l\in\mathbb{N}$, let $p,q\in\mathbb{N}_{\geq 3}$, let $b\in[q]$,
let $G$ be a graph,
let $\bar{a}\in V(G)^l$,
and let $(\mathcal{A},\mathfrak{R})$ be a well-aligned
$(s^h \cdot p,q)$-railed annulus flatness pair of $G\setminus V(\bar{a})$.
It holds that $\tw(\mathsf{Gaifman}(E_{\mathcal{A}}^{(s,h,p,b,l)}(G,\bar{a})))\leq \tw(G[V(\mathsf{Compass}_\mathfrak{R}(\mathcal{A}))\cup\mathsf{inn}_\mathfrak{R}(\mathcal{A})])+1$.
\end{observation}

\subsection{Item (b) of the Local-Global-Irrelevancy Condition}
This subsection contains the proof of the following lemma,
which implies item (b) of the~\textsl{Local-Global-Irrelevancy Condition} of the reduction
between
$(\mathsf{CMSO/tw}\!+\!\mathsf{dp},\mathsf{hw})$ and $(\mathsf{CMSO},\mathsf{tw})$ (see~\Cref{def_red}).
Recall that $\Psi_{m,d,s,t,b}^\sigma$ is the set of formulas from~\Cref{lemma_number_partial_patterns} (of vocabulary $\sigma_{m,r,t,l}\cup\{R_1,\ldots, R_{m+r}\}
$)
and $\Phi_{m,r,t}$ is the set of formulas from the Definability Condition (see~\Cref{prop:patterns} in~\Cref{subsec:anntypes}).

\begin{lemma}\label{lem_equirep}
Let $l,m,r,t\in\mathbb{N}$.
We set $d:=r+\funref{@transversales}(r,l)$,
$s:=\funref{fun:bid}(t)$,
$h:=m+d$,
$p :=\funref{@aristocracias}(d)$,
and $b:=\funref{fun_ulinkastar}(d)$.
Also, let $G$ be a graph,
let $\bar{R}:=(R_1,\ldots, R_{m+r})\in(2^{V(G)})^{m+r}$, let $\bar{a}\in V(G)^l$,
and let $(\mathcal{A},\mathfrak{R})$ be a well-aligned $(s^h\cdot p,q)$-railed annulus flatness pair of $G\setminus V(\bar{a})$,
where $q\geq 5/2\cdot \funref{fun_ulinkastar}(d)$.
For every $S_0,\ldots, S_{m+r}\subseteq \mathsf{inn}_{\mathfrak{R}}(\mathcal{A})$,
\begin{eqnarray*}
& & \text{if $(S_0,\ldots, S_{m+r})$ is $\Psi_{m,r,t,l}$-irrelevant for $(E_{\mathcal{A}}^{(s,h,p,b,l)}(G,\bar{a}),\bar{R}^{\star})$,}\\
& & \text{then $(S_0,\ldots, S_{m+r})$ is $\Phi_{m,r,t}$-irrelevant for $(G,\bar{R})$,}
\end{eqnarray*}
where $\bar{R}^\star:= (R_1\cap V(E_{\mathcal{A}}^{(s,h,p,b,l)}(G,\bar{a})),\ldots, R_{m+r} \cap V(E_{\mathcal{A}}^{(s,h,p,b,l)}(G,\bar{a})))$ and $\Psi_{m,r,t,l}:=\Psi_{m,d,s,t,b}^{\{E\}^{\langle l \rangle}}$.
\end{lemma}

The proof of~\Cref{lem_equirep} follows from the next result (\Cref{lemma:intermediate1}), combined with~\Cref{lem_inter2equiv}. Before stating~\Cref{lemma:intermediate1}, we provide some additional definitions and notational conventions needed in the proof of~\Cref{lemma:intermediate1}.

\paragraph{Traces of tuples of sets (of bounded annotated treewidth) and vertices.}
Let $m,r,d,t,l\in\mathbb{N}$, where $r\leq d$, and set $h:=m+d$.
Let $G$ be a graph, let $\bar{a}\in V(G)^l$,
and let $(\mathcal{A},\mathfrak{R})$ be a well-aligned $(\funref{fun:bid}(t)^{d}\cdot p,q)$-railed annulus flatness pair of $G\setminus V(\bar{a})$, for some $p,q\geq 3$.
Also, let $\ell_1,\ell_2\in\mathbb{N}$ such that $\ell_1+\ell_2\leq h$.
Given a $\bar{V} = (V_1,\ldots,V_{\ell_1})\in (2^{V(G)})^{\ell_1}$,
and $\bar{v}\in V(G)^{\ell_2}$,
where for each $j\in[\ell_1]$, $\tw(G,V_j)\leq t$,
we define the \emph{trace of $\bar{V},\bar{v}$},
denoted by ${\sf trace}(\bar{V},\bar{v})$,
to be the $h$-tuple $w_1\ldots w_{h}\in [0,\funref{fun:bid}(t)-1]^h$ inductively defined as follows:
for every $i\in[\ell_1]$, we set $w_i$ to be the maximum index in $[0,\funref{fun:bid}(t)-1]$ such that
$V_i\cap V(\cupall\mathsf{Influence}_\mathfrak{R}(\mathcal{C}_{w_1\ldots w_i}))=\emptyset$.
Also, for every $i\in[\ell_2]$,
we set $w_{\ell_1+i}$ to be the maximum index in $[0,\funref{fun:bid}(t)-1]$ such that
$\{v_i\}\cap V(\cupall\mathsf{Influence}_\mathfrak{R}(\mathcal{C}_{w_1\ldots w_{m+i}}))=\emptyset$.
Note that such indices always exist in $[0,\funref{fun:bid}(t)-1]$ because of~\Cref{lemma:buffer}.

Given a $\bar{V} = (V_1,\ldots,V_{\ell_1})\in (2^{V(G)})^{\ell_1}$,
and $\bar{v}\in V(G)^{\ell_2}$,
where for each $j\in[\ell_1]$, $\tw(G,V_j)\leq t$,
if $w_1\ldots w_h$ is the ${\sf trace}(\bar{V},\bar{v})$, then we set
$V_{i}^{\sf in}:=V_i\cap V_{\mathcal{A}_{w_1\ldots w_i}}^{\mathsf{inn}}$
and
$v_{i}^{\sf in}:=v_i,$ if $v_i\in  V_{\mathcal{A}_{w_1\ldots w_{m+i}}}^{\mathsf{inn}}$, while $v_i:=\mathspace$, if otherwise.

\begin{lemma}\label{lemma:intermediate1}
Let $l,m,d,r,t\in\mathbb{N}$, where
$r\leq d$.
We set
$s:=\funref{fun:bid}(t)$,
$h:=m+d$,
$p :=\funref{@aristocracias}(d)$,
and $b:=\funref{fun_ulinkastar}(d)$.
Also, let $G$ be a graph,
let $R_1,\ldots, R_{m+r}\subseteq V(G)$, let $\bar{a}\in V(G)^l$,
and let $(\mathcal{A},\mathfrak{R})$ be a well-aligned $(s^h \cdot p,q)$-railed annulus flatness pair of $G\setminus V(\bar{a})$,
where $q\geq 5/2\cdot \funref{fun_ulinkastar}(d)$.
For every $S_0,\ldots, S_{m+r}\subseteq \mathsf{inn}_{\mathfrak{R}}(\mathcal{A})$,
\begin{eqnarray*}
& & \text{if $(S_0,\ldots, S_{m+r})$ is $\Psi_{m,r,t,l}$-irrelevant for $(E_{\mathcal{A}}^{(s,h,p,b,l)}(G,\bar{a}),\bar{R}^{\star})$,}\\
& & \text{then
$\mathsf{tp}_{m,d,t}^{(m+d,r)}
({\sf ap}_{l}(G,\bar{a}),R_1,\ldots, R_{m+r})
 = \mathsf{tp}_{m,d,t}^{(m+d,r)}
({\sf ap}_{l}(G\setminus S_0,\bar{a}),R_1\setminus S_1,\ldots, R_{m+r}\setminus S_{m+r}),
$}
\end{eqnarray*}
where $\bar{R}^\star:= (R_1\cap V(E_{\mathcal{A}}^{(s,h,p,b,l)}(G,\bar{a})),\ldots, R_{m+r} \cap V(E_{\mathcal{A}}^{(s,h,p,b,l)}(G,\bar{a}))$ and $\Psi_{m,r,t,l}:=\Psi_{m,d,s,t,b}^{\{E\}^{\langle l \rangle}}$.
\end{lemma}

\begin{proof}
For every $i\in[m+r]$, we set $\bar{R}^\star_i:=R_i\cap V(E_{\mathcal{A}}^{(s,h,p,b,l)}(G,\bar{a}))$.
Let $S_0,\ldots,S_{m+r}\subseteq \mathsf{inn}_{\mathfrak{R}}(\mathcal{A})$ such that $(S_0,\ldots, S_{m+r})$ is $\Psi_{m,r,t,l}$-irrelevant for $(E_{\mathcal{A}}^{(s,h,p,b,l)}(G,\bar{a}),\bar{R}^{\star})$.
We also set $G' = G\setminus S_0$ and for every $i\in[m+r]$, we set $R_i ' = R_i\setminus S_i$.

We first consider the
$\{E\}^{\langle l\rangle}$-structures
${\sf ap}_{l}(G,\bar{a})$ and ${\sf ap}_{l}(G',\bar{a})$.
We use $G_{\sf ap}$ (resp. $G_{\sf ap}'$) to denote the Gaifman graph of ${\sf ap}_{l}(G,\bar{a})$ (resp. ${\sf ap}_{l}(G',\bar{a})$), i.e., the graph whose vertex set is $V(G)$ (resp. $V(G')$)
and whose edge set is the set of edges of
${\sf ap}_{l}(G,\bar{a})$ (resp. ${\sf ap}_{l}(G',\bar{a})$).
Observe that $G\setminus V(\bar{a})$ and $G_{\sf ap}\setminus V(\bar{a})$ are the same graph, and that $(\mathcal{A},\mathfrak{R})$ is also a railed annulus flatness pair of $G_{\sf ap}\setminus V(\bar{a})$.
Also,
since in $G_{\sf ap}$
there are no edges between $V(\bar{a})$ and
$V(G_{\sf ap}\setminus V(\bar{a}))$,
we can update $\mathfrak{R}$ by adding $V(\bar{a})$ to
$(V(G)\setminus V(\bar{a}))\setminus (\mathsf{Compass}_\mathfrak{R}(\mathcal{A}) \cup \mathsf{inn}_\mathfrak{R}(\mathcal{A}))$,
and observe that
after this modification of $\mathfrak{R}$,
$(\mathcal{A},\mathfrak{R})$ is a railed annulus flatness pair of $G_{\sf ap}$.
For the same reason and since $S_0\subseteq \mathsf{inn}_{\mathfrak{R}}(\mathcal{A})$, we can assume that $(\mathcal{A},\mathfrak{R})$ is a railed annulus flatness pair of $G_{\sf ap}'$.
In the rest of the proof, we will write  $V(G)$ to refer to $V({\sf ap}_{l}(G,\bar{a}))$, since these two sets are equal by definition.
Similarly, we write $V(G')$ to denote $V({\sf ap}_{l}(G',\bar{a}))$.
Also, we use the following notation:
\[\mathfrak{L}:=(E_{\mathcal{A}}^{(s,h,p,b,l)}(G,\bar{a}),R_1^\star,\ldots, R_{m+r}^\star)\text{ and }\mathfrak{L}':=(E_{\mathcal{A}}^{(s,h,p,b,l)}(G',\bar{a}),R_1^\star\setminus S_1,\ldots, R_{m+r}^\star\setminus S_{m+r})\]
and
\[\mathfrak{H}:=({\sf ap}_{l}(G,\bar{a}),R_1,\ldots,R_{m+r})\text{ and }\mathfrak{H}' :=
({\sf ap}_{l}(G',\bar{a}),R_1',\ldots,R_{m+r}').\]
We aim to show that
$\mathsf{ltp}_{m,d,t}^{(m+d,r)}
(\mathfrak{L})
=
\mathsf{ltp}_{m,d,t}^{(m+d,r)}
(\mathfrak{L}')
\implies
\mathsf{tp}_{m,d,t}^{(m+d,r)}
(\mathfrak{H})
=
\mathsf{tp}_{m,d,t}^{(m+d,r)}
(\mathfrak{H}').
$
For this, we assume that $\mathsf{ltp}_{m,d,t}^{(m+d,r)}
(\mathfrak{L})
=
\mathsf{ltp}_{m,d,t}^{(m+d,r)}
(\mathfrak{L}').$

We prove inductively, for increasing values of $i$, that the following two statements hold for every $i\in[m+d]$, where $j:=\min\{m,m+d-i\}$ and $\ell:=\max\{0,d-i\}$:

\begin{itemize}
\item[(i)]
Let a $j$-tuple
$\bar{V}=(V_1,\ldots, V_j)\in (2^{V(G)})^{j}$,
where for each $h\in[j]$, $V_h\subseteq R_h$ and $\tw(G,V_h)\leq t$, and
let an $\ell$-tuple $\bar{v} = (v_1,\ldots,v_\ell) \in V(G)^\ell$,
where for every $h\in[\min\{\ell,r\}]$, $v_h\in R_{m+h}$.
There is a $j$-tuple
$\bar{U}^{\sf in} = (U_1^{\sf in},\ldots, U_j^{\sf in})\in (2^{V(G')})^j$,
where for each $h\in[j]$,
$U_h^{\sf in}\subseteq R_h^\star\setminus S_h$,
and an $\ell$-tuple $\bar{u}^{\sf in}=(u_1^{\sf in},\ldots,u_\ell^{\sf in}) \in (V(G')\cup\{\mathspace\})^\ell$,
where for every $h\in[\min\{\ell,r\}]$, $u_h^{\sf in}\in R_{m+h}^\star\setminus S_{m+h}$,
such that
\[\mathsf{ltp}_{m,d,t}^{(i,r)}
(\mathfrak{L},\bar{w},\bar{V}^{\sf in},\bar{v}^{\sf in})=
\mathsf{ltp}_{m,d,t}^{(i,r)}
(\mathfrak{L}',\bar{w},\bar{U}^{\sf in},\bar{u}^{\sf in})\implies\mathsf{tp}_{m,d,t}^{(i,r)}(
\mathfrak{H},\bar{V},\bar{v})
= \mathsf{tp}_{m,d,t}^{(i,r)}(
\mathfrak{H}',\bar{U},\bar{u}),\]
where $\bar{w}:=\mathsf{trace}(\bar{V},\bar{v})$
and, for every $h\in[j]$, $U_h = U_h^{\sf in} \cup (V_h \setminus V_h^{\mathsf{in}})$
and for every $h\in[\ell]$, $u_h =v_h,$ if $u_h^{\sf in} =  \mathspace$ and $u_h = u_h^{\mathsf{in}},$ if otherwise.

\item[(ii)]
Let a $j$-tuple
$\bar{U}=(U_1,\ldots, U_j)\in (2^{V(G')})^{j}$,
where for each $h\in[j]$, $U_h\subseteq R_h'$ and $\tw(G',U_h)\leq t$, and
let an $\ell$-tuple $\bar{u} = (u_1,\ldots,u_\ell) \in V(G')^\ell$,
where for every $h\in[\min\{\ell,r\}]$, $u_h\in R_{m+h}'$.
There exists a $j$-tuple
$\bar{V}^{\sf in} = (V_1^{\sf in},\ldots, V_j^{\sf in})\in (2^{V(G)})^j$,
where for each $h\in[j]$,
$V_h^{\sf in}\subseteq R_h^\star$,
and an $\ell$-tuple $\bar{v}^{\sf in}=(v_1^{\sf in},\ldots,v_\ell^{\sf in}) \in (V(G)\cup\{\mathspace\})^\ell$,
where for every $h\in[\min\{\ell,r\}]$, $v_h^{\sf in}\in R_{m+h}^\star$,
such that
\[\mathsf{ltp}_{m,d,t}^{(i,r)}
(\mathfrak{L},\bar{w},\bar{V}^{\sf in},\bar{v}^{\sf in})=
\mathsf{ltp}_{m,d,t}^{(i,r)}
(\mathfrak{L}',\bar{w},\bar{U}^{\sf in},\bar{u}^{\sf in})\implies\mathsf{tp}_{m,d,t}^{(i,r)}(
\mathfrak{H},\bar{V},\bar{v})
= \mathsf{tp}_{m,d,t}^{(i,r)}(
\mathfrak{H}',\bar{U},\bar{u}),\]
where $\bar{w}:=\mathsf{trace}(\bar{U},\bar{u})$
and, for every $h\in[j]$, $V_h = V_h^{\sf in} \cup (U_h \setminus U_h^{\mathsf{in}})$
and for every $h\in[\ell]$, $v_h =u_h,$ if $v_h^{\sf in} =  \mathspace$ and $v_h = v_h^{\mathsf{in}},$ if otherwise.
\end{itemize}
We start by proving (i) and (ii) for $i=0$.
We only prove (i), the proof of (ii) being completely symmetrical.
Let an $m$-tuple
$\bar{V} = (V_1,\ldots, V_m)\in 2^{R_1}\times \cdots \times 2^{R_m}$,
where for every $h\in[m]$, $\tw(G,V_h)\leq t$,
and let a $d$-tuple $\bar{v} = (v_1,\ldots,v_d) \in V(G)^d$,
where for every $h\in[r]$, $v_h\in R_{m+h}$.
Since we assumed that
$\mathsf{ltp}_{m,d,t}^{(m+d,r)}
(\mathfrak{L})
=
\mathsf{ltp}_{m,d,t}^{(m+d,r)}
(\mathfrak{L}'),$
there exist
an $m$-tuple
$\bar{U}^{\mathsf{in}} = (U_1^{\mathsf{in}},\ldots, U_m^{\mathsf{in}})$
where for every $i\in[m]$, $U_i^{\mathsf{in}}\in 2^{R_i^\star\setminus S_i}$,
and a $d$-tuple $\bar{u}^{\mathsf{in}} = (u_1^{\mathsf{in}},\ldots,u_d^{\mathsf{in}}) \in V(\mathfrak{L}')\cup\{\mathspace\})^d$,
where for every $i\in[r]$, $u_h^{\mathsf{in}}\in R_{m+h}^\star\setminus S_{m+h}$
and
\begin{equation}
\mathsf{ltp}_{m,d,t}^{(0,r)}
(\mathfrak{L},\bar{w},\bar{V}^{\mathsf{in}}, \bar{v}^{\sf in}) =  \mathsf{ltp}_{m,d,t}^{(0,r)}
(\mathfrak{L}',\bar{w},\bar{U}^{\sf in},\bar{u}^{\sf in}).\label{eq_0partialpatt}
\end{equation}
For every $i\in[m]$, we set $U_i := U_i^{\sf in} \cup (V_i \setminus V_i^{\mathsf{in}})$
and for every $i\in[d]$, we set $u_i :=v_i,$ if $u_i^{\sf in} =  \mathspace$ and $u_i := u_i^{\mathsf{in}},$ if otherwise.
We now show that
$ \mathsf{ltp}_{m,d,t}^{(0,r)}
(\mathfrak{H},\bar{V},\bar{v})
=  \mathsf{ltp}_{m,d,t}^{(0,r)}(\mathfrak{H},\bar{U},\bar{u}).$

Recall that $\mathfrak{H}=({\sf ap}_{l}(G,\bar{a}),R_1,\ldots,R_{m+r})\text{ and }\mathfrak{H}' =
({\sf ap}_{l}(G',\bar{a}),R_1',\ldots,R_{m+r}').$
In order to prove that $ \mathsf{ltp}_{m,d,t}^{(0,r)}
(\mathfrak{H},\bar{V},\bar{v})
=  \mathsf{ltp}_{m,d,t}^{(0,r)}(\mathfrak{H},\bar{U},\bar{u})$,
we have to show that
for every $\psi(\bar{X},\bar{x}) \in \mathcal{L}_{t}^{m,d}[\{E\}^{\langle l \rangle}]$,
${\sf ap}_{l}(G,\bar{a})\models \psi(\bar{V},\bar{v}) \iff {\sf ap}_{l}(G',\bar{a})\models \psi(\bar{U},\bar{u})$.
Recall that the formulas in $\mathcal{L}_{t}^{m,d}[\{E\}^{\langle l \rangle}]$
are boolean combinations of the following atomic formulas (in this setting, first-order terms are either first-order variables or among the constant symbols $c_1,\ldots,c_l$, the latter being interpreted as $a_1,\ldots,a_l$ in $\mathsf{ap}_l(G,\bar{a})$):
\begin{itemize}
\item $x=y$, where $x$ and $y$ are first-order terms,
\item $x\in C_i$ for some unary relation symbol $C_i\in \{E\}^{\langle l \rangle}$ and some first-order term  $x$,
\item $x\in X_i$ for some set variable $X_i$ and first-order term $x$,
\item $\{x,y\}\in E$, where $x$ and $y$ are first-order terms,
\item $\mathsf{card}_p(X)$, for $p\in[2,\hat{c}]$ and some set variable $X$, where $\hat{c}$ is a fixed constant\footnote{In the
the last paragraph of~\Cref{subsec:atp}, we assumed an upper bound $\hat{c}$ on the maximum $p$ for which $\mathsf{card}_p(X)$ appears in the formulas of this paper.},
and
\item $\mathsf{dp}_k(x_1,y_1,\ldots,x_k,y_k)$ for $k\in\mathbb{N}$ and some first-order terms $x_1,y_1,\ldots,x_k,y_k$,
\end{itemize}
with the additional property that if ${\sf ap}_{l}(G,\bar{a})\models \psi(\bar{V},\bar{v})$ then for every $i\in[m]$, $\tw(G,V_i)\leq t$.
We will prove that each of the above atomic formulas is satisfied if we replace ${\sf ap}_{l}(G,\bar{a})$ with ${\sf ap}_{l}(G',\bar{a})$ and $\bar{V},\bar{v}$ with $\bar{U},\bar{u}$.

For the first four cases the statement is implied as an easy observation by~\eqref{eq_0partialpatt} and by the definition of $\bar{U}$ and $\bar{u}$.
To show that for every $i,j\in[d]$, $v_i$ and $v_j$ are adjacent in ${\sf ap}_{l}(G,\bar{a})$ if and only if $u_i$ and $u_j$
are adjacent in ${\sf ap}_{l}(G',\bar{a})$, observe that
$v_1,\ldots, v_d, u_1,\ldots, u_d$ do not belong to $V(\cupall{\sf Influence}_{\mathfrak{R}}(\mathcal{C}_{\bar{w}}))$, thus they belong either to $V_{\mathcal{A}_{\bar{w}}}^{\mathsf{inn}}$
or to $(V(G)\setminus V(\bar{a}))\setminus (V(\cupall{\sf Influence}_{\mathfrak{R}}(\mathcal{C}_{\bar{w}})) \cup V_{\mathcal{A}_{\bar{w}}}^{\mathsf{inn}})$, and there is no edge in
${\sf ap}_{l}(G,\bar{a})$ between these two sets.
For the predicate $\mathsf{card}_p(X)$, observe that for every $i\in[m]$,
by~\eqref{eq_0partialpatt}, for every $p\in[2,\hat{c}]$, $|V_i^{\sf in}|$ is a multiple of $p$ if and only if $|U_i^{\sf in}|$ is a multiple of $p$. Therefore, since $U_i = U_i^{\sf in} \cup (V_i \setminus V_i^{\mathsf{in}})$, we have that for every $p\in[2,\hat{c}]$,
$|V_i|$ is a multiple of $p$ if and only if $|U_i|$ is a multiple of $p$.

To conclude the proof of the claim, we show that the satisfaction of the disjoint-paths predicates is also preserved when shifting from $({\sf ap}_{l}(G,\bar{a}),v_1,\ldots,v_d)$ to $({\sf ap}_{l}(G',\bar{a}),u_1,\ldots,u_d)$.
Let $s_1,t_1,\ldots, s_k,t_k\in \{v_1,\ldots,v_d,a_1,\ldots,a_l\}$ and
$s_1',t_1',\ldots,s_k',t_k'\in\{u_1,\ldots,u_d,a_1,\ldots,a_l'\}$.
We will show that
$$G_{\sf ap}\models {\sf dp}_{k} (s_1,t_1,\ldots, s_k,t_k)
\iff
G_{\sf ap}'\models{\sf dp}_k (s_1',t_1',\ldots,s_k',t_k').$$
Since there are no edges in $G_{\sf ap}$ (resp. $G_{\sf ap}'$) between vertices of $V(\bar{a})$ and $V(G_{\sf ap})\setminus V(\bar{a})$ (resp. $V(G_{\sf ap}')\setminus V(\bar{a})$), it suffices to prove that
\begin{itemize}
\item[(a)]
for every
linkage $L$ of $G_{\sf ap}$
such that $T(L)\subseteq \{v_1,\ldots,v_d\}$
there is a linkage $L'$ of $G_{\sf ap}'$ such that
$T(L')\subseteq \{u_1,\ldots,u_d\}$ and
$P(L')$ is obtained from $P(L)$ after replacing, in each $\{z,w\}\in P(L)$,
every occurrence of $v_i$ by $u_i$, for every $i\in[d]$, and
\item[(b)]
for every
linkage $L'$ of $G_{\sf ap}'$ such that
$T(L')\subseteq\{u_1,\ldots,u_d\}$,
there is a linkage $L$ of $G_{\sf ap}$ such that
$T(L)\subseteq \{v_1,\ldots,v_d\}$ and
$P(L)$ is obtained from $P(L')$ after replacing, in each $\{z,w\}\in P(L')$,
every occurrence of $u_i$ by $v_i$, for every $i\in[d]$.
\end{itemize}

To show item~(a),
let $L$ be a linkage of $G_{\sf ap}$.
By~\Cref{lem_colomodelsrerout},
there is a linkage $\tilde{L}_1$ of $G_{\bar{w}}$,
a linkage $\tilde{L}_2$ of $G_{\bar{w}}^{\sf out}$,
and a set $Z\subseteq \{r_{1}^{\langle\bar{w}\rangle},\ldots,r_{\funref{fun_ulinkastar}(d)}^{\langle\bar{w}\rangle}\}$  such that
\begin{itemize}
\item $T(\tilde{L}_1)  =  (T(L)\cap V(G_{\bar{w}})) \cup Z$,
\item $T(\tilde{L}_2)  =  (T(L)\cap V(G_{\bar{w}}^{\sf out})) \cup Z$,
and
\item $\tilde{L}_1 \cup \tilde{L}_2$ is a linkage of ${\sf Leveling}_{(\mathcal{A}, \mathfrak{R})}^{\langle \mathcal{C}_{\bar{w}}\rangle}(G_{\sf ap})$ that is equivalent to $L$.
\end{itemize}
It is easy to see that, by~\eqref{eq_0partialpatt},
there is a linkage $\tilde{L}_1'$ of $G_{\bar{w}}'$
such that $P(\tilde{L}_1')$ is obtained from $P(\tilde{L}_1)$ after replacing, in each $\{z,w\}\in P(\tilde{L}_1)$,
every occurrence of $v_i^{\sf in}$ by $u_i^{\sf in}$, for every $i\in[d]$.
Also, observe that $G_{\bar{w}}^{\mathsf{out}}$ and $G_{\bar{w}}^{' \mathsf{out}}$ are equal.
Therefore, the graph $\tilde{L}':=\tilde{L}_1'\cup \tilde{L}_2$ is a linkage of  ${\sf Leveling}_{(\mathcal{A}, \mathfrak{R})}^{\langle \mathcal{C}_{\bar{w}}\rangle}(G_{\sf ap}')$ and $P(L')$ can be obtained from $P(L)$ after replacing,
in each $\{z,w\}\in P(L)$,
every occurrence of $v_i$ by $u_i$, for every $i\in[d]$.
By~\Cref{lem_levelingpaths}, there is a linkage $L'$ of $G_{\sf ap}'$ such that
$L'\equiv \tilde{L}'$.
Therefore, $P(L')$ is obtained from $P(L)$ after replacing in each $\{z,w\}\in P(L)$,
every occurrence of $v_i$ by $u_i$, for every $i\in[d]$.
Thus, item~(a) follows.
The proof for item~(b) is completely symmetrical to the one for item~(a), i.e., it is obtained by interchanging $L$ with $L'$, $G_{\bar{w}}$ with $G_{\bar{w}}'$,
and $v_i$ with $u_i$ for every $i\in[d]$.
\medskip

To conclude the proof by induction, let us
show how to prove (i) for every $i\geq 1$, assuming that (i) and (ii) hold for $i-1$.
Using the assumption that $\mathsf{ltp}_{m,d,t}^{(m+d,r)}
(\mathfrak{L})
=
\mathsf{ltp}_{m,d,t}^{(m+d,r)}
(\mathfrak{L}'),$
we can find a $j$-tuple $\bar{U}^{\mathsf{in}}$ and an $\ell$-tuple $\bar{u}^{\sf in}$ in $G'$, such that
\begin{equation}
\mathsf{ltp}_{m,d,t}^{(i,r)}
(\mathfrak{L},\bar{w},\bar{V}^{\mathsf{in}}, \bar{v}^{\sf in})
 = \mathsf{ltp}_{m,d,t}^{(i,r)}
(\mathfrak{L}',\bar{w},\bar{U}^{\sf in},\bar{u}^{\sf in}).\label{eq_ipartialpatt}
\end{equation}
By considering $V_{j+1}$ or $v_{\ell+1}$ in $G$ (depending on whether $i$ is larger than $d$ or not),
we can define  $V_{j+1}^{\mathsf{in}}$ or $v_{\ell+1}^{\mathsf{in}}$ accordingly, and use~\eqref{eq_ipartialpatt} in order to find $U_{j+1}^{\mathsf{in}}$ or $u_{\ell+1}^{\mathsf{in}}$ in $G'$ such that the obtained layered annotated $(i-1)$-types are the same.
By the induction hypothesis,
equality of the layered annotated $(i-1)$-types implies equality of the annotated $(i-1)$-types. This, in turn, implies equality of annotated $i$-types.
The proof of the lemma is concluded by observing that, since (i) and (ii) hold for $i=m+d$,
then $\mathsf{ltp}_{m,d,t}^{(m+d,r)}
(\mathfrak{L})
=
\mathsf{ltp}_{m,d,t}^{(m+d,r)}
(\mathfrak{L}')
\implies
\mathsf{tp}_{m,d,t}^{(m+d,r)}(\mathfrak{H})
=
\mathsf{tp}_{m,d,t}^{(m+d,r)}
(\mathfrak{H}').$
\end{proof}

The proof of~\Cref{lem_equirep} can be reproduced in order to give the following result.

\begin{lemma}\label{lemmabis}
Let $l,m,r,t\in\mathbb{N}$ and let $d=r+\funref{@transversales}(r,l)$.
Let $\sigma$ be a colored-graph vocabulary.
Also, let $\mathfrak{G}$ be a $\sigma^{s,m+d}$-structure (whose Gaifman graph is $G$),
let $R_1,\ldots, R_{m+r}\subseteq V(G)$, let $\bar{a}\in V(G)^l$,
and $(\mathcal{A},\mathfrak{R})$ be a well-aligned $(\funref{fun:bid}(t)^{m+d}\cdot\funref{@aristocracias}(d),q)$-railed annulus flatness pair of $G\setminus V(\bar{a})$,
where $q\geq 5/2\cdot \funref{fun_ulinkastar}(d)$.
For every $S_0,\ldots, S_{m+r}\subseteq \mathsf{inn}_{\mathfrak{R}}(\mathcal{A})$
such that for every $i\in[m+r]$, $S_i\subseteq R_i$ and $R_i\setminus S_i \subseteq V(G)\setminus S_0$ and
\begin{eqnarray*}
&\mathsf{ltp}_{m,d,t}^{(m+d,r)}
(E_{\mathcal{A}}^{(s,h,p,b,l)}(G,\bar{a}),R_1^\star,\ldots, R_{m+r}^\star)\\
& =\\
& \mathsf{ltp}_{m,d,t}^{(m+d,r)}
(E_{\mathcal{A}}^{(s,h,p,b,l)}(G\setminus S_0,\bar{a}),R_1^\star\setminus S_1,\ldots,R_{m+r}^\star\setminus S_{m+r}),
\end{eqnarray*}
where for every $i\in[m+r]$, $R_i^\star = R_i \cap V(E_{\mathcal{A}}^{(s,h,p,b,l)}(G,\bar{a}))$, it holds that
\[
\mathsf{ltp}_{m,d,t}^{(m+d,r)}(\mathfrak{G},R_1,\ldots,R_{m+r})
=
\mathsf{ltp}_{m,d,t}^{(m+d,r)}(\mathfrak{G}\setminus S_0,R_1\setminus S_1,\ldots,R_{m+r}\setminus S_{m+r}).\]
\end{lemma}

\subsection{Item (c) of the Local-Global-Irrelevancy Condition}\label{subsec:3c}
In this subsection we prove item (b) of the~\textsl{Local-Global-Irrelevancy Condition} of the reduction
between
$(\mathsf{CMSO/tw}\!+\!\mathsf{dp},\mathsf{hw})$ and $(\mathsf{CMSO},\mathsf{tw})$ (see~\Cref{def_red}),
i.e., the existence of a tuple $(S_0,\ldots, S_{m+r})$ where $S_0\neq\emptyset$ for~\Cref{lem_equirep}; see~\Cref{lemma:finalirr}.
We start by stating the following result, which can be obtained from~\cite[Lemma 5.7]{ThilikosW23excl}; see~also~\cite[Theorem~1.1]{ChuzhoyT21towa}.

\begin{proposition}\label{prop:union}
There is a function $\newfun{fun:uniontw}:\mathbb{N}^2\to\mathbb{N}$ such that
for every $r,t\in\mathbb{N}$,
if $G$ is a graph and $X_1,\ldots,X_r \subseteq V(G)$
such that, for $i\in [r]$, $\tw(G,X_i) \leq t$, then $\tw(G,X_1\ldots \cup X_r)\leq \funref{fun:uniontw}(r,t)$.
\end{proposition}

We first show how to find sets $S_1,\ldots,S_{m+r}$ that are irrelevant for the annotated sets.

\begin{lemma}\label{lemma:reducing_annotation}
There is a function $\newfun{fun:annotation}:\mathbb{N}^5\to\mathbb{N}$ such that for
every $l,m,d,r,t,z\in\mathbb{N}$, where
$r\leq d$, if
\begin{itemize}
\item $s:=\funref{fun:bid}(t)$,
$h:=m+d$,
$p :=\funref{@aristocracias}(d)$,
$b:=\funref{fun_ulinkastar}(d)$,
\item $G$ is a graph,
 $R_1,\ldots, R_{m+r}\subseteq V(G)$,
$\bar{a}\in V(G)^l$,
\item $(\mathcal{A},\mathfrak{R})$ is a well-aligned $(x,x)$-railed annulus flatness pair of $G\setminus V(\bar{a})$,
for $x\geq\funref{fun:annotation}(l,m,d,r,t,z)$, and
\item $Z:= V(\cupall\mathsf{Influence}_{\mathfrak{R}}(C_p))\cup\mathsf{inn}_{\mathfrak{R}}(\mathcal{A})$, where $C_p$ is the $p$-th cycle of $\mathcal{A}$,
\end{itemize}
then there is a $(z,z)$-railed annulus flatness pair $(\mathcal{A}',\mathfrak{R}')$ of $G\setminus V(\bar{a})$
and
sets $S_1,\ldots, S_{m+r}\subseteq Z$,
such that
\begin{itemize}
\item $V(\mathsf{Compass}_{\mathfrak{R}'}(\mathcal{A}'))\cup\mathsf{inn}_{\mathfrak{R}'}(\mathcal{A}')$ is a subset of
$Z$ and is disjoint from $\bigcup_{i\in[m+r]}R_i\setminus S_i$, and

\item if $R^\star_i = R_i\cap V(E_{\mathcal{A}}^{(s,h,p,b,l)}(G,\bar{a}))$, then
\begin{eqnarray*}
& \mathsf{ltp}_{m,d,t}^{(m+d,r)}
(E_{\mathcal{A}}^{(s,h,p,b,l)}(G,\bar{a}),R_1^\star,\ldots, R_{m+r}^\star)\\
& = \\
& \mathsf{ltp}_{m,d,t}^{(m+d,r)}
(E_{\mathcal{A}}^{(s,h,p,b,l)}(G,\bar{a}),R^\star_1\setminus S_1,\ldots, R^\star_{m+r}\setminus S_{m+r}).
\end{eqnarray*}
\end{itemize}
\end{lemma}

\begin{proof}
We first set $q:=5/2\cdot b$.
Then, for every $i\in[m+r]$,
we set
\begin{eqnarray*}
\alpha_i  & : = & \funref{fun:uniontw}(t,g_i),\text{ where }g_i := \funref{fun:partial-patterns}(m+d-i,l,m,d,s,t),\\
p_i & := & s^h\cdot p + (\funref{fun:bid}(\alpha_i) \cdot h_i+1),\\
h_i & := & 2p_{i+1}+2, \text{ where $p_{m+r+1} = \max\{z,q\}$},\text{ and}\\
\funref{fun:annotation}(l,m,d,t,z) & := & p + p_1
\end{eqnarray*}
We set $\mathfrak{G}^{\langle 1\rangle}:=E_{\mathcal{A}}^{(s,h,p,b,l)}(G,\bar{a})$ and
for every $i\in[m+r]$, $R^{\langle 1 \rangle}_i = R^\star_i$, and $\bar{R}^{\langle 1\rangle}:=\bar{R}^\star$.
Our goal is to prove the following claim inductively.
Observe that the lemma follows for $i=m+r+1$.
\begin{claim}\label{cla}
Let $i\in[m+r+1]$.
There exist
$S_1,\ldots,S_{i-1}\subseteq V(\mathfrak{G}^{\langle 1\rangle})$  and a well-aligned $(p_i,p_i)$-railed annulus flatness pair
$(\mathcal{A}_i,\mathfrak{R}_i)$ of $G\setminus V(\bar{a})$
such that
\begin{itemize}
\item[(a)] $V(\mathsf{Compass}_{\mathfrak{R}_i}(\mathcal{A}_i))\cup \mathsf{inn}_{\mathfrak{R}_i}(\mathcal{A}_i)$ is a subset of $Z$ and disjoint from $\bigcup_{j\in[i-1]}R_j\setminus S_j$ and
\item[(b)] $\mathsf{ltp}_{m,d,t}^{(m+d,r)}
(\mathfrak{G}^{\langle 1 \rangle},\bar{R}^{\langle 1 \rangle})
 =   \mathsf{ltp}_{m,d,t}^{(m+d,r)}
(\mathfrak{G}^{\langle 1\rangle},R^{\langle 1\rangle}_1\setminus S_1,\ldots,R^{\langle 1 \rangle}_{i-1}\setminus S_{i-1},R^{\langle 1 \rangle}_{i},\ldots,R^{\langle 1 \rangle}_{m+r})$.
\end{itemize}
\end{claim}
\medskip

\noindent{\emph{Proof of~\Cref{cla}.}}
For the base case of the induction, i.e., $i=1$, we can set $(\mathcal{A}_1,\mathfrak{R}_1)$ to be a $(p_1,p_1)$-railed annulus flatness pair of $G\setminus V(\bar{a})$,
where $V(\mathsf{Compass}_{\mathfrak{R}_1} (\mathcal{A}_1))\cup \mathsf{inn}_{\mathfrak{R}_1}(\mathfrak{A}_1) \subseteq Z$.
Such $(\mathcal{A}_1,\mathfrak{R}_1)$ exists because $\mathcal{A}$ has
$s^h\cdot p + p_1$ cycles.
Note that in this case, (a) and (b) hold trivially.

Assume now that the claim holds for $i$. We prove that it also holds for $i+1$, i.e., we find an $S_{i}$ and an $(\mathcal{A}_{i+1},\mathfrak{R}_{i+1})$ that also satisfy (a) and (b).

Using $(\mathcal{A}_i,\mathfrak{R}_i)$, we consider $E_{\mathcal{A}_i}^{(s,h,p,b,l)}(G,\bar{a})$, which we denote by $\mathfrak{G}^{\langle i\rangle}$ for brevity.
Also, for every $j\in[i-1]$,
we set $R^{\langle i \rangle}_j= (R_j\setminus S_j) \cap V(\mathfrak{G}^{\langle i\rangle})$, for every $j\in[i,m+r]$, $R^{\langle i \rangle}_j := R_j \cap V(\mathfrak{G}^{\langle i \rangle})$,
and $\bar{R}^{\langle i\rangle}:=(R^{\langle i \rangle}_1,\ldots, R^{\langle i \rangle}_{m+r})$.
Note that, by (a), for every $j\in[i-1]$,
$R^{\langle i \rangle}_j = \emptyset$.
We fix $\bar{w}:=w_1 \ldots w_{i-1}$, where for each $j\in[i-1]$,
$w_j = s-1$.
It is crucial to observe that for every $j\in[i-1]$,
\begin{equation}
\label{eq_void}
\mathsf{ltp}_{m,d,t}^{(m+d-j+1,r)}(\mathfrak{G}^{\langle i \rangle},\bar{R}^{\langle i \rangle},w_1\ldots w_{j-1},\emptyset^{j-1}) = \{\mathsf{ltp}_{m,d,t}^{(m+d-j,r)}(\mathfrak{G}^{\langle i \rangle},\bar{R}^{\langle i \rangle},w_1\ldots w_{j-1}w_j,\emptyset^{j-1}\emptyset)\}.
\end{equation}
By~\Cref{lemma_number_partial_patterns}, the size of the set $\{ \mathsf{ltp}_{m,d,t}^{(m+d-i,r)}(\mathfrak{G}^{\langle i \rangle},\bar{R}^{\langle i \rangle},\bar{w}w_2,\emptyset^{i-1} X_2)\mid w_i\in [0,s-1],\ X_i\subseteq R_i^{\langle i \rangle}\}$ is upper-bounded by $\funref{fun:partial-patterns}(m+d-i,l,m,d,s,t) = g_i$.
For every $w_i\in [0,s-1]$,
we set $\mathcal{Q}^{w_i}$ to be the minimum size collection of subsets of $R_i^{\langle i \rangle}$ such that if
there is a set $X\subseteq R_i^{\langle i \rangle}$ satisfying
\begin{equation}
\tw(G,X)\leq t \text{ and } \mathsf{ltp}_{m,d,t}^{(m+d-i,r)}(\mathfrak{G}^{\langle i \rangle},\bar{R}^{\langle i \rangle},\bar{w}w_i,\emptyset^i X)\in  \mathsf{ltp}_{m,d,t}^{(m+d-i+1,r)}(\mathfrak{G}^{\langle i \rangle},\bar{R}^{\langle i \rangle},\bar{w},\emptyset^i),\label{eq_representi}
\end{equation}
then there is also a set $X'\in \mathcal{Q}^{w_i}$ satisfying~\eqref{eq_representi}.
We set $\mathcal{Q}_i = \bigcup_{w_i \in [0,\funref{fun:bid}(t)-1]} \mathcal{Q}^{w_i}$ and observe that $|\mathcal{Q}_i|\leq g_i$.
We also set $S_i := R_i \setminus \bigcup_{Q\in \mathcal{Q}_i} Q$ and
$\bar{R}^{\langle i \rangle}_{\sf new}$ to be the tuple obtained from $\bar{R}^{\langle i \rangle}$ by replacing $R^{\langle i \rangle}_i$ with $R^{\langle i \rangle}_i\setminus S_i$.
By definition of $\mathcal{Q}_i$ and~\eqref{eq_void},
we have that
\begin{equation}
\mathsf{ltp}_{m,d,t}^{(m+d,r)}
(\mathfrak{G}^{\langle i \rangle},\bar{R}^{\langle i\rangle})
 =   \mathsf{ltp}_{m,d,t}^{(m+d,r)}
(\mathfrak{G}^{\langle i\rangle},\bar{R}^{\langle i \rangle}_{\sf new}).\label{eq_part}
\end{equation}
Using~\Cref{lemmabis} and the inductive hypothesis that (b) holds for $i$,~\eqref{eq_part} implies that
\[\mathsf{ltp}_{m,d,t}^{(m+d,r)}
(\mathfrak{G}^{\langle 1\rangle},\bar{R}^{\langle 1 \rangle})
 =   \mathsf{ltp}_{m,d,t}^{(m+d,r)}
(\mathfrak{G}^{\langle 1\rangle},R^{\langle 1\rangle}_1\setminus S_1,\ldots,R^{\langle 1\rangle}_i\setminus S_i, R^{\langle 1\rangle}_{i+1},\ldots, R^{\langle 1\rangle}_{m+r}).
\]
This shows (b) for $i+1$.
Let us now show (a) for $i+1$.
Since $|\mathcal{Q}_i|\leq g_i$,
\Cref{prop:union} implies that
$\tw(G,\bigcup_{Q\in \mathcal{Q}_i} Q) \leq \funref{fun:uniontw}(t,g_i) = \alpha_i$.
Let $(\hat{\mathcal{A}}_i,\hat{\mathfrak{R}}_i)$ be a $(p_i-(s^h\cdot p),p_i-(s^h\cdot p))$-railed annulus flatness pair of $G\setminus V(\bar{a})$,
where $V(\mathsf{Compass}_{\hat{\mathfrak{R}}_i} (\hat{\mathcal{A}}_i))\cup \mathsf{inn}_{\hat{\mathfrak{R}}_i}(\hat{\mathfrak{A}}_i) \subseteq V(\cupall\mathsf{Influence}_{\mathfrak{R}}(C_{s^h\cdot p}^{i}))$ and $C_{s^h\cdot p}^i$ is the $(s^h\cdot p)$-th cycle of $\mathcal{A}_i$.
Such $(\hat{\mathcal{A}}_i,\hat{\mathfrak{R}}_i)$ exists because $\mathcal{A}_i$ has
$p_i$ cycles.
Since $p_i - (s^h\cdot p)  = \funref{fun:bid}(\alpha_i) \cdot h_i+1$,~\Cref{lemma:buffer} implies that there is a subsequence of cycles $\mathcal{C}_i '$ of $\hat{\mathcal{A}}_i$ of size $h_i$ such that $\bigcup_{Q\in \mathcal{Q}_i} Q$ does not intersect $V(\cupall\mathsf{Influence}_{\hat{\mathfrak{R}}_i}(\mathcal{C}_i))$.
Since $h_i = 2p_{i+1}+2$,
we may consider a $(p_{i+1},p_{i+1})$-railed annulus flatness pair $(\mathcal{A}_{i+1},\mathfrak{R}_{i+1})$ of $G\setminus V(\bar{a})$
such that $V(\mathsf{Compass}_{\mathfrak{R}_{i+1}} (\mathcal{A}_{i+1}))\cup \mathsf{inn}_{\mathfrak{R}_{i+1}}(\mathfrak{A}_{i+1})$ is a subset of $V(\cupall\mathsf{Influence}_\mathfrak{R}(\mathcal{C}_i'))$.
The latter implies that $\bigcup_{j\in[i]} R_i\setminus S_i$ and $V(\mathsf{Compass}_{\mathfrak{R}_{i+1}} (\mathcal{A}_{i+1}))\cup \mathsf{inn}_{\mathfrak{R}_{i+1}}(\mathfrak{A}_{i+1})$ are disjoint, showing (a) for $i+1$. This concludes the proof of the claim and of the lemma.
\end{proof}

To show item (c) of the \textsl{Local-Global-Irrelevancy Definition},
it remains to prove the existence of a non-empty set $S_0$ that can be discarded from our annotated structures (see~\Cref{def_red}). The removal of such a set $S_0$ should
not affect the satisfaction of any atomic formula (or its negation) of $\mathsf{CMSO/tw}\!+\mathsf{dp}$. For all atomic formulas except of the disjoint-paths predicate, this is trivially true.
For the disjoint-paths predicate, we use the following variant of the \textsl{Unique Linkage Theorem} that can be derived from~\cite[Theorem~23]{BasteST19hittIV} (see also~\cite[Proposition~3]{GolovachST22model_arXiv} and~\cite{RobertsonS09XXI,KawarabayashiW10asho}).

\begin{proposition}\label{lemma:irrelevant_area}
There is a function $\newfun{fun:irrelevant}:\mathbb{N}^2\to\mathbb{N}$ such that for every $l,d,y\in\mathbb{N}$, if $G$ is a graph, $(\mathcal{A},\mathfrak{R})$ is a railed annulus flatness pair of $G$, such that $C_1,\ldots, C_p$ are the cycles of $\mathcal{A}$, where $p\geq \funref{fun:irrelevant}(d,y)$, and $v_1,u_1,\ldots, v_d,u_d$ are vertices of $G$ that do not belong in $V(\mathsf{Compass}_{\mathfrak{R}}(\mathcal{A}))\cup\mathsf{inn}_{\mathfrak{R}}(\mathcal{A})$,
then for every linkage $L$ of $G$ such that $T(L)\subseteq \{v_1,u_1,\ldots,v_d,u_d\}$, there is an equivalent linkage $L'$ of $G$ that does not intersect $V(\cupall \mathsf{Influence}_{\mathfrak{R}}(C_y))\cup\mathsf{inn}_{\mathfrak{R}}(\mathcal{A})$.\end{proposition}

Using~\Cref{lemma:reducing_annotation} and~\Cref{lemma:irrelevant_area},
we can show the following.
Recall that $\Psi_{m,d,s,t,b}^\sigma$ is the set of formulas from~\Cref{lemma_number_partial_patterns} (of vocabulary $\sigma_{m,r,t,l}\cup\{R_1,\ldots, R_{m+r}\}
$); see also~\Cref{lem_equirep}.

\begin{lemma}\label{lemma:finalirr}
There is a function $\newfun{fun:item3red}:\mathbb{N}^6\to\mathbb{N}$ such that for
every $l,m,d,r,t,y\in\mathbb{N}$, where
$r\leq d$, if
\begin{itemize}
\item $s:=\funref{fun:bid}(t)$,
$h:=m+d$,
$p :=\funref{@aristocracias}(d)$,
$b:=\funref{fun_ulinkastar}(d)$,
\item $G$ is a graph,
 $R_1,\ldots, R_{m+r}\subseteq V(G)$,
$\bar{a}\in V(G)^l$,
\item $(\mathcal{A},\mathfrak{R})$ is a well-aligned $(x,x)$-railed annulus flatness pair of $G\setminus V(\bar{a})$,
for $x\geq\funref{fun:item3red}(l,m,d,r,t,y)$, and
\item $Z:= V(\cupall\mathsf{Influence}_{\mathfrak{R}}(C_{s^h\cdot p}))\cup\mathsf{inn}_{\mathfrak{R}}(\mathcal{A})$, where $C_{s^h\cdot p}$ is the $(s^h\cdot p)$-th cycle of $\mathcal{A}$,
\end{itemize}
then there exist $S_0,\ldots,S_{m+r}\subseteq Z,$ where $S_0\neq\emptyset$,
 such that
 \[\text{$(S_0,\ldots, S_{m+r})$ is $\Psi_{m,r,t,l}$-irrelevant for $(E_{\mathcal{A}}^{(s,h,p,b,l)}(G,\bar{a}),\bar{R}^{\star}),$}\]
 where $\bar{R}^\star:= (R_1\cap V(E_{\mathcal{A}}^{(s,h,p,b,l)}(G,\bar{a})),\ldots, R_{m+r} \cap V(E_{\mathcal{A}}^{(s,h,p,b,l)}(G,\bar{a})))$ and $\Psi_{m,r,t,l}:=\Psi_{m,d,s,t,b}^{\{E\}^{\langle l \rangle}}$.
 \end{lemma}

\section{Conclusions and directions for further research}
\label{sec_conclusions}

In this paper we proved that model checking on non-trivial
minor-closed graph classes is tractable for problems expressible
in the newly introduced logic $\mathsf{CMSO}/\tw+{\sf dp}$ (\Cref{main_theorem_intro}). We wish to stress
that our proof is constructive, in the sense that one can construct a  (meta-)algorithm that receives as an input a formula $φ\in \mathsf{CMSO}/\tw+{\sf dp}$ and an integer $h$, and outputs an algorithm $\mathbb{A}_{φ,h}$
that, given an $K_{h}$-minor-free graph $G$, outputs whether
$G\models φ$ or not. \Cref{main_theorem_intro} holds also for structures where the combinatorial restriction is that the Gaifman graph
excludes $H$ as a minor. In our presentation we adopted the graph-theoretic setting, as the proof for the ``more general'' version on structures is
not essentially different but notationally much heavier.

\paragraph{On the choice of $\tw$ in the parametrization of the quantifiers.}
Note that, in the definition of $\mathsf{CMSO}/\tw$, based on the parameterized quantification introduced in \eqref{quant}, annotated treewidth plays a pivotal role.  The annotated version of every
minor-monotone graph parameter $\p$ (such as treewidth, as it is our case) was defined in~\cite{ThilikosW23excl} by setting $\p(G,X):=\max\{\p(H)\mid \mbox{$H$ is an $X$-minor of $G$}\}$. Likewise, one may
parameterize the quantification in \textsf{CMSO} formulas by
setting
\[\exists_{\p≤k} X\ φ \coloneqq \exists X\ (\p(G,X)≤k \wedge φ),\]
hence defining the fragment  $\mathsf{CMSO}/\p$ of $\mathsf{CMSO}$.
Certainly, if $\p$ is any parameter that is functionally larger than $\tw$, then $\mathsf{CMSO}/\tw$  is more expressive than $\mathsf{CMSO}/\p$ (recall that a parameter $\p$ is \emph{functionally larger} than a parameter $\p'$ if there is a function $f:\N\to\N$ such that for every graph
$G$, $\p'(G)≤f(\p(G))$). But what if we consider $\mathsf{CMSO}/\p$
for some parameter $\p$ that is {\sl not} functionally larger than $\tw$? Such a parameter might be, for instance, the Euler genus, the apex\footnote{The \emph{apex number} of a graph $G$ is the minimum number of vertices whose removal yields a planar graph.} number, or the Hadwiger number. In such a case, the parameter $\p$ would be bounded for all planar graphs, which, in turn,
would imply that, on planar graphs, $\mathsf{CMSO}/\p$ would be as expressive as $\mathsf{CMSO}$. This, combined with the impossibility  results of \cite{Kreutzer12,KreutzerT10lowe,MakowskyM03tree} (see also~\cite[Section 9]{GroheK09})
implies that no AMT as the one of \Cref{main_theorem_intro} can be expected for $\mathsf{CMSO}/\p$
when $\p$ is not functionally larger than treewidth. In this sense,
the choice of treewidth in the definition of $\mathsf{CMSO}/\tw$ is the \textsl{best possible}
for an AMT as the one in \Cref{main_theorem_intro}. This justifies the use of $\exists_{k} X\ φ$ as a simplified way to write $\exists_{\tw≤k} X\ φ$ in \eqref{quant}, and indicates that  \Cref{main_theorem_intro} is ``optimal'' in what concerns the choice of the fragment $\mathsf{CMSO}/\tw$.

The introduction of $\mathsf{CMSO}/\p$ for every minor-monotone
parameter introduces a ``fine-grained'' viewpoint of \textsf{CMSO}
whose algorithmic properties, in our opinion, deserve
independent research. Is it possible to generate analogous hierarchies
for other partial ordering relations and other fragments of \textsf{CMSO}, such as \textsf{CMSO}$_{1}$ where quantification is allowed only on vertex sets?

\paragraph{Further than minors.}
The next question is whether and when we may go further than the
combinatorial threshold of minor-closedness. The next horizon for this might be topological minor-closedness, where we know that
{\sf FO+conn} and the more general  {\sf FO+dp} admit tractable model checking because of the results in \cite{PilipczukSSTV22algor} and \cite{SchirrmacherSSTV24mode}, respectively.
We are not in position to conjecture whether it is possible to replace, in \Cref{main_theorem_intro}, $H$-minor-freeness by
$H$-topological-minor-freeness.
However, we cannot exclude that this might be possible for $\mathsf{CMSO}/\p$ for some suitable choice of $\p$ such as ${\sf vc}$ (vertex cover) or ${\sf td}$ (treedepth).
Note that for all logics that subsume {\sf FO+conn}, such as {\sf FO+dp}, ${\sf \tilde{\Theta}}^{\sf dp}$, or {\sf CMSO/tw+dp} (cf.~\Cref{fig-Hasse}), $H$-topological-minor-freeness is the most general combinatorial condition we may hope for if we restrict to monotone graph classes (assuming a mild condition on the efficiency
of the encoding), as proved in~\cite{PilipczukSSTV22algor}.

\paragraph{Extending the predicates.}
Another direction is to consider variants or extensions of $\mathsf{CMSO}/\tw$ that use more expressive predicates than the disjoint paths. In the end of \Cref{def_twMSO} we introduced the
logic $\mathsf{CMSO/tw}\!+\!\mathsf{dp}^+$ where we additionally demand the disjoint paths to avoid some vertex set. Our proofs can be easily adapted in order to work for this more general logic as follows. First we note that after replacing the use of the disjoint-paths predicate by $\mathsf{dp}^+$ in the notion of annotated types, semi-annotated types, and layered types, all results from~\Cref{sec_gametrees} and~\Cref{sec_enhancedstruct} are directly obtained for $\mathsf{CMSO/tw}\!+\!\mathsf{dp}^+$. Moreover, in what concerns~\Cref{lem_colomodelsrerout} from~\Cref{subsec_conventions}, let us stress that because of the Linkage Combing Lemma~\cite[Corollary 3.4]{GolovachST23comb} the obtained linkages $\tilde{L}_1, \tilde{L}_2$ can be assumed to have the additional property that $(\tilde{L}_1 \cup \tilde{L}_2)\setminus V({\sf Leveling}_{(\mathcal{A}, \mathfrak{R})}^{\langle \mathcal{C}_{\bar{w}}\rangle}(G))\subseteq L\setminus V({\sf Leveling}_{(\mathcal{A}, \mathfrak{R})}^{\langle \mathcal{C}_{\bar{w}}\rangle}(G))$. Using this additional assumption, the core result of the proof of the \textsl{Local-Global-Irrelevancy Condition} in~\Cref{sec_exchangability}, namely~\Cref{lemma:intermediate1}, can be proved for $\mathsf{CMSO/tw}\!+\!\mathsf{dp}^+$. The only modification comes in the last part of the proof, where the additional assumption is used to guarantee that if the linkage $L$ is disjoint from the sets in $\bar{V}^{\mathsf{in}}$ then the same holds for the linkage $\tilde{L}_1 \cup \tilde{L}_2$.

Another direction is to consider parity conditions in the paths of ${\sf dp}_{k} (s_1,t_1,\ldots, s_k,t_k)$. We believe that our technique can be applied in this direction as well.

\paragraph{Irrelevant vertices in sub-quadratic time.}
An insisting open question is whether (or when) it is possible to improve the quadratic running time of the algorithms implied by \Cref{main_theorem_intro}. Notice that this running time
is quite linked to the idea of the irrelevant vertex technique, given that the irrelevant parts
are removed one by one and the cost their detection is already linear.
An improvement of the running time in \Cref{main_theorem_intro}
would require either the detection of ``many'' irrelevant vertices at once
(this is been achieved so far for particular modification problems such as  in~\cite{Kawarabayashi09}) or
would be based on new ideas beyond the irrelevant vertex technique.

\paragraph{Reductions between AMTs.}
In~\Cref{reduct_sec}  we defined a concept of reduction between AMTs  based on the notion of \textsl{vertex irrelevancy} and the fact that
``local irrelevancy'', certified by some AMT, may imply ``global irrelevancy'', yielding AMTs on more general combinatorial structures.
We defined this reducibility notion in a most  abstract way as we believe that it might be a good base  for the study, the understanding,  and the proof of other AMTs in the future.


\end{document}